\pdfoutput=1
\documentclass{amsart}
\usepackage{amsmath}
\usepackage{amssymb}
\usepackage{amsthm}
\usepackage{hyperref}
\usepackage{multicol}
\usepackage{upgreek}

\theoremstyle{definition}
\newtheorem{definition}{Definition}[section]

\theoremstyle{remark}
\newtheorem{remark}[definition]{Remark}

\theoremstyle{plain}
\newtheorem{proposition}[definition]{Proposition}
\newtheorem{theorem}[definition]{Theorem}

\usepackage{tikz}
\usepackage{tikz-3dplot}
\usepackage{braids}
\usepackage{relsize}

\usetikzlibrary{arrows}

\tikzset{
  baseline=(current bounding box.center),
  x=28pt, y=28pt, font=\small,
  >=stealth,
  vertex/.style={inner sep=0pt},
  bvertex/.style={vertex, draw, shape=circle, fill=white, opacity=0.7, minimum size=4pt},
  lvertex/.style={vertex, draw, fill= white, opacity=1, minimum size=9pt, font=\scriptsize},
  gnode/.style={lvertex, shape=circle},
  fnode/.style={lvertex, shape=rectangle},
  q->/.style={->, shorten >=0.5pt, font=\smaller[3]},
}

\newcommand{\CN}{\mathcal{N}}
\newcommand{\cD}{\mathcal{D}}

\newcommand{\CX}{\mathcal{X}}
\newcommand{\at}{\tilde{a}}

\newcommand{\Xt}{\widetilde{X}}

\newcommand{\xt}{\tilde{x}}

\newcommand{\CA}{\mathcal{A}}
\newcommand{\CH}{\mathcal{H}}
\newcommand{\CS}{\mathcal{S}}
\newcommand{\Z}{\mathbb{Z}}
\newcommand{\Q}{\mathbb{Q}}
\newcommand{\R}{\mathbb{R}}
\newcommand{\C}{\mathbb{C}}

\newcommand{\Ad}{\mathrm{Ad}}

\newcommand{\Xtrop}{\mathsf{X}}
\newcommand{\Ptrop}{\mathbb{P}_{\mathrm{trop}}}
\newcommand{\pitrop}{\pi_{\mathrm{trop}}}
\newcommand{\sgntrop}{\upepsilon}
\newcommand{\Puniv}{\mathbb{P}_{\mathrm{univ}}}

\newcommand{\id}{\mathop{\mathrm{id}}\nolimits}
\newcommand{\rmd}{\mathrm{d}}
\newcommand{\sgn}{\mathrm{sgn}}

\newcommand{\seed}{\Sigma}
\newcommand{\auto}{\alpha}
\newcommand{\ct}{\mathbf{c}}
\newcommand{\K}[1]{\mathbf{K}_{#1}}

\newcommand{\xop}{\hat{x}}
\newcommand{\bop}{\hat{b}}
\newcommand{\xtop}{\hat{\tilde{x}}}

\newcommand{\Xop}{\widehat{X}}
\newcommand{\Bop}{\widehat{B}}
\newcommand{\Xtop}{\widehat{\widetilde{X}}}

\newcommand{\iu}{\mathrm{i}}
\renewcommand{\Re}{\mathop\mathrm{Re}}
\renewcommand{\Im}{\mathop\mathrm{Im}}
\newcommand{\Tr}{\mathop\mathrm{Tr}}

\newcommand{\ah}{\hat{a}}
\newcommand{\ph}{\hat{p}}

\newcommand{\ket}[1]{|#1\rangle}
\newcommand{\bra}[1]{\langle #1|}
\newcommand{\braket}[2]{\langle #1|#2\rangle}

\newcommand{\U}{\mathrm{U}}

\renewcommand{\flat}{\prime}

\title[Cluster transformations, the tetrahedron equation
and\dots]{Cluster transformations, the tetrahedron equation and
  three-dimensional gauge theories}

\author{Xiaoyue Sun}
\address{Yau Mathematical Sciences Center and Department of Mathematical Sciences, Tsinghua University, China}

\author{Junya Yagi}
\address{Yau Mathematical Sciences Center, Tsinghua University, China}

\date{}

\begin{document}
\begin{abstract}
  We define three families of quivers in which the braid relations of
  the symmetric group $S_n$ are realized by mutations and
  automorphisms.  A sequence of eight braid moves on a reduced word
  for the longest element of $S_4$ yields three trivial cluster
  transformations with 8, 32 and 32 mutations.  For each of these
  cluster transformations, a unitary operator representing a single
  braid move in a quantum mechanical system solves the tetrahedron
  equation.  The solutions thus obtained are constructed from the
  noncompact quantum dilogarithm and can be identified with the
  partition functions of three-dimensional $\CN = 2$ supersymmetric
  gauge theories on a squashed three-sphere.
\end{abstract}

\maketitle

\section{Introduction}
\label{sec:introduction}

The Zamolodchikov tetrahedron equation \cite{MR611994b} is a
fundamental relation for integrability of quantum field theories in
$2+1$ spacetime dimensions and of statistical mechanical models on
three-dimensional lattices, much in the same way as its
lower-dimensional analog, the Yang--Baxter equation, is a fundamental
relation in integrable $(1+1)$-dimensional quantum field theories and
two-dimensional lattice models.  Compared to the Yang--Baxter
equation, however, our understanding of the tetrahedron equation is
still limited despite its obvious importance and relatively long
history.

In this work we hope to shed some light on the tetrahedron equation by
uncovering its connections to quantum cluster algebras and
three-dimensional supersymmetric gauge theories.

\subsection{The Yang--Baxter equation and the tetrahedron equation}

Graphically, the Yang--Baxter equation is represented as an equality
between two configurations of three intersecting lines in a plane.
The tetrahedron equation is likewise represented as an equality
between two configurations of four intersecting planes in a
three-dimensional space.  See Figure~\ref{fig:YBE-TE}.
Combinatorially, the Yang--Baxter equation and the tetrahedron
equation can be understood in terms of the reduced expressions for the
longest elements of the symmetric groups $S_3$ and $S_4$,
respectively.  Basic notions regarding the symmetric groups are
recalled in section~\ref{sec:wiring}.

\begin{figure}
  \centering
  \tikzset{font=\scriptsize}
  \tdplotsetmaincoords{120}{0}
  \begin{tikzpicture}[tdplot_main_coords, scale=0.5, baseline=14pt]
    \draw[thick, -stealth, shorten >=-12pt, shorten <=-12pt]
    (0,0,{sqrt(3)}) -- ({sqrt(3)},-1,0);

    \draw[thick, stealth-, shorten >=-12pt, shorten <=-12pt]
    (0,0,{sqrt(3)}) -- ({-sqrt(3)},-1,0);

    \draw[line width=4pt, -, white]
    ({-sqrt(3)/2},-1,0) -- ({sqrt(3)/2},-1,0);
    \draw[thick, -stealth, shorten >=-12pt, shorten <=-12pt]
    ({-sqrt(3)},-1,0) -- node[below left, shift={(-8pt,2pt)}] {} ({sqrt(3)},-1,0);

    \node[above=3pt] at (0,0,{sqrt(3)}) {$R_{13}$};
    \node[above=-1pt, xshift=-4pt] at ({-sqrt(3)},-1,0) {$R_{12}$};
    \node[above=-1pt, xshift=4pt] at ({sqrt(3)},-1,0) {$R_{23}$};
  \end{tikzpicture}
  \hspace{-0.1em}=\hspace{-0.2em}
  \tdplotsetmaincoords{-60}{0}
  \begin{tikzpicture}[tdplot_main_coords, scale=0.5, baseline=-14.4pt]
    \draw[thick, -stealth, shorten >=-12pt, shorten <=-12pt]
    ({-sqrt(3)},-1,0) -- ({sqrt(3)},-1,0);

    \draw[thick, -stealth, shorten >=-12pt, shorten <=-12pt]
    (0,0,{sqrt(3)}) -- ({sqrt(3)},-1,0);

    \draw[thick, stealth-, shorten >=-12pt, shorten <=-12pt]
    (0,0,{sqrt(3)}) -- ({-sqrt(3)},-1,0);

    \node[above=1pt] at (0,0,{sqrt(3)}) {$R_{13}$};
    \node[above=-1pt, xshift=4pt] at ({-sqrt(3)},-1,0) {$R_{23}$};
    \node[above=-1pt, xshift=-4pt] at ({sqrt(3)},-1,0) {$R_{12}$};
  \end{tikzpicture}
  \hfill
  \tdplotsetmaincoords{135}{0}
  \begin{tikzpicture}[tdplot_main_coords, scale=0.5, baseline=-5.5pt]
    \draw[thick, stealth-, shorten >=-12pt, shorten <=-12pt]
    ({-sqrt(3)},-1,0) -- (0,{sqrt(3)},0);

    \draw[thick, stealth-, shorten >=-12pt, shorten <=-12pt]
    ({sqrt(3)},-1,0) -- (0,{sqrt(3)},0);

    \draw[thick, -stealth, shorten >=-12pt, shorten <=-12pt]
    (0,0,{sqrt(3)}) -- ({sqrt(3)},-1,0);

    \draw[thick, stealth-, shorten >=-12pt, shorten <=-12pt]
    (0,0,{sqrt(3)}) -- ({-sqrt(3)},-1,0);

    \draw[thick, stealth-, shorten >=-12pt, shorten <=-12pt]
    (0,0,{sqrt(3)}) -- (0,{sqrt(3)},0);

    \draw[line width=4pt, -, white]
    ({-sqrt(3)/2},-1,0) -- ({sqrt(3)/2},-1,0);
    \draw[thick, -stealth, shorten >=-12pt, shorten <=-12pt]
    ({-sqrt(3)},-1,0) -- node[below left, shift={(-8pt,2pt)}] {} ({sqrt(3)},-1,0);

    \node[right=-1pt] at (0,{sqrt(3)},0) {$R_{123}$};
    \node[above=1pt, xshift=-10pt] at (0,0,{sqrt(3)}) {$R_{134}$};
    \node[below=1pt, xshift=-3pt] at ({-sqrt(3)},-1,0) {$R_{124}$};
    \node[above=1pt, xshift=-5pt] at ({sqrt(3)},-1,0) {$R_{234}$};
  \end{tikzpicture}
  \hspace{0.5em}=\hspace{-0.3em}
  \tdplotsetmaincoords{-45}{0}
  \begin{tikzpicture}[tdplot_main_coords, scale=0.5]
    \draw[thick, -stealth, shorten >=-12pt, shorten <=-12pt]
    ({-sqrt(3)},-1,0) -- ({sqrt(3)},-1,0);

    \draw[thick, -stealth, shorten >=-12pt, shorten <=-12pt]
    ({-sqrt(3)},-1,0) -- (0,{sqrt(3)},0);

    \draw[thick, -stealth, shorten >=-12pt, shorten <=-12pt]
    ({sqrt(3)},-1,0) -- (0,{sqrt(3)},0);

    \draw[thick, -stealth, shorten >=-12pt, shorten <=-12pt]
    (0,0,{sqrt(3)}) -- ({sqrt(3)},-1,0);

    \draw[thick, stealth-, shorten >=-12pt, shorten <=-12pt]
    (0,0,{sqrt(3)}) -- ({-sqrt(3)},-1,0);

    \draw[line width=4pt, -, white]
    (0,0,{sqrt(3/2)}) -- (0,{sqrt(3)/2},0);
    \draw[thick, -stealth, shorten >=-12pt, shorten <=-12pt]
    (0,0,{sqrt(3)}) -- (0,{sqrt(3)},0);

    \node[left=0pt] at (0,{sqrt(3)},0) {$R_{123}$};
    \node[below=1pt, xshift=10pt] at (0,0,{sqrt(3)}) {$R_{134}$};
    \node[below=0pt, xshift=5pt] at ({-sqrt(3)},-1,0) {$R_{234}$};
    \node[above=2pt, xshift=4pt] at ({sqrt(3)},-1,0) {$R_{124}$};
  \end{tikzpicture}

  \caption{Graphical representations of the Yang--Baxter equation
    (left) and the tetrahedron equation (right).}
  \label{fig:YBE-TE}
\end{figure}
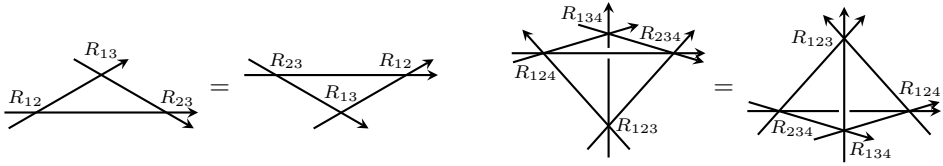

A solution of the Yang--Baxter equation is associated with adjacent
transpositions.  For concreteness let us consider a solution of vertex
type.  This is a set $R$ of linear operators
\begin{equation}
  R_{ab}\colon V_a \otimes V_b \to V_b \otimes V_a \,,
\end{equation}
where $V_a$, $a \in \{1,2,3\}$, are vector spaces and $a < b$.  (The
operator $R_{ab}$ is often denoted by $\check{R}_{ab}$ in the
literature.)  In the graphical representation, $V_a$ is the $a$th line
and $R_{ab}$ is the crossing of the $a$th and $b$th lines.
Corresponding to the two adjacent transpositions $s_1$, $s_2$ of
$S_3$, the R-matrix $R$ can act on a tensor product
$V_a \otimes V_b \otimes V_c$ of $V_1$, $V_2$, $V_3$ in two ways,
either by $R_{ab} \otimes \id_{V_c}$ or $\id_{V_a} \otimes R_{bc}$.
(The identity maps are often suppressed in the notation.)  The longest
element of $S_3$ has two reduced expressions $s_1 s_2 s_1$ and
$s_2 s_1 s_2$, and the Yang--Baxter equation
\begin{equation}
  R_{23} R_{13} R_{12} = R_{12} R_{13} R_{23}
\end{equation}
reflects the equality $s_1 s_2 s_1 = s_2 s_1 s_2$ satisfied by the two
reduced expressions as elements of $S_3$.

The tetrahedron equation lifts the above structure to one higher
level.  Its solution is associated with braid moves
\begin{equation}
  s_a s_{a+1} s_a \to s_{a+1} s_a s_{a+1} \,,
\end{equation}
as opposed to adjacent transpositions.  A vertex-type solution is a
set $R$ of linear operators
\begin{equation}
  R_{abc}\colon
  V_{ab} \otimes V_{bc} \otimes V_{ac}
  \to V_{ac} \otimes V_{bc} \otimes V_{ab} \,,
\end{equation}
where $a$, $b$, $c \in \{1,2,3,4\}$ and $a < b < c$.  Graphically, the
vector space $V_{ab}$ is represented by the intersection of the $a$th
and $b$th planes and $R_{abc}$ is represented by the intersection of
the $a$th, $b$th and $c$th planes.  The tetrahedron equation
\begin{equation}
  \label{eq:TE}
  R_{234} R_{134} R_{124} R_{123}
  =
  R_{123} R_{124} R_{134} R_{234}
\end{equation}
corresponds to an equivalence between two sequences of braid moves on
reduced expressions for the longest element of $S_4$ modulo far
commutativity, which takes
$s_1 s_2 s_3 s_1 s_2 s_1 \sim s_1 s_2 s_1 s_3 s_2 s_1$ to
$s_3 s_2 s_1 s_3 s_2 s_3 \sim s_3 s_2 s_3 s_1 s_2 s_3$.

The relation between the graphical and combinatorial interpretations
of the Yang--Baxter equation and the tetrahedron equation can be made
manifest by use of wiring diagrams.  A wiring diagram on $n$ wires is
a diagrammatic representation of a word for a permutation in $S_n$.
The wiring diagrams for the reduced words $121$ and $212$ for the
longest element of $S_3$, corresponding to the reduced expressions
$s_1 s_2 s_1$ and $s_2 s_1 s_2$, are
\begin{equation}
  121
  =
  \tikz{\braid[gap=0, rotate=90, scale=0.4] a_1 a_2 a_1;}
  \ ,
  \qquad
  212
  =
  \tikz{\braid[gap=0, rotate=90, scale=0.4] a_2 a_1 a_2;}
  \ .
\end{equation}
The equality between these wiring diagrams is topologically equivalent
to the graphical representation of the Yang--Baxter equation.  To
relate the two interpretations of the tetrahedron equation, for each
side of the equation we vertically juxtapose the relevant sequence of
wiring diagrams and think of these wiring diagrams as ``time slices''
of the surfaces swept out by four moving wires \cite{MR1402762}.  The
four surfaces correspond to the four planes in the graphical
representation, as illustrated in Figure~\ref{fig:TE}.  From the
figure we see the correspondence
\begin{equation}
  \begin{aligned}
    \text{segments of wires}
    & \leftrightarrow
    \text{regions on planes} \,,
    \\
    \text{intersections of wires}
    & \leftrightarrow
    \text{intersections of planes} \,,
    \\
    \text{regions in wiring diagrams}
    & \leftrightarrow
    \text{regions in the three-dimensional space} \,.
  \end{aligned}
\end{equation}

\begin{figure}
  \centering
  \begin{tikzpicture}
    \begin{scope}[xscale=0.5, yscale=0.15, black!30, rotate=90, thick]
      \braid[gap=0] a_1 a_2 a_3 a_1 a_2 a_1;
      \braid[gap=0, shift={(-10,0)}] a_1 a_3 a_2 a_3 a_1 a_2;
      \braid[gap=0, shift={(-5,0)}] a_1 a_2 a_3 a_2 a_1 a_2;
      \braid[gap=0, shift={(-15,0)}] a_3 a_1 a_2 a_1 a_3 a_2;
      \braid[gap=0, shift={(-20,0)}] a_3 a_2 a_1 a_2 a_3 a_2;
      \braid[gap=0, shift={(-25,0)}] a_3 a_2 a_1 a_3 a_2 a_3;

      \node[black] at (2.5,1.5) {$123121$};
      \node[black] at (0,1.5) {${\raisebox{0.5ex}{$\hspace{-0.15em}\scriptstyle \beta_{234}$}}\big\downarrow%
\phantom{\hspace{-0.15em}\scriptstyle \beta_{234}}$};
      \node[black] at (-2.5,1.5) {$123212$};
      \node[black] at (-5,1.5) {${\raisebox{0.5ex}{$\hspace{-0.15em}\scriptstyle \beta_{134}$}}\big\downarrow%
\phantom{\hspace{-0.15em}\scriptstyle \beta_{134}}$};
      \node[black] at (-7.5,1.5) {$132312$};
      \node[black] at (-10,1.5) {$\big\downarrow$};
      \node[black] at (-12.5,1.5) {$312132$};
      \node[black] at (-15,1.5) {${\raisebox{0.5ex}{$\hspace{-0.15em}\scriptstyle \beta_{124}$}}\big\downarrow%
\phantom{\hspace{-0.15em}\scriptstyle \beta_{124}}$};
      \node[black] at (-17.5,1.5) {$321232$};
      \node[black] at (-20,1.5) {${\raisebox{0.5ex}{$\hspace{-0.15em}\scriptstyle \beta_{123}$}}\big\downarrow%
\phantom{\hspace{-0.15em}\scriptstyle \beta_{123}}$};
      \node[black] at (-22.5,1.5) {$321323$};

      \draw[red, line width=2pt, rounded corners, opacity=0.8] (1.5,-0.75) -- (-8.5,-0.75) -- (-13.5,-1.75) -- (-15.5,-2.75);
      \draw[red, line width=1pt, rounded corners, opacity=0.8] (-19.5,-4.75) -- (-21.5,-5.75);

      \draw[red, line width=1.5pt, rounded corners, opacity=0.8] (2.5,-1.75) -- (-2.5,-1.75) -- (-4.5,-2.75);
      \draw[red, line width=1.5pt, rounded corners, opacity=0.8] (-19.5,-4.75) -- (-22.5,-4.75);

      \draw[red, line width=1pt, rounded corners, opacity=0.8] (3.5,-2.75) -- (-4.5,-2.75);
      \draw[red, line width=2pt, rounded corners, opacity=0.8] (-15.5,-2.75) -- (-23.5,-2.75);
      
      \draw[red, line width=2pt, rounded corners, opacity=0.8] (1.5,-3.75) -- (-0.5,-4.75);
      \draw[red, line width=1pt, rounded corners, opacity=0.8] (-19.5,-4.75) -- (-21.5,-3.75);

      \draw[red, line width=1.5pt, rounded corners, opacity=0.8] (2.5,-4.75) -- (-0.5,-4.75);
      \draw[red, line width=1.5pt, rounded corners, opacity=0.8] (-15.5,-2.75) -- (-17.5,-1.75) -- (-22.5,-1.75);

      \draw[red, line width=2pt, rounded corners, opacity=0.8] (1.5,-5.75) -- (-0.5,-4.75);
      \draw[red, line width=1pt, rounded corners, opacity=0.8] (-4.5,-2.75) -- (-6.5,-1.75) -- (-11.5,-0.75) -- (-21.5,-0.75);

      \draw[purple, line width=1.5pt, rounded corners, opacity=0.8] (-15.5,-2.75) -- (-19.5,-4.75);
      \draw[purple, line width=1pt, rounded corners, opacity=0.8] (-4.5,-2.75) -- (-6.5,-3.75) -- (-11.5,-4.75) -- (-19.5,-4.75);
      \draw[purple, line width=1.5pt, rounded corners, opacity=0.8] (-4.5,-2.75) -- (-15.5,-2.75);
      \draw[purple, line width=1.5pt, rounded corners, opacity=0.8] (-0.5,-4.75) -- (-2.5,-5.75) -- (-17.5,-5.75) -- (-19.5,-4.75);
      \draw[purple, line width=2pt, rounded corners, opacity=0.8] (-0.5,-4.75) -- (-8.5,-4.75) -- (-13.5,-3.75) -- (-15.5,-2.75);
      \draw[purple, line width=1.5pt, rounded corners, opacity=0.8] (-0.5,-4.75) -- (-2.5,-3.75) -- (-4.5,-2.75);
   \end{scope}
  \end{tikzpicture}
  \quad = \quad
  \begin{tikzpicture}
    \begin{scope}[xscale=0.5, yscale=0.15, black!30, rotate=90, thick]
      \braid[gap=0] a_1 a_2 a_1 a_3 a_2 a_1;
      \braid[gap=0, shift={(-5,0)}] a_2 a_1 a_2 a_3 a_2 a_1;
      \braid[gap=0, shift={(-10,0)}] a_2 a_1 a_3 a_2 a_3 a_1;
      \braid[gap=0, shift={(-15,0)}] a_2 a_3 a_1 a_2 a_1 a_3;
      \braid[gap=0, shift={(-20,0)}] a_2 a_3 a_2 a_1 a_2 a_3;
      \braid[gap=0, shift={(-25,0)}] a_3 a_2 a_3 a_1 a_2 a_3;

      \node[black] at (2.5,-8) {$121321$};
      \node[black] at (0,-8) {$\phantom{\hspace{-0.15em}\scriptstyle \beta_{123}}%
        \big\downarrow{\raisebox{0.5ex}{$\hspace{-0.15em}\scriptstyle \beta_{123}$}}$};
      \node[black] at (-2.5,-8) {$212321$};
      \node[black] at (-5,-8) {$\phantom{\hspace{-0.15em}\scriptstyle \beta_{124}}%
        \big\downarrow{\raisebox{0.5ex}{$\hspace{-0.15em}\scriptstyle \beta_{124}$}}$};
      \node[black] at (-7.5,-8) {$213231$};
      \node[black] at (-10,-8) {$\big\downarrow$};
      \node[black] at (-12.5,-8) {$231213$};
      \node[black] at (-15,-8) {$\phantom{\hspace{-0.15em}\scriptstyle \beta_{134}}%
        \big\downarrow{\raisebox{0.5ex}{$\hspace{-0.15em}\scriptstyle \beta_{134}$}}$};
      \node[black] at (-17.5,-8) {$232123$};
      \node[black] at (-20,-8) {$\phantom{\hspace{-0.15em}\scriptstyle \beta_{234}}%
        \big\downarrow{\raisebox{0.5ex}{$\hspace{-0.15em}\scriptstyle \beta_{234}$}}$};
      \node[black] at (-22.5,-8) {$323123$};

      \draw[red, line width=2pt, rounded corners, opacity=0.8] (1.5,-0.75) -- (-0.5,-1.75);
      \draw[red, line width=1pt, rounded corners, opacity=0.8] (-4.5,-3.75) -- (-6.5,-4.75) -- (-11.5,-5.75) -- (-16.5,-5.75) -- (-21.5,-5.75);

      \draw[red, line width=1.5pt, rounded corners, opacity=0.8] (2.5,-1.75) -- (-1,-1.75);
      \draw[red, line width=1.5pt, rounded corners, opacity=0.8] (-15.5,-3.75) -- (-17.5,-4.75) -- (-22.5,-4.75);

      \draw[red, line width=2pt, rounded corners, opacity=0.8] (1.5,-2.75) -- (-0.5,-1.75);
      \draw[red, line width=1pt, rounded corners, opacity=0.8] (-19.5,-1.75) -- (-21.5,-2.75);

      \draw[red, line width=1pt, rounded corners, opacity=0.8] (3.5,-3.75) -- (-1.5,-3.75) -- (-4.5,-3.75);
      \draw[red, line width=2pt, rounded corners, opacity=0.8] (-15.5,-3.75) -- (-18.5,-3.75) -- (-23.5,-3.75);

      \draw[red, line width=1.5pt, rounded corners, opacity=0.8] (2.5,-4.75) -- (-2.5,-4.75) -- (-4.5,-3.75);
      \draw[red, line width=1.5pt, rounded corners, opacity=0.8] (-19.5,-1.75) -- (-22.5,-1.75);

      \draw[red, line width=2pt, rounded corners, opacity=0.8] (1.5,-5.75) -- (-3.5,-5.75) -- (-8.5,-5.75) -- (-13.5,-4.75) -- (-15.5,-3.75);
      \draw[red, line width=1pt, rounded corners, opacity=0.8] (-19.5,-1.75) -- (-21.5,-0.75);

      \draw[purple, line width=1.5pt, rounded corners, opacity=0.8] (-0.5,-1.75) -- (-2.5,-2.75) -- (-4.5,-3.75);
      \draw[purple, line width=2pt, rounded corners, opacity=0.8] (-0.5,-1.75) -- (-3.5,-1.75) -- (-8.5,-1.75) -- (-13.5,-2.75) -- (-15.5,-3.75);
      \draw[purple, line width=1.5pt, rounded corners, opacity=0.8] (-0.5,-1.75) -- (-2.5,-0.75) -- (-7.5,-0.75) -- (-12.5,-0.75) -- (-17.5,-0.75) -- (-19.5,-1.75);
      \draw[purple, line width=1.5pt, rounded corners, opacity=0.8] (-4.5,-3.75) -- (-12.5,-3.75) -- (-15.5,-3.75);
      \draw[purple, line width=1pt, rounded corners, opacity=0.8]
      (-4.5,-3.75) -- (-6.5,-2.75) -- (-11.5,-1.75) -- (-16.5,-1.75) -- (-19.5,-1.75);
      \draw[purple, line width=1.5pt, rounded corners, opacity=0.8] (-15.5,-3.75) -- (-19.5,-1.75);
    \end{scope}
  \end{tikzpicture}
  \caption{The tetrahedron equation arises from two sequences of braid
    moves with the isotopic start and end wiring diagrams on four
    wires.  Each wiring diagram represents a slice of four surfaces
    that bound (after flattened to planes) a tetrahedron.  The six
    intersection curves of these surfaces are drawn in color.  Each
    braid move $\beta_{abc}$ changes the local configuration of the
    surfaces and acts downward.  The corresponding R-matrix $R_{abc}$
    acts upward.}
  \label{fig:TE}
\end{figure}
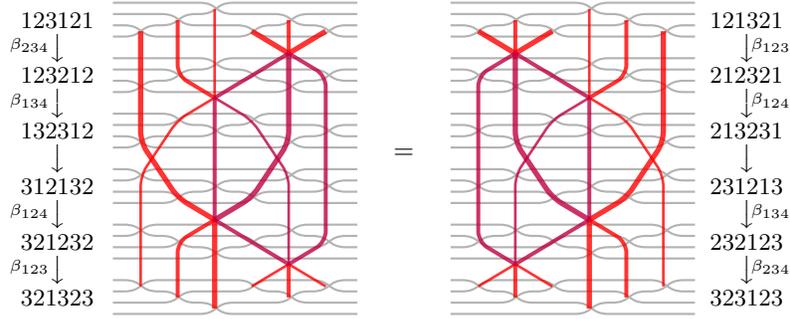

A solution of the tetrahedron equation obtained by Kapranov and
Voevodsky \cite{MR1278735} realizes the symmetric group structure with
representations of the quantized coordinate ring $A_q(A_3)$ on
$q$-oscillator Fock spaces.  This is the same solution as the one
discovered by Bazhanov and Sergeev \cite{Bazhanov:2005as}, as pointed
out in \cite{Kuniba:2012ei}, and is expected to arise from a brane
configuration in M-theory \cite{Yagi:2022tot}.

\subsection{Summary}

In this paper we will construct three solutions of the tetrahedron
equation with a help of quantum cluster transformations and explain
how these solutions arise from three-dimensional supersymmetric gauge
theories.  Let us summarize the main ideas.

To each word for a permutation in $S_n$ we assign three cluster seeds,
or equivalently three quivers, which we call the triangle, square and
butterfly quivers.  Thus we obtain three families of quivers assigned
to the words for the permutations in $S_n$.  In each of these
families, braid moves are realized as transformations of quivers
composed of mutations and automorphisms relabeling vertices
(Proposition \ref{prop:R}).  For the triangle quivers a braid move is
a single mutation followed by an automorphism, whereas for the other
two families a braid move involves four mutations.

The theory of quantum cluster varieties \cite{MR2567745, MR2470108}
tells us how to represent such quiver transformations in quantum
mechanics.  For every seed $\seed$, there is a corresponding quantum
mechanical system whose observables are generated by pairs of
variables indexed by the vertices of $\seed$ and satisfy commutation
relations determined by the way in which the vertices are connected by
arrows.  The Hilbert space of states $\CH_\seed$ of the system is the
space of wavefunctions of half of these variables.  A composition
$\ct$ of mutations and automorphisms that transforms $\seed$ to
another seed $\seed'$ induces an isomorphism
$\K{\ct}\colon \CH_{\seed'} \xrightarrow{\sim} \CH_\seed$ between the
associated quantum mechanical systems, which is constructed from the
noncompact quantum dilogarithm.  The intertwiner $\K{\ct}$ transforms
the observables by conjugation, the transformation known as the
quantum cluster transformation $\ct^q$.  A crucial property of
$\K{\ct}$ is that if $\seed' = \seed$ and $\ct^q$ is the identity map,
then $\K{\ct}$ itself is the identity map
(Propositions~\ref{prop:trivial-cq} and ~\ref{prop:trivial-K}).

The two sequences of braid moves that lead to the tetrahedron equation
can be concatenated (with one of them reversed) to form a loop of
braid moves from one reduced expression for the longest element of
$S_4$ to itself.  By construction, on each of the three quivers
assigned to this reduced expression the corresponding quiver
transformation $\ct$ acts trivially.  A key observation made in this
paper is that $\ct^q$ is also trivial
(Proposition~\ref{prop:trivial-loop-ct}).

Therefore, $\K{\ct}$ equals the identity map.  This equality can be
rewritten as the tetrahedron equation solved by the intertwiner
corresponding to a single braid move.  In this way we obtain three
solutions of the tetrahedron equation, associated with the triangle,
square and butterfly quivers.  They define local Boltzmann weights for
three-dimensional statistical mechanical lattice models with
continuous spin variables.

The triangle, square and butterfly quivers have appeared in connection
with gauge theories and the Yang--Baxter equation
\cite{Bazhanov:2011mz, Yagi:2015lha, Yamazaki:2015voa}.  In that
context, these quivers describe supersymmetric gauge theories with
four supercharges, and the Yang--Baxter equation is interpreted as an
infrared duality \cite{Yamazaki:2012cp}.  Computing physical
quantities that are invariant under these dualities, such as
supersymmetric indices, one obtains solutions of the Yang--Baxter
equation.

The solutions of the tetrahedron equation constructed in this paper
also admit gauge theory interpretations.  In section \ref{sec:QFT} we
show that for any composition $\ct$ of mutations and automorphisms
from a quiver $\seed$ to another quiver $\seed'$, the intertwiner
$\K{\ct}$ can be identified with the partition function of a
three-dimensional $\CN = 2$ supersymmetric gauge theory on a squashed
three-sphere.  A similar observation was made for a closely related
operator in \cite{Terashima:2013fg}, and we adapt their computation to
$\K{\ct}$.

\subsection{Relations to other works}

In \cite{Yamazaki:2016wnu}, Yamazaki derived an equation of the form
\begin{equation}
  L_{23}(\tau_{23}) L_{13}(\tau_{13}) L_{12}(\tau_{12})
  =
  L_{12}(\tau'_{12}) L_{13}(\tau'_{13}) L_{23}(\tau'_{23})
\end{equation}
as an equality between the sphere partition functions of
two-dimensional $\CN = (2,2)$ supersymmetric gauge theories described
by the butterfly quivers assigned to the reduced expressions
$s_1 s_2 s_1$ and $s_2 s_1 s_2$.  This equation is almost the
Yang--Baxter equation but not quite since $L_{ab}$ depends on a set of
variables $\tau_{ab}$.  In the gauge theory language, these variables
are Fayet--Iliopoulos parameters and known to transform under quiver
mutations in the same way as classical $\CX$-variables do
\cite{Benini:2014mia}.  Yamazaki's equation can be understood as a
classical limit of the tetrahedron equation of type $RLLL = LLLR$:
\begin{equation}
  \label{eq:RLLL}
  L_{234} L_{134} L_{124} R_{123}
  =
  R_{123} L_{124} L_{134} L_{234} \,,
\end{equation}
where $R_{123}$ is the R-matrix for the butterfly quiver constructed
in this paper, which transforms
quantum $\CX$-variables by
conjugation.

A classical limit of the RLLL relation \eqref{eq:RLLL} with different
$L$ has been obtained in \cite{Gavrylenko:2020eov}, also from the
butterfly quiver.  The L-operator considered in
\cite{Gavrylenko:2020eov} is a classical limit of the L-operator
constructed in \cite{Bazhanov:2005as}, which satisfies the RLLL
relation with a $q$-oscillator-valued solution of the tetrahedron
equation \cite{MR1278735, Bazhanov:2005as}.  This fact suggests that
there is a close relation between the R-matrix for the butterfly
quiver and the R-matrix of \cite{MR1278735, Bazhanov:2005as}.

It appears that the three-dimensional gauge theory corresponding to
the intertwiner $\K{\ct}\colon \CH_{\seed'} \to \CH_\seed$ is (an
infrared description of) a domain wall that separates two different
parameter configurations of a four-dimensional $\CN = 2$
supersymmetric field theory, for which the BPS spectra are encoded in
$\seed$ and $\seed'$, respectively.  At least it is akin to such
domain walls which arise from pairs of M5-branes compactified on
three-manifolds with boundary \cite{Terashima:2011qi, Dimofte:2011ju,
  Cecotti:2011iy}, and it has been argued that in the infrared on the
Coulomb branch, the relevant sector of the four-dimensional theory on
$S^3_b$ is, essentially, captured by the quantum mechanical systems
assigned to $\seed$ and $\seed'$.  There are, however, differences
between the construction of those domain walls and our construction,
namely we use twice as many cluster coordinates and twice as many
quantum dilogarithms.

\section{Cluster transformations}
\label{sec:cluster}

In this section we recall relevant definitions and facts about cluster
transformations.  We will mainly follow the conventions of
\cite{MR2567745, MR2470108}.  The main result of this section is
Proposition \ref{prop:trivial-K} which provides, for each trivial
cluster transformation, an identity satisfied by a product of
noncompact quantum dilogarithms.  This identity will play a crucial
role in the construction of solutions of the tetrahedron equation in
section~\ref{sec:TE}.

\subsection{Cluster varieties}

A \emph{cluster seed} (or simply \emph{seed}) $\seed$ is a pair
$(I,\varepsilon)$ of a finite set $I$ and a skew-symmetric integer
matrix $\varepsilon = (\varepsilon_{ij})_{i,j \in I}$, called the
\emph{exchange matrix} of the seed.  We will identify a seed with a
\emph{quiver}, a directed graph consisting of vertices connected by
arrows.  The quiver corresponding to a seed $\seed = (I, \varepsilon)$
has $|I|$ vertices, labeled by elements of $I$, and $\varepsilon_{ij}$
arrows $j \to i$ between vertices $i$, $j \in I$ if
$\varepsilon_{ij} > 0$.

A seed $\seed' = (I,\varepsilon')$ is said to be obtained from a seed
$\seed = (I,\varepsilon)$ by the \emph{mutation
  $\mu_k\colon \seed \to \seed'$ in the direction of $k \in I$} if
\begin{equation}
  \varepsilon'_{ij}
  =
  \begin{cases} 
    -\varepsilon_{ij}
    & \text{if $i = k$ or $j = k$} \,;
    \\
    \varepsilon_{ij}
    + \frac12 (|\varepsilon_{ik}| \varepsilon_{kj}
    + \varepsilon_{ik} |\varepsilon_{kj}|)
    & \text{otherwise} \,.
  \end{cases}
\end{equation}
The quiver corresponding to $\seed'$ is obtained from that
corresponding to $\seed$ by the following procedure:
\begin{enumerate}
\item For each pair of arrows $i \to k$ and $k \to j$, draw an arrow $i \to j$.
\item Reverse the directions of all arrows incident to the vertex $k$.
\item Delete pairs of arrows $i \to j$ and $j \to i$ (``$2$-cycles'').
\end{enumerate}

To a seed $\seed = (I, \varepsilon)$ we assign three algebraic tori.
Let $\C^\times$ be the multiplicative group of complex numbers.  The
\emph{seed $\CA$-torus} $\CA_\seed$ is the algebraic torus
$(\C^\times)^I$, and the standard coordinates $A := (A_i)_{i \in I}$
of $\CA_\seed$ are referred to as the \emph{cluster $\CA$-variables}
(or simply \emph{$\CA$-variables}).  The \emph{seed $\cD$-torus}
$\cD_{\seed}$ is the torus $(\C^\times)^I \times (\C^\times)^I$,
equipped with the \emph{cluster $\cD$-variables}
$(X,B) := (X_i, B_i)_{i \in I}$.  The \emph{seed $\CX$-torus}
$\CX_{\seed}$ is the torus $(\C^\times)^I$ obtained from $\cD_\seed$
by projection to the first factor and equipped with the \emph{cluster
  $\CX$-variables} $X = (X_i)_{i \in I}$.

A mutation $\mu_k\colon \seed \to \seed'$ induces birational maps
$\mu_k\colon \CA_{\seed} \to \CA_{\seed'}$,
$\mu_k\colon \CX_{\seed} \to \CX_{\seed'}$ and
$\mu_k\colon \cD_{\seed} \to \cD_{\seed'}$, all denoted by the same
symbol $\mu_k$.  On the cluster variables the pullback of $\mu_k$ acts
by
\begin{align}
  \label{eq:A'}
  \mu_k^*(A'_i)
  &:=
  \begin{cases}
    A_k^{-1} \Bigl(\prod_{j \in I \mid \varepsilon_{kj} > 0} A_j^{\varepsilon_{kj}}
      + \prod_{j \in I \mid \varepsilon_{kj} < 0} A_j^{-\varepsilon_{kj}}\Bigr)
    & \text{if $i = k$} \,;
    \\
    \smash{A_i}
    & \text{if $i \neq k$} \,,
  \end{cases}
  \\
  \label{eq:X'}
  \mu_k^*(X'_i)
  &:=
  \begin{cases}
    X_k^{-1}
    & \text{if $i = k$} \,;
    \\
    X_i (1 + X_k^{\sgn(-\varepsilon_{ik})})^{-\varepsilon_{ik}}
    & \text{if $i \neq k$} \,,
  \end{cases}
  \\
  \label{eq:B'}
  \mu_k^*(B'_i)
  &:=
  \begin{cases}
    (X_k \prod_{j \mid \varepsilon_{kj} > 0} B_j^{\varepsilon_{kj}}
    + \prod_{j \mid \varepsilon_{kj} < 0} B_j^{-\varepsilon_{kj}})
    /B_k (1 + X_k)
    & \text{if $i = k$} \,;
    \\
    B_i
    & \text{if $i \neq k$} \,.
  \end{cases}
\end{align}
Here $\sgn(a) = +1$ for $a \geq 0$ and $-1$ for $a < 0$.

Mutations are involutive: applying $\mu_k$ twice leaves a seed
invariant and $\mu_k^* \circ \mu_k^*$ is the identity map on each of
the seed tori.

For a permutation $\auto\colon I \to I$, a seed
$\seed' = (I,\varepsilon')$ is said to be obtained from a seed $\seed$
by the \emph{automorphism} $\auto\colon \seed \to \seed'$ if
\begin{equation}
  \varepsilon'_{\auto(i)\auto(j)} = \varepsilon_{ij} \,.
\end{equation}
On the seed tori $\auto$ acts by relabeling coordinates:
\begin{equation}
  \label{eq:s}
  \auto^*(A'_{\auto(i)}) = A_i \,,
  \quad
  \auto^*(X'_{\auto(i)}) = X_i \,,
  \quad
  \auto^*(B'_{\auto(i)}) = B_i \,.
\end{equation}

A \emph{cluster transformation} $\ct\colon \seed \to \seed'$ is a
composition of mutations and automorphisms that takes a seed $\seed$
to a seed $\seed'$.  We write
\begin{equation}
  \seed \xrightarrow{\ct} \seed'
\end{equation}
to mean that $\seed'$ is obtained from $\seed$ by a cluster
transformation $\ct$.  A cluster transformation
$\ct\colon \seed \to \seed'$ induces birational maps from the seed
tori assigned to $\seed$ to those assigned to $\seed'$, which are the
compositions of the birational maps corresponding to the mutations and
automorphisms consisting of $\ct$.

\begin{definition}
  A cluster transformation $\ct\colon \seed \to \seed$ is said to be
  \emph{trivial} if it acts on $\CA_\seed$, $\CX_\seed$ and
  $\cD_\seed$ by the identity maps.
\end{definition}

\begin{remark}
  In fact, $\ct$ is trivial if it acts trivially on one of
  $\CA_\seed$, $\CX_\seed$ and $\cD_\seed$.  If $\ct$ acts trivially
  on $\CX_\seed$ or $\cD_\seed$, then it acts on the tropical
  $\CX$-variables trivially and hence is trivial by
  Theorem~\ref{eq:trivial-trop}.  Suppose $\ct$ acts on $\CA_\seed$
  trivially.  For each vertex $i \in I$, add another vertex and
  connect it to the vertex $i$ by an arrow; the resulting seed
  $\widetilde{\seed} = (\tilde{I}, \tilde{\varepsilon})$ has
  $\det \tilde{\varepsilon} = \pm 1$.  Theorem 4.3 of \cite{MR2931902}
  (applied to the case in which the semifield $\mathbb{P}$ is trivial)
  shows that the cluster transformation $\tilde{\ct}$ on
  $\widetilde{\seed}$ corresponding to $\ct$ leaves
  $\widetilde{\seed}$ invariant and acts on $\CA_{\widetilde{\seed}}$
  trivially.  Then, $\tilde{\ct}$ also acts on
  $\CX_{\widetilde{\seed}}$ trivially because the isomorphism
  $\tilde{p}_{\widetilde{\seed}}\colon \CA_{\widetilde{\seed}} \to
  \CX_{\widetilde{\seed}}$ given by
  $p_{\widetilde{\seed}}^* X_k = \prod_{i \in I}
  A_i^{\tilde{\varepsilon}_{ki}}$ commutes with $\tilde{\ct}$.  Since
  $\tilde{\ct}$ restricts to $\ct$ on $\CX_\seed$, the latter acts on
  $\CX_\seed$ trivially.
\end{remark}

Let $\ct\colon \seed \to \seed'$ be a cluster transformation from a
seed $\seed = (I, \varepsilon)$ to a seed
$\seed' = (I, \varepsilon')$.  By definition there is a decomposition
\begin{equation}
  \label{eq:c-decomposition}
  \ct = \ct[N] \circ \dotsb \circ \ct[2] \circ \ct[1] \,,
\end{equation}
where $\ct[t]$ is either a mutation or an automorphism.  This
decomposition defines a collection of seeds
$\seed[t] = (I, \varepsilon[t])$, $t = 1$, $2$, $\dotsc$, $N+1$, such
that
\begin{equation}
  \seed =: \seed[1]
  \xrightarrow{\ct[1]} \seed[2]
  \xrightarrow{\ct[2]}
  \dotsb
  \xrightarrow{\ct[N]} \seed[N+1] := \seed' \,.
\end{equation}
The cluster transformation $\ct$ is trivial if and only if
$\varepsilon[N+1] = \varepsilon$, $A[N+1] = A$, $X[N+1] = X$ and
$B[N+1] = B$.

Gluing the seed $\CA$-tori $\CA_{\seed'}$ assigned to all seeds
$\seed'$ related to $\seed$ by cluster transformations, we obtain the
\emph{cluster $\CA$-variety} $\CA_{|\seed|}$.  Similarly, the
\emph{cluster $\CX$-variety} $\CX_{|\seed|}$ and the \emph{cluster
  $\cD$-variety} $\cD_{|\seed|}$ are constructed from the seed
$\CX$-tori $\CX_{\seed'}$ and the seed $\cD$-tori $\cD_{\seed'}$
assigned to all seeds $\seed'$ related to $\seed$ by cluster
transformations.  The Poisson structure on $\cD_{\seed}$ given by
\begin{equation}
  \label{eq:D-Poisson}
  \{X_i, X_j\} = \varepsilon_{ij} X_i X_j \,,
  \quad
  \{X_i, B_j\} = \delta_{ij} X_i B_j \,,
  \quad
  \{B_i, B_j\} = 0
\end{equation}
is invariant under cluster transformations and defines a Poisson
structure on $\cD_{|\seed|}$ and $\CX_{|\seed|}$.

\subsection{Tropical $\CX$-variables}

Sometimes it is convenient to introduce the cluster variables $A$, $X$
and $B$ as formal variables assigned to a seed $\seed$.  In this
context, the pairs $(\seed, A)$, $(\seed, X)$ and $(\seed, (X,B))$ are
referred to as an \emph{$\CA$-seed}, an \emph{$\CX$-seed} and a
\emph{$\cD$-seed}, respectively, and a cluster transformation
$\ct\colon \seed \to \seed'$ is interpreted as providing relations
between variables $A$, $X$, $B$ assigned to $\seed$ and variables
$A'$, $X'$, $B'$ assigned to $\seed'$ via formulas \eqref{eq:A'},
\eqref{eq:X'}, \eqref{eq:B'} and \eqref{eq:s}.  For each $t$ in a
decomposition~\eqref{eq:c-decomposition} of $\ct$, the variable
$A_i[t]$ (pulled back by $c[t-1] \circ \dotsb \circ c[2] \circ c[1]$)
is an element of the \emph{universal semifield} $\Puniv(A)$, which is
the set of nonzero rational functions in the initial variables
$A = (A_i)_{i \in I}$ that can be written as subtraction-free
expressions, endowed with addition and multiplication.  Similarly,
$X_i[t] \in \Puniv(X)$ and $B_i[t] \in \Puniv((X,B))$.

For an $I$-tuple of variables $u = (u_i)_{i \in I}$, the
\emph{tropical semifield} $\Ptrop(u)$ is defined as the abelian
multiplicative group freely generated by $u$, endowed with the
addition $\oplus$ given by
\begin{equation}
  \prod_{i \in I} u_i^{a_i} \oplus \prod_{i \in I} u_i^{b_i}
  := \prod_{i \in I} u_i^{\min(a_i,b_i)} \,.
\end{equation}
There is a semifield homomorphism
$\pitrop\colon \Puniv(u) \to \Ptrop(u)$, known as the
\emph{tropicalization map}, such that
\begin{equation}
  \pitrop(u_i) = u_i \,,
  \quad
  i \in I \,,
  \qquad
  \pitrop(c) = 1
  \,,
  \quad
  c \in \Q_{>0} \,.
\end{equation}
The tropicalization of $\CX$-variables are called \emph{tropical
  $\CX$-variables}.

For any $t \in \{1, 2, \dotsc, L\}$, we have
\begin{equation}
  \Xtrop_i[t] := \pitrop(X_i[t])  = \prod_{j \in I} \Xtrop_j^{a_{ij}[t]}
\end{equation}
with either $a_{ij}[t] \geq 0$ for all $j \in I$ or
$a_{ij}[t] \leq 0$ for all $j \in I$ (``sign coherence'').  The
\emph{tropical sign} $\sgntrop(\Xtrop_i[t])$ of $\Xtrop_i[t]$ is defined by
\begin{equation}
  \sgntrop(\Xtrop_i[t])
  :=
  \begin{cases}
    +1 & \text{if $a_{ij}[t] \geq 0$ for all $j \in I$} \,;
    \\
    -1 & \text{if $a_{ij}[t] \leq 0$ for all $j \in I$} \,.
  \end{cases}
\end{equation}
If $\ct[t]$ is a mutation in the direction of $k \in I$, then
\begin{equation}
  \label{eq:y-mutation}
  \Xtrop_i[t+1]
  =
  \begin{cases}
    \Xtrop_k[t]^{-1}
    & \text{if $i = k$} \,;
    \\
    \Xtrop_i[t] \Xtrop_k[t]^{[\sgntrop(\Xtrop_k[t]) \varepsilon_{ik}]_+}
    & \text{if $i \neq k$} \,,
  \end{cases}
\end{equation}
where $[a]_+ = a$ for $a \geq 0$ and $0$ for $a < 0$.

The following theorem provides considerable simplification of the
criterion for a cluster transformation to be trivial:
\begin{theorem}[{\cite[Theorem 5.1]{MR3029994}}]
  \label{eq:trivial-trop}
  A cluster transformation $\ct\colon \seed \to \seed$ from a seed
  $\seed = (I, \varepsilon)$ to itself is trivial if and only if
  $\pitrop(\ct^*(X_i)) = \pitrop(X_i)$ for all $i \in I$.
\end{theorem}

\subsection{Quantum cluster varieties}

There is canonical deformation quantization of the algebra of regular
functions on $\CX_{|\seed|}$ and that on $\cD_{|\seed|}$, which leads
to the notions of quantum cluster $\CX$-variety and quantum cluster
$\cD$-variety.

Let $q$ be a formal parameter.  For a seed
$\seed = (I,\varepsilon)$, the \emph{quantum torus algebra}
$\mathbf{D}^q_\seed$ is the algebra over $\Z[q, q^{-1}]$ generated by
variables $(X^q, B^q) = (X^q_i, B^q_i)_{i \in I}$ subject to the
relations
\begin{equation}
  \label{eq:XB}
  q^{-\varepsilon_{ij}} X^q_i X^q_j =  q^{-\varepsilon_{ji}} X^q_j X^q_i \,,
  \quad
  q^{-\delta_{ij}} X^q_i B^q_j = q^{\delta_{ij}} B^q_j X^q_i \,,
  \quad
  B^q_i B^q_j = B^q_j B^q_i \,.
\end{equation}
We introduce an algebra over $\C$ generated by a formal parameter
$\hbar$ and variables $x^\hbar = (x^\hbar_i)_{i \in I}$,
$b^\hbar = (b^\hbar_i)_{i \in I}$ such that
\begin{equation}
  \label{eq:xb-comm}
  [x^\hbar_i, x^\hbar_j] = 2\pi\iu \hbar \varepsilon_{ij} \,,
  \quad
  [x^\hbar_i, b^\hbar_j] = 2\pi\iu \hbar \delta_{ij} \,,
  \quad
  [b^\hbar_i, b^\hbar_j] = 0 \,.
\end{equation}
Then, 
\begin{equation}
  q := \exp(\pi\iu\hbar) \,,
  \quad
  X^q_i := \exp(x^\hbar_i) \,,
  \quad
  B^q_i := \exp(b^\hbar_i)
\end{equation}
gives an embedding of $\mathbf{D}^q_\seed$ into this algebra.

We also introduce
\begin{equation}
  \xt^\hbar_i := x^\hbar_i + \sum_{j \in I} \varepsilon_{ij} b^\hbar_j
\end{equation}
and set
\begin{equation}
  \label{eq:Xt}
  \Xt^q_i := \exp(\xt^\hbar_i) = X^q_i \prod_{j \in I} (B^q_j)^{\varepsilon_{ij}} \,.
\end{equation}
These variables satisfy the relations
\begin{equation}
  \label{eq:xt-comm}
  [\xt^\hbar_i, \xt^\hbar_j] = -2\pi\iu \hbar \varepsilon_{ij} \,,
  \quad
  [\xt^\hbar_i, b^\hbar_j] = 2\pi\iu \hbar \delta_{ij} \,,
  \quad
  [x^\hbar_i, \xt^\hbar_j] = 0
\end{equation}
and
\begin{equation}
  \label{eq:XtBX}
  q^{\varepsilon_{ij}} \Xt^q_i \Xt^q_j =  q^{\varepsilon_{ji}} \Xt^q_j \Xt^q_i \,,
  \quad
  q^{-\delta_{ij}} \Xt^q_i B^q_j = q^{\delta_{ij}} B^q_j \Xt^q_i \,,
  \quad
  X^q_i \Xt^q_j = X^q_j \Xt^q_i \,.
\end{equation}

Let $\mathbb{D}^q_\seed$ be the skew-field of fractions of
$\mathbf{D}^q_\seed$.  For a cluster transformation
$\ct\colon \seed \to \seed'$, we define the \emph{quantum cluster
  transformation}
\begin{equation}
  \ct^q \colon \mathbb{D}^q_{\seed'} \to \mathbb{D}^q_\seed
\end{equation}
as follows.  (Note the direction; $\ct^q$ quantizes the pullback by
$\ct\colon \cD_\seed \to \cD_{\seed'}$.)  For $\ct = \auto$, the map
$\ct^q = \auto^q$ is given by
\begin{equation}
  \auto^q(X'^q_{\auto(i)}) = X^q_i \,,
  \qquad
  \auto^q(B'^q_{\auto(i)}) = B^q_i \,.
\end{equation}
For $\ct = \mu_k$, the map $\ct^q = \mu_k^q$ is a composition of two
maps:
\begin{equation}
  \label{eq:mu_k^q}
  \mu_k^q
  := \mu_k^{\sharp(+)} \circ \mu_k^{\flat(+)}
  = \mu_k^{\sharp(-)} \circ \mu_k^{\flat(-)} \,.
\end{equation}
The fact that $\mu_k^q$ admits two decompositions will be important.
In general, $\ct^q$ is given by the composition
$\ct^q[1] \circ \ct^q[2] \circ \dotsb \circ \ct^q[N]$ of quantum
cluster transformations corresponding to a decomposition of $\ct$ into
mutations and automorphisms.

The ``automorphism part''
$\mu_k^{\sharp(\epsilon)}\colon \mathbb{D}^q_\seed \to
\mathbb{D}^q_\seed$ of $\mu_k^q$ is defined by
\begin{equation}
  \mu_k^{\sharp(\epsilon)}
  :=
  \Ad_{\Psi^q((X^q_k)^\epsilon)^\epsilon \Psi^q((\Xt^q_k)^\epsilon)^{-\epsilon}} \,,
  \quad
  \epsilon = \pm \,,
\end{equation}
where
\begin{equation}
  \Psi^q(x) := \prod_{k=1}^\infty (1 + q^{2k-1} x)^{-1}
\end{equation}
is the \emph{quantum dilogarithm}.  Using the difference equation
\begin{equation}
  \Psi^q(q^2 x) = (1 + qx) \Psi^q(x) \,,
\end{equation}
one can show that $\mu_k^{\sharp(\epsilon)}(X'^q_i)$ and
$\mu_k^{\sharp(\epsilon)}(B'^q_i)$ belong to $\mathbb{D}^q_\seed$
\cite[Lemma 3.2]{MR2470108}.

The ``monomial part''
$\mu_k^{\flat(\epsilon)}\colon \mathbf{D}^q_{\seed'} \to
\mathbf{D}^q_\seed$ is given by
\begin{equation}
  \mu_k^{\flat(\epsilon)}(X'^\hbar_i)
  := \exp\bigl(m_k^{(\epsilon)}(x'^\hbar_i)\bigr) \,,
  \qquad
  \mu_k^{\flat(\epsilon)}(B'^\hbar_i)
  := \exp\bigl(m_k^{(\epsilon)}(b'^\hbar_i)\bigr) \,,
\end{equation}
where
\begin{align}
  \label{eq:x-mutation-monomial}
  m_k^{(\epsilon)}(x'^\hbar_i)
  &:=
  \begin{cases}
    -x^\hbar_k
    &
    \text{if $i = k$} \,;
    \\
    x^\hbar_i + [\epsilon\varepsilon_{ik}]_+ x^\hbar_k
    &
    \text{if $i \neq k$} \,,
  \end{cases}
  \\
  m_k^{(\epsilon)}(b'^\hbar_i)
  &:=
  \begin{cases}
    -b^\hbar_k + \sum_{j \in I} [-\epsilon\varepsilon_{kj}]_+ b^\hbar_j
    &
    \text{if $i = k$} \,;
    \\
    b^\hbar_i
    &
    \text{if $i \neq k$} \,.
  \end{cases}
\end{align}
The transformations of $x^\hbar$ and $b^\hbar$ are dual to each other:
if we write
$m_k^{(\epsilon)}(x'^\hbar_i) = \sum_{j \in I}
(M_k^{(\epsilon)})_{ij} x^\hbar_j$ using a matrix $M_k^{(\epsilon)}$,
then
$m_k^{(\epsilon)}(b'^\hbar_i) = \sum_{j \in I} b^\hbar_j
((M_k^{(\epsilon)})^{-1})_{ji}$.

For $q = 1$, the formula for $\ct^q$ reduces to that for the action of
$\ct$ on $\cD_\seed$.  In particular, if the quantum cluster
transformation $\ct^q$ induced from $\ct\colon \seed \to \seed$ is the
identity map, then $\ct$ also acts trivially on $\cD_\seed$ and hence
on the tropicalization of $\CX_\seed$, which implies that $\ct$ is a
trivial cluster transformation by Theorem~\ref{eq:trivial-trop}.  It
turns out that the converse is also true:
\begin{proposition}  
  \label{prop:trivial-cq}
  A cluster transformation $\ct\colon \seed \to \seed$ is trivial if
  and only if $\ct^q = \id_{\mathbb{D}^q_\seed}$.
\end{proposition}

This proposition can be proved from Proposition 5.21 of
\cite{MR2470108} and Theorem~4.3 of \cite{MR2931902}.  The proposition
also follows from Proposition \ref{prop:trivial-K}.

\subsection{Representations of quantum cluster varieties}

From now on we take $\hbar$ to be a positive real number.  For a
seed $\seed = (I,\varepsilon)$, let
\begin{equation}
  \CH_\seed := L^2(\CA_\seed^+)
\end{equation}
be the Hilbert space of square-integrable complex functions on the set
of positive real points $\CA_\seed^+ \cong \R^I$ of $\CA_\seed$.  This
is the Hilbert space of states in quantum mechanics of a particle
moving in $\CA_\seed^+$.  Coordinates of $\CA_\seed^+$ are given by
$a_i := \log A_i$, $i \in I$.

The differential operators
\begin{equation}
  \xop_i
  = \pi\iu\hbar \frac{\partial}{\partial a_i}
    - \sum_{j \in I} \varepsilon_{ij} a_j \,,
  \qquad
  \bop_i = 2a_i
\end{equation}
on functions in $\CH_\seed$ satisfy the
commutation relations~\eqref{eq:xb-comm}, and their exponentials
\begin{equation}
  \Xop_i := \exp(\xop_i) \,,
  \qquad
  \Bop_i := \exp(\bop_i) 
\end{equation}
satisfy relations~\eqref{eq:XB}.  The operators
\begin{equation}
  \xtop_i
  = \pi\iu\hbar \frac{\partial}{\partial a_i}
  + \sum_{j \in I} \varepsilon_{ij} a_j \,,
  \qquad
  \Xtop_i := \exp(\xtop_i) \,.
\end{equation}
satisfy relations~\eqref{eq:xt-comm} and \eqref{eq:XtBX}.

Let $\mathbb{L}_\seed$ be the space of Laurent polynomials in the
quantum variables $(X^q, B^q)$ assigned to $\seed$ such that for any
cluster transformation $\ct'\colon \seed' \to \seed$, the quantum
cluster transformation
$\ct'^q\colon \mathbb{D}^q_\seed \to \mathbb{D}^q_{\seed'}$ maps them
to Laurent polynomials in $(X'^q, B'^q)$.  The operators corresponding
to the elements of $\mathbb{L}_\seed$ preserve a certain subspace
$\CS_\seed$ of rapidly decreasing functions in $\CH_\seed$, and they
provide a representation of $\mathbb{L}_\seed$ on $\CS_\seed$.

A cluster transformation $\ct\colon \seed \to \seed'$ gives rise to an
isomorphism of algebras
$\ct^q\colon \mathbb{L}_{\seed'} \to \mathbb{L}_\seed$.  The
representation of $\mathbb{L}_\seed$ on $\CS_\seed$ and that of
$\mathbb{L}_{\seed'}$ on $\CS_{\seed'}$ are intertwined by a unitary
operator $\K{\ct}\colon \CH_{\seed'} \to \CH_\seed$:
\begin{equation}
  \K{\ct} \widehat{A} \K{\ct}^{-1} = \widehat{\ct^q(A)} \,,
  \qquad
  A \in \mathbb{L}_{\seed'} \,.
\end{equation}
Here $\widehat{A}$ denotes the operator corresponding to $A$.

For $\ct = \auto$, the intertwiner is simply
\begin{equation}
  \K{\auto} := \auto^* \,,
\end{equation}
the pullback by $\auto$ considered as a transformation on $\R^I$:
\begin{equation}
  \auto\colon
  (a_i)_{i \in I}
  \mapsto (a'_i)_{i \in I} = (a_{\alpha^{-1}(i)})_{i \in I} \,.
\end{equation}

The intertwiner $\K{\mu_k}$ for a mutation $\mu_k$ decomposes as
\begin{equation}
  \label{eq:K_mu_k}
  \K{\mu_k}
  :=
  \K{\mu_k}^{\sharp(+)}
  \K{\mu_k}^{\flat(+)}
  =
  \K{\mu_k}^{\sharp(-)}
  \K{\mu_k}^{\flat(-)} \,,
\end{equation}
corresponding to the decompositions~\eqref{eq:mu_k^q}.  The fact that
$\K{\mu_k}$ admits these two decompositions was observed by
Kim~\cite{MR4179968}.  We will show the equality of the two
decompositions in section~\ref{sec:QFT}.

The monomial part
$\K{\mu_k}^{\flat(\epsilon)}\colon \CH_{\seed'} \to \CH_\seed$ is the
pullback by the transformation
\begin{equation}
  \label{eq:a-mutation-monomial}
  (a_i)_{i \in I} \mapsto (a'_i)_{i \in I} \,,
  \quad
  a'_i
  =
  \begin{cases}
    -a_k + \sum_{j \in I} [-\epsilon\varepsilon_{kj}]_+ a_j
    &
    \text{if $i = k$} \,;
    \\
    a_i
    &
    \text{if $i \neq k$} \,.
  \end{cases}
\end{equation}

The automorphism part
$\K{\mu_k}^{\sharp(\epsilon)}\colon \CH_\seed \to \CH_\seed$ is given by
\begin{equation}
  \K{\mu_k}^{\sharp(\epsilon)}
  = \Phi^\hbar(\epsilon\xop_k)^\epsilon \Phi^\hbar(\epsilon\xtop_k)^{-\epsilon} \,.
\end{equation}
The function $\Phi^\tau(z)$, $\tau \in \C$, is the \emph{noncompact
  quantum dilogarithm}.  It is defined for
$|\Im z| < \pi(1 + \Re \tau)$ by the integral
\begin{equation}
  \Phi^\tau(z)
  :=
  \exp\biggl(-\frac14 \int_{\R + \iu0}
  \frac{e^{-\iu wz}}{\sinh(\pi w) \sinh(\pi\tau w)}
  \frac{\rmd w}{w}\biggr) \,,
\end{equation}
where the contour along the real axis going from $-\infty$ to $\infty$
and bypassing the origin from above, and is analytically continued to
the entire complex plane.  The unitarity of $\K{\mu_k}$ follows from
the property that 
\begin{equation}
  \overline{\Phi^\tau(z)} = \Phi^\tau(\bar{z})^{-1}
\end{equation}
if $\tau$ is a positive real number or a pure phase.

For $\Im \tau > 0$, we have
\begin{equation}
  \Phi^\tau(z) = \frac{\Psi^q(e^z)}{\Psi^{1/q^\vee}(e^{z/\tau})} \,,
  \quad
  q^\vee := e^{\pi\iu/\tau} \,,
\end{equation}
and hence $\Phi^\tau(z)$ satisfies the difference equations
\begin{align}
  \label{eq:Phi-diff-tau}
  \Phi^\tau(z + 2\pi\iu\tau) &= (1 + q e^z) \Phi^\tau(z) \,,
  \\
  \label{eq:Phi-diff-1}
  \Phi^\tau(z + 2\pi\iu) &= (1 + q^\vee e^{z/\tau}) \Phi^\tau(z) \,.
\end{align}
From the first of these equations we see that conjugation by
$\K{\mu_k}^{\sharp(\epsilon)}$ acts as $\mu_k^{\sharp(\epsilon)}$ on
$\Xop_i$ and $\Bop_i$.

\subsection{Quantum dilogarithm identities}

Let $\ct\colon \seed \to \seed$ be a trivial cluster transformation.
By Proposition \ref{prop:trivial-cq}, the quantum cluster
transformation $\ct^q$ is also trivial.  The dual variables
$X^{q^\vee}$, $B^{q^\vee}$, defined by
\begin{equation}
  X^{q^\vee} := \exp(x^{\hbar^\vee}) \,,
  \quad
  B^{q^\vee} := \exp(b^{\hbar^\vee}) \,,
  \quad
  x^{\hbar^\vee} := \hbar^{-1} x^\hbar \,,
  \quad
  b^{\hbar^\vee} := \hbar^{-1} b^\hbar \,,
\end{equation}
commute with $X^q$, $B^q$ and satisfy relations \eqref{eq:XB} with $q$
replaced by $q^\vee = \exp(\pi\iu/\hbar)$.  The difference
equation~\eqref{eq:Phi-diff-1} shows that conjugation by $\K{\mu_k}$
acts on $\Xop^\vee$, $\Bop^\vee$ as $\mu_k^{q^\vee}$, whereas
conjugation by $\K{\auto}$ acts in the same way on $\Xop^\vee$,
$\Bop^\vee$ and on $\Xop$, $\Bop$.  It follows that $\K{\ct}$ commutes
with $\Xop^\vee$, $\Bop^\vee$ since it commutes with $\Xop$, $\Bop$.
The fact that $\K{\ct}$ commutes with both sets of variables implies
that $\K{\ct} = \lambda_{\ct} \id_{\CH_\seed}$ for some complex number
$\lambda_{\ct}$ with $|\lambda_{\ct}| = 1$ \cite[Theorem
5.4]{MR2470108}.

In fact, $\lambda_{\ct} = 1$, as was pointed out for important
specific cluster transformations by Kim~\cite{MR4179968}.

\begin{proposition}
  \label{prop:trivial-K}
  $\K{\ct} = \id_{\CH_\seed}$ for a trivial cluster transformation
  $\ct\colon \seed \to \seed$.
\end{proposition}

We will demonstrate this proposition using quantum dilogarithm
identities proved in \cite{MR2861174}, which we now explain.

Consider a decomposition of $\ct$ into mutations and automorphisms.
The positions of the automorphisms in the decomposition can be moved
by the relation
\begin{equation}
  \mu_k \circ \auto = \auto \circ \mu_{\auto^{-1}(k)} \,,
\end{equation}
hence we can decompose $\ct$ into the form
\begin{equation}
  \ct\colon \seed
  =: \seed[1]
  \xrightarrow{\mu_{k[1]}} \seed[2]
  \xrightarrow{\mu_{k[2]}}
  \dotsb
  \xrightarrow{\mu_{k[L]}} \seed[L+1]
  \xrightarrow{\auto}
  \seed \,.
\end{equation}
Let $\Xtrop_i[t] := \pitrop(X_i[t])$ be the tropicalization of the variable
$X_i[t]$ of $\CX_{\seed[t]}$.  For each $t$, the relation between
$\Xtrop[t]$ and $\Xtrop = \Xtrop[1]$ defines an $I$-tuple of integers
$\gamma[t] = (\gamma_i[t])_{i \in I} \in \Z^I$ and a sign
$\sgntrop[t] \in \{\pm\}$ by
\begin{equation}
  \Xtrop_{k[t]}[t] =: \prod_{i \in I} \Xtrop_i^{\gamma_i[t]} \,,
  \qquad
  \sgntrop[t] := \sgntrop(\Xtrop_{k[t]}[t])
  =
  \begin{cases}
    +1 & \text{$\gamma_i[t] \geq 0$ for all $i \in I$} \,;
     \\
     -1 & \text{$\gamma_i[t] \leq 0$ for all $i \in I$} \,.
  \end{cases}
\end{equation}
The noncompact quantum dilogarithm satisfies the following identity:

\begin{theorem}[{\cite[Theorem 4.5]{MR2861174}}, {\cite[Theorem
    5.16]{MR2931896}}]
  \begin{equation}
    \label{eq:Phi-identity-tropical}
    \Phi^\hbar(\sgntrop[1] \gamma[1] \cdot \xop)^{\sgntrop[1]}
    \Phi^\hbar(\sgntrop[2] \gamma[2] \cdot \xop)^{\sgntrop[2]}
    \dotsm
    \Phi^\hbar(\sgntrop[L] \gamma[L] \cdot \xop)^{\sgntrop[L]}
    = \id_{\CH_\seed} \,.
  \end{equation}
\end{theorem}

The operator
$\gamma[t] \cdot \xop := \sum_{i \in I} \gamma_i[t] \cdot \xop_i$
appearing in the above identity can be understood as follows.
Comparing the transformation~\eqref{eq:y-mutation} of $\Xtrop_i[t]$
under $\mu_k$ and the definition~\eqref{eq:x-mutation-monomial} of
$m_k^{(\epsilon)}$, we see that $x^\hbar_i[t+1]$ transforms under
$m_{k[t]}^{(\sgntrop[t])}$ in the same way as $\log \Xtrop_i[t+1]$
does under $\mu_{k[t]}$.  Thus we have
\begin{equation}
  \gamma[t] \cdot x^\hbar
  =
  \mu_{k[1]}^{\flat(\sgntrop[1])}
  \circ \mu_{k[2]}^{\flat(\sgntrop[2])}
  \circ \dotsm
  \circ \mu_{k[t-1]}^{\flat(\sgntrop[t-1])}
  (x^\hbar_{k[t]}[t]) \,,
\end{equation}
and $\gamma[t] \cdot \xop$ is the corresponding operator obtained from
$\xop_{k[t]}[t]$ by conjugation by
$\K{\mu_{k[1]}}^{\flat(\sgntrop[1])}
\K{\mu_{k[2]}}^{\flat(\sgntrop[2])} \dotsm
\K{\mu_{k[t-1]}}^{\flat(\sgntrop[t-1])}$.

\begin{proof}[Proof of Proposition \ref{prop:trivial-K}]
  Since $\overline{\Phi^\hbar(z)} = \Phi^\hbar(\bar{z})^{-1}$ and
  $\overline{\xop_i} = -\xtop_i$, taking the complex conjugate of
  identity~\eqref{eq:Phi-identity-tropical} and conjugating the
  resulting identity with the map $a \mapsto -a$, we obtain
  \begin{equation}
    \Phi^\hbar(\sgntrop[1]\gamma[1] \cdot \xtop)^{-\sgntrop[1]}
    \Phi^\hbar(\sgntrop[2]\gamma[2] \cdot \xtop)^{-\sgntrop[2]}
    \dotsm
    \Phi^\hbar(\sgntrop[L]\gamma[L] \cdot \xtop)^{-\sgntrop[L]}
    = \id_{\CH_\seed} \,.
  \end{equation}
  Furthermore, $a_i$ transforms under
  $\K{\mu_{k[t]}}^{\flat(\sgntrop[t])}$ in the same way as $b^\hbar_i$
  does under $m_{k[t]}^{(\sgntrop[t])}$, and the latter transforms in
  the dual manner to $x^\hbar_i$ and hence to the transformation of
  $\log\Xtrop_i[t]$ under $\mu_{k[t]}$.  Since $\ct$ acts on
  $\Xtrop_i$ trivially by Theorem~\ref{eq:trivial-trop}, we have
  \begin{equation}
    \K{\mu_{k[1]}}^{\flat(\sgntrop[1])}
    \K{\mu_{k[2]}}^{\flat(\sgntrop[2])}
    \dotsm
    \K{\mu_{k[L]}}^{\flat(\sgntrop[L])}
    \K{\auto}
    =
    \id_{\CH_\seed} \,.
  \end{equation}
  Therefore,
  \begin{equation}
    \begin{split}
      \K{\ct}
      &=
      \K{\mu_{k[1]}}^{\sharp(\sgntrop[1])}
      \K{\mu_{k[1]}}^{\flat(\sgntrop[1])}
      \K{\mu_{k[2]}}^{\sharp(\sgntrop[2])}
      \K{\mu_{k[2]}}^{\flat(\sgntrop[2])}
      \dotsm
      \K{\mu_{k[L]}}^{\sharp(\sgntrop[L])}
      \K{\mu_{k[L]}}^{\flat(\sgntrop[L])}
      \K{\auto}
      \\
      &=
      \Phi^\hbar(\sgntrop[1] \gamma[1] \cdot \xop)^{\sgntrop[1]}
      \Phi^\hbar(\sgntrop[2] \gamma[2] \cdot \xop)^{\sgntrop[2]}
      \dotsm
      \Phi^\hbar(\sgntrop[L] \gamma[L] \cdot \xop)^{\sgntrop[L]}
      \\
      & \quad
      \cdot
      \Phi^\hbar(\sgntrop[1] \gamma[1] \cdot \xtop)^{-\sgntrop[1]}
      \Phi^\hbar(\sgntrop[2] \gamma[2] \cdot \xtop)^{-\sgntrop[2]}
      \dotsm
      \Phi^\hbar(\sgntrop[L] \gamma[L] \cdot \xtop)^{-\sgntrop[L]}
      \\
      & \quad
      \cdot
      \K{\mu_{k[1]}}^{\flat(\sgntrop[1])}
      \K{\mu_{k[2]}}^{\flat(\sgntrop[2])}
      \dotsm
      \K{\mu_{k[L]}}^{\flat(\sgntrop[L])}
      \K{\auto}
    \end{split}
  \end{equation}
  is the identity map.
\end{proof}

\section{Trivial cluster transformations from the longest element of
  $S_4$}
\label{sec:trivial-ct}

In this section we introduce the three families of quivers described
in section~\ref{sec:introduction}, which are assigned to the words for
permutations in the symmetric group $S_n$.  A loop of braid moves on
reduced expressions for the longest element of $S_4$ yields trivial
cluster transformations acting on the seed tori of relevant quivers.

\subsection{Symmetric groups and wiring diagrams}
\label{sec:wiring}

The \emph{symmetric group} $S_n$ is the group of permutations of
$\{1, 2, \dotsc, n\}$.  It is generated by the \emph{adjacent
  transpositions} $\{s_a\}_{a = 1}^{n-1}$ satisfying the
relations
\begin{alignat}{2}
  & s_a^2 = 1 \,,
  \\
  \label{eq:s_is_j=s_js_i}
  & s_a s_b = s_b s_a \quad \text{for $|a - b| \geq 2$}
                      & \quad & \text{(far commutativity)}\,,
  \\
  & s_a s_{a+1} s_a = s_{a+1} s_a s_{a+1}
                                  & \quad & \text{(braid relation)}\,.
\end{alignat}

A \emph{word} for a permutation $s \in S_n$ is a finite string
$a_1 a_2 \dotso a_k$ of elements of $\{1, 2, \dotsc, n\}$ such that
$s = s_{a_1} s_{a_2} \dotsm s_{a_k}$.  The \emph{length} $l(s)$ of $s$
is the minimal number such that
$s = s_{a_1} s_{a_2} \dotsm s_{a_{l(s)}}$ for some $a_1$, $a_2$,
$\dotsc$, $a_{l(s)} \in \{1, 2, \dotsc, n\}$.  The expression
$s_{a_1} s_{a_2} \dotsm s_{a_{l(s)}}$ and the string
$a_1 a_2 \dotsc a_{l(s)}$ are called a \emph{reduced expression} for
$s$ and a \emph{reduced word} for $s$, respectively.  A reduced
expression for a permutation can be transformed to any other reduced
expression for the same permutation by a sequence of far commutativity
and braid relations (Tits' lemma).  The longest element of $S_n$ is
the order reversing permutation $a \mapsto n - a + 1$ and its length
is $n(n-1)/2$.

To the end of constructing quivers corresponding to words for
permutations in $S_n$, we represent words diagrammatically.

\begin{definition}
  A \emph{wiring diagram on $n$ wires} is a union of $n$ continuous
  paths, called \emph{wires}, inside the vertical strip
  $\{(x,y) \in \R^2 \mid 0 \leq x \leq 1\}$ such that
  \begin{enumerate}
  \item the wires start from distinct points on the left boundary
    ($x = 0$) and end at distinct points on the right boundary
    ($x = 1$);
  \item no three wires intersect at a point; and
  \item no two intersections of wires take place at the same
    horizontal position ($x$-coordinate).
  \end{enumerate}
  Two wiring diagrams are identified if they are related by an isotopy
  that preserves the horizontal ordering of the intersections of
  wires.
\end{definition}

The words for permutations in $S_n$ are in one-to-one correspondence
with the wiring diagrams on $n$ wires.  To represent a word
$a_1 a_2 \dotso a_k$ by a wiring diagram, we move rightward in the
positive $x$-direction from the left boundary of the vertical strip
toward the right boundary, and let wires intersect according to the
letters appearing in the reduced word.  Thus, we first let the $a_1$th
wire and the $(a_1 + 1)$st wire intersect, counted from bottom to top
in the ascending order of the $y$-coordinates of the wires, and next
let the $a_2$th wire and the $(a_2 + 1)$st intersect, again counted
from bottom to top but at a horizontal position right to the first
intersection, and so on.  Conversely, given a wiring diagram on $n$
wires, the corresponding word for a permutation can be written down.

\begin{definition}
  A wiring diagram is said to be \emph{reduced} if no two wires
  intersect more than once.
\end{definition}

The reduced words are in one-to-one correspondence with the reduced
wiring diagrams.  For example, the reduced word $123121$ for the
longest element of $S_4$ is represented by the reduced wiring diagram
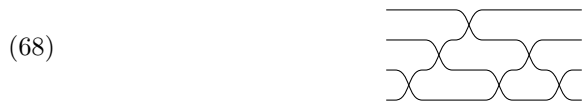
\begin{equation}
  \begin{tikzpicture}[rotate=90]
    \braid[gap=0, scale=0.4] a_1 a_2 a_3 a_1 a_2 a_1;
  \end{tikzpicture}
\end{equation}
As can be seen from this example, a wiring diagram for a reduced word
of the longest element of $S_n$ has the property that each pair of
wires intersect exactly once.

Given a wiring diagram on $n$ wires, we name the wires $1$, $2$,
$\dotsc$, $n$ from bottom to top according to their positions on the
left boundary of the vertical strip.  We need to label the segments
and intersections of wires.  The segments of wire $a$ is labeled, from
left to right, $a_1$, $a_2$, $\dotsc$, $a_n$.  The intersection of
wires $a$ and $b$ with $a < b$ is labeled $ab$.

A wiring diagram divides the vertical strip into \emph{chambers},
i.e., connected components of the complement of the wires, and we also
need labels for them.  We label the chambers with subsets of
$\{1, 2, \dotsc, n\}$ as follows.  First, we label the chamber
extending to $y = -\infty$ the empty set $\emptyset$.  Starting from
the chamber $\emptyset$, we can go to any other chamber by crossing
some number of wires.  If we reach a new chamber from a chamber $C$ by
crossing wire $a$, then we label that chamber $C \cup \{a\}$.

Figure~\ref{fig:labels-121} illustrates our labeling scheme with the
wiring diagram for the reduced word $121$ for the longest element of
$S_3$.

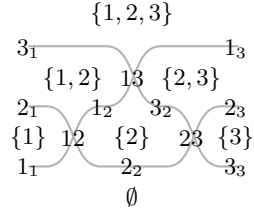
\begin{figure}
  \centering
  \begin{tikzpicture}[scale=0.8, rotate=90]
    \braid[gap=0, thick, black!30] a_1 a_2 a_1;

    \node[font=\small, vertex] (1_1) at (1,0) {$1_1$};
    \node[font=\small, vertex] (1_2) at (2,-1.25) {$1_2$};
    \node[font=\small, vertex] (1_3) at (3,-3.5) {$1_3$};
    
    \node[font=\small, vertex] (2_1) at (2,0) {$2_1$};
    \node[font=\small, vertex] (2_2) at (1,-1.75) {$2_2$}; 
    \node[font=\small, vertex] (2_3) at (2,-3.5) {$2_3$};
    
    \node[font=\small, vertex] (3_1) at (3,0) {$3_1$};
    \node[font=\small, vertex] (3_2) at (2,-2.25) {$3_2$};
    \node[font=\small, vertex] (3_3) at (1,-3.5) {$3_3$};
    
    \node[font=\small, vertex] (12) at (1.5,-0.75) {$12$};
    \node[font=\small, vertex] (13) at (2.5,-1.75) {$13$};
    \node[font=\small, vertex] (23) at (1.5,-2.75) {$23$};
    
    \node[font=\small, vertex] (e) at (0.5,-1.75) {$\emptyset$};
    \node[font=\small, vertex] (1) at (1.5,0) {$\{1\}$};
    \node[font=\small, vertex] (2) at (1.5,-1.75) {$\{2\}$};
    \node[font=\small, vertex] (3) at (1.5,-3.5) {$\{3\}$};
    \node[font=\small, vertex] (23) at (2.5,-2.75) {$\{2,3\}$};
    \node[font=\small, vertex] (12) at (2.5,-0.75) {$\{1,2\}$};
    \node[font=\small, vertex] (123) at (3.6,-1.75) {$\{1,2,3\}$};
\end{tikzpicture}
  \caption{Labeling of the segments, intersections and chambers of the
    wiring diagram for the reduced word $121$ for the longest element
    of $S_3$.}
  \label{fig:labels-121}
\end{figure}

\subsection{Quivers assigned to wiring diagrams}

\begin{definition}
  The \emph{triangle quiver}, the \emph{square quiver} and the
  \emph{butterfly quiver} assigned to a word for a permutation of
  $S_n$ is constructed as follows:
  \begin{enumerate}
  \item Around each intersection of wires of the corresponding wiring
    diagram, place vertices and connect them with arrows according to
    the rule
    \begin{equation}
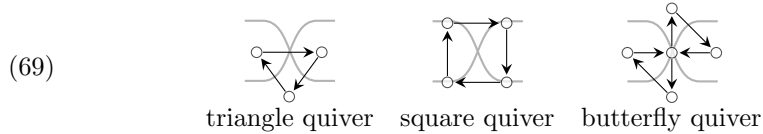
      \begin{tabular}{ccc}
        \begin{tikzpicture}[scale=0.8, rotate=90]
          \braid[gap=0, thick, black!30] a_1;
          \node[bvertex] (1) at (1.5,-0.2) {};
          \node[bvertex] (e) at (0.75,-0.75) {};
          \node[bvertex] (2) at (1.5,-1.3) {};
          \node[bvertex, draw=none] at (2.25,-0.75) {};
          \draw[q->] (e) -- (1);
          \draw[q->] (1) -- (2);
          \draw[q->] (2) -- (e);
        \end{tikzpicture}
        &
          \begin{tikzpicture}[scale=0.8, rotate=90]
            \braid[gap=0, thick, black!30] a_1;
            \node[bvertex] (1_1) at (1,-0.25) {};
            \node[bvertex] (1_2) at (2,-1.25) {};
            \node[bvertex] (2_1) at (2,-0.25) {};
            \node[bvertex] (2_2) at (1,-1.25) {};
            \draw[q->] (2_2) -- (1_1);
            \draw[q->] (1_2) -- (2_2);
            \draw[q->] (2_1) -- (1_2);
            \draw[q->] (1_1) -- (2_1);
          \end{tikzpicture}
        &
          \begin{tikzpicture}[scale=0.8, rotate=90]
            \braid[gap=0, thick, black!30] a_1;
            \node[bvertex] (1) at (1.5,0) {};
            \node[bvertex] (e) at (0.75,-0.75) {};
            \node[bvertex] (2) at (1.5,-1.5) {};
            \node[bvertex] (12) at (1.5,-0.75) {};
            \node[bvertex] (1c2) at (2.25,-0.75) {};
            
            \draw[q->] (e) -- (1);
            \draw[q->] (12) -- (e);
            \draw[q->] (1) -- (12);
            \draw[q->] (2) -- (12);
            \draw[q->] (12) -- (1c2);
            \draw[q->] (1c2) -- (2);
          \end{tikzpicture}
        \\
        triangle quiver & square quiver & butterfly quiver
      \end{tabular}
    \end{equation}
    
  \item Label each vertex with the name of the segment, intersection
    or chamber where the vertex is placed, and identify vertices with
    the same labels.
    
  \item Delete $2$-cycles formed in the previous step.
    
  \item For the triangle quiver and the butterfly quiver, add a vertex
    for each subset of $\{1,2, \dotsc, n\}$ if not present already.
    This vertex is disconnected from any other vertices.
  \end{enumerate}
\end{definition}

The introduction of disconnected vertices in the last step is a
technicality that can be avoided if we allow either mutations or
automorphisms to rename the elements of the index set for a quiver.
Figure \ref{fig:quivers-123121} shows the three quivers assigned to
the reduced word $123121$ for the longest element of $S_4$.

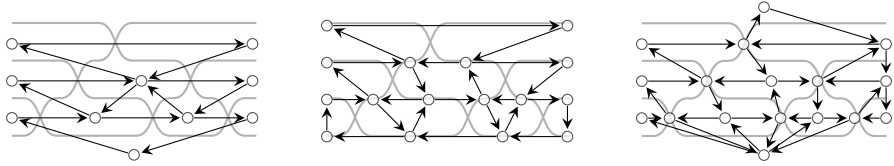
\begin{figure}
  \centering
  \begin{tikzpicture}[scale=0.5, rotate=90]
    \braid[gap=0, thick, black!30] a_1 a_2 a_3 a_1 a_2 a_1;

    \node[bvertex] (e) at (0.5,-3.3) {};
    \node[bvertex] (1) at (1.5,0) {};
    \node[bvertex] (1c2) at (2.5,0) {};
    \node[bvertex] (1c2c3) at (3.5,0) {};
    \node[bvertex] (4) at (1.5,-6.5) {};
    \node[bvertex] (3c4) at (2.5,-6.5) {};
    \node[bvertex] (2c3c4) at (3.5,-6.5) {};
    \node[vertex] (1c2c3c4) at (4.5,-3.3) {\phantom{\tikz{\node[bvertex]{};}}};

    \node[bvertex] (2) at (1.5,-2.25) {};
    \node[bvertex] (2c3) at (2.5,-3.5) {};
    \node[bvertex] (3) at (1.5,-4.75) { }; 
    
    \draw[q->] (e) -- (1);
    \draw[q->] (1) -- (2);
    
    \draw[q->] (2) -- (1c2);
    \draw[q->] (2c3) -- (2);
    \draw[q->] (1c2) -- (2c3);
    
    \draw[q->] (2c3) -- (1c2c3);
    \draw[q->] (2c3c4) -- (2c3);
    \draw[q->] (1c2c3) -- (2c3c4);
    
    \draw[q->] (2) -- (3);
    
    \draw[q->] (3) -- (2c3);
    \draw[q->] (3c4) -- (3);
    \draw[q->] (2c3) -- (3c4);
    
    \draw[q->] (4) -- (e);
    \draw[q->] (3) -- (4);
  \end{tikzpicture}
  \qquad
  \begin{tikzpicture}[scale=0.5, rotate=90]
    \braid[gap=0, thick, black!30] a_1 a_2 a_3 a_1 a_2 a_1;
  
    \node[bvertex] (1_1) at (1,0) {};
    \node[bvertex] (2_1) at (2,0) {};
    \node[bvertex] (3_1) at (3,0) {};
    \node[bvertex] (4_1) at (4,0) {};
    \node[bvertex] (4_4) at (1,-6.5) {};
    \node[bvertex] (3_4) at (2,-6.5) {};
    \node[bvertex] (2_4) at (3,-6.5) {};
    \node[bvertex] (1_4) at (4,-6.5) {};
  
    \node[bvertex] (2_2) at (1,-2.25) {}; 
    \node[bvertex] (1_2) at (2,-1.25) {};
    \node[bvertex] (3_2) at (2,-2.75) {};
    \node[bvertex] (1_3) at (3,-2.25) {};
    \node[bvertex] (3_3) at (1,-4.75) {}; 
    \node[bvertex] (2_3) at (2,-4.25) {};
    \node[bvertex] (4_2) at (3,-3.75) {};
    \node[bvertex] (4_3) at (2,-5.25) {};
  
    \draw[q->] (1_1) -- (2_1);
    \draw[q->] (2_1) -- (1_2);
    \draw[q->] (1_2) -- (2_2);
    \draw[q->] (2_2) -- (1_1);
    
    \draw[q->] (1_2) -- (3_1);
    \draw[q->] (3_1) -- (1_3);
    \draw[q->] (1_3) -- (3_2);
    \draw[q->] (3_2) -- (1_2);
        
    \draw[q->] (1_3) -- (4_1);
    \draw[q->] (4_1) -- (1_4);
    \draw[q->] (1_4) -- (4_2);
    \draw[q->] (4_2) -- (1_3);
        
    \draw[q->] (2_2) -- (3_2);
    \draw[q->] (3_2) -- (2_3);
    \draw[q->] (2_3) -- (3_3);
    \draw[q->] (3_3) -- (2_2);
        
    \draw[q->] (2_3) -- (4_2);
    \draw[q->] (4_2) -- (2_4);
    \draw[q->] (2_4) -- (4_3);
    \draw[q->] (4_3) -- (2_3);
        
    \draw[q->] (3_3) -- (4_3);
    \draw[q->] (4_3) -- (3_4);
    \draw[q->] (3_4) -- (4_4);
    \draw[q->] (4_4) -- (3_3);
  \end{tikzpicture}
  \qquad
  \begin{tikzpicture}[scale=0.5, rotate=90]
    \braid[gap=0, thick, black!30] a_1 a_2 a_3 a_1 a_2 a_1;

    \node[bvertex] (e) at (0.5,-3.3) {};
    \node[bvertex] (1) at (1.5,0) {};
    \node[bvertex] (1c2) at (2.5,0) {};
    \node[bvertex] (1c2c3) at (3.5,0) {};
    \node[bvertex] (4) at (1.5,-6.6) {};
    \node[bvertex] (3c4) at (2.5,-6.6) {};
    \node[bvertex] (2c3c4) at (3.5,-6.6) {};
    \node[bvertex] (1c2c3c4) at (4.5,-3.3) {};

    \node[bvertex] (2) at (1.5,-2.25) {};
    \node[bvertex] (2c3) at (2.5,-3.5) {};
    \node[bvertex] (3) at (1.5,-4.75) {}; 
    
    \node[bvertex] (12) at (1.5,-0.75) {};
    \node[bvertex] (13) at (2.5,-1.75) {};
    \node[bvertex] (14) at (3.5,-2.75) {};
    \node[bvertex] (23) at (1.5,-3.75) {};
    \node[bvertex] (24) at (2.5,-4.75) {};
    \node[bvertex] (34) at (1.5,-5.75) {};
    
    \draw[q->] (12) -- (e);
    \draw[q->] (e) -- (1);
    \draw[q->] (1) -- (12);
    \draw[q->] (12) -- (1c2);
    \draw[q->] (2) -- (12);

    \draw[q->] (13) -- (2);
    \draw[q->] (1c2) -- (13);
    \draw[q->] (13) -- (1c2c3);
    \draw[q->] (2c3) -- (13);

    \draw[q->] (14) -- (2c3);
    \draw[q->] (1c2c3) -- (14);
    \draw[q->] (14) -- (1c2c3c4);
    \draw[q->] (1c2c3c4) -- (2c3c4);
    \draw[q->] (2c3c4) -- (14);

    \draw[q->] (23) -- (e);
    \draw[q->] (e) -- (2);
    \draw[q->] (2) -- (23);
    \draw[q->] (23) -- (2c3);
    \draw[q->] (3) -- (23);

    \draw[q->] (24) -- (3);
    \draw[q->] (2c3) -- (24);
    \draw[q->] (24) -- (2c3c4);
    \draw[q->] (2c3c4) -- (3c4);
    \draw[q->] (3c4) -- (24);

    \draw[q->] (34) -- (e);
    \draw[q->] (e) -- (3);
    \draw[q->] (3) -- (34);
    \draw[q->] (34) -- (3c4);
    \draw[q->] (3c4) -- (4);
    \draw[q->] (4) -- (34);
  \end{tikzpicture}
  \caption{The triangle, square and butterfly quivers assigned to the
    reduced word $123121$ for the longest element of $S_4$.  The
    vertex labels and the disconnected vertices are not shown.}
  \label{fig:quivers-123121}
\end{figure}

The three quivers assigned to a wiring diagram depend only on the
isotopy class of the wiring diagram.  For example, to the two wiring
diagrams
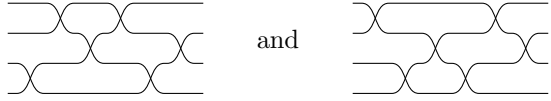
\begin{equation*}
  \tikz{\braid[gap=0, scale=0.4, rotate=90] a_1 a_3 a_2 a_3 a_1 a_2;} 
  \qquad
  \text{and}
  \qquad
  \tikz{\braid[gap=0, scale=0.4, rotate=90] a_3 a_1 a_2 a_1 a_3 a_2;}
\end{equation*}
the same quivers are assigned.  Therefore, these quivers are really
assigned to an \emph{equivalence class} of words with respect to far
commutativity \eqref{eq:s_is_j=s_js_i}.  Any two equivalence classes
of reduced words for the same permutation can be obtained from one
another by a sequence of braid relations, which translate to a
sequence of local moves in the wiring diagrams.

\begin{definition}
  The \emph{braid move} $\beta_{abc}$ on a wiring diagram is the
  local transformation
  \begin{equation}
    \begin{tikzpicture}[scale=0.4, rotate=90]
      \braid[gap=0] a_1 a_2 a_1;
      \node at (1, 0.75) {$a$};
      \node at (2, 0.75) {$b$};
      \node at (3, 0.75) {$c$};
    \end{tikzpicture}
    \quad \rightarrow \quad
    \begin{tikzpicture}[scale=0.4, rotate=90]
      \braid[gap=0] a_2 a_1 a_2;
      \node at (1, 0.75) {$a$};
      \node at (2, 0.75) {$b$};
      \node at (3, 0.75) {$c$};
    \end{tikzpicture}
  \end{equation}
  on wires $a$, $b$, $c$.
\end{definition}

The braid move $\beta_{abc}$ on a wiring diagram transforms the
assigned quivers.  A key fact is that the transformations induced by a
braid move can be expressed as compositions of mutations and
automorphisms:
\begin{proposition}
  \label{prop:R}
  The braid move $\beta_{abc}$ induces a cluster transformation on the
  triangle, square and butterfly quivers assigned to a wiring diagram.
\end{proposition}

\begin{proof}
  Since braid moves are local transformations, it is sufficient to
  check the statement for wiring diagrams on three wires.  (Although
  some arrows can be canceled when more wires are added, the action of
  $\mu_k$ depends only on $\varepsilon_{ki}$, $i \in I$, and for the
  mutations relevant to the braid move these integers remain
  unchanged.)

  For the triangle quiver, $\beta_{123}$ acts on the reduced word $121$ as
  \begin{equation}
    \begin{tikzpicture}[scale=0.8, rotate=90]
      \braid[gap=0, thick, black!30] a_1 a_2 a_1;
      \node[vertex] (e) at (0.5,-1.75) {$\emptyset$};
      \node[vertex] (1) at (1.5,0) {$\{1\}$};
      \node[vertex] (2) at (1.5,-1.75) {$\{2\}$};
      \node[vertex] (3) at (1.5,-3.5) {$\{3\}$};
      \node[vertex] (2c3) at (2.5,-2.75) {$\{2,3\}$};
      \node[vertex] (1c2) at (2.5,-0.75) {$\{1,2\}$};
      \node[vertex] (1c2c3) at (3.5,-1.75) {$\{1,2,3\}$};
      
      \draw[q->] (e) -- (1);
      \draw[q->] (1) -- (2);
      \draw[q->] (3) -- (e);
      \draw[q->] (2) -- (3);
      \draw[q->] (2) -- (1c2);
      \draw[q->] (2c3) -- (2);
      \draw[q->] (1c2) -- (2c3);
    \end{tikzpicture}
    \quad
    \xrightarrow{\beta_{123}}
    \quad
    \begin{tikzpicture}[scale=0.8, rotate=90]
      \braid[gap=0, thick, black!30] a_2 a_1 a_2;
      \node[vertex] (e) at (0.5,-1.75) {$\emptyset$};
      \node[vertex] (1) at (1.5,-0.75) {$\{1\}$};
      \node[vertex] (3) at (1.5,-2.75) {$\{3\}$};
      \node[vertex] (1c3) at (2.5,-1.75) {$\{1,3\}$};
      \node[vertex] (2c3) at (2.5,-3.5) {$\{2,3\}$};
      \node[vertex] (1c2) at (2.5,0) {$\{1,2\}$};
      \node[vertex] (1c2c3) at (3.5,-1.75) {$\{1,2,3\}$};
      
      \draw[q->] (e) -- (1);
      \draw[q->] (3) -- (e);
      \draw[q->] (1) -- (3);
      \draw[q->] (1) -- (1c2);
      \draw[q->] (1c3) -- (1);
      \draw[q->] (1c2) -- (1c3);
      \draw[q->] (3) -- (1c3);
      \draw[q->] (2c3) -- (3);
      \draw[q->] (1c3) -- (2c3);
    \end{tikzpicture}
  \end{equation}
  This transformation is the mutation $\mu_{{\{2\}}}$ at ${\{2\}}$,
  followed by relabeling of the vertex $\{2\}$ to $\{1,3\}$:
  \begin{equation}
    \beta_{123} = \auto_{\{2\}, \{1,3\}} \circ \mu_{\{2\}}
    \quad
    \text{(triangle quiver)} \,.
  \end{equation}
  Here $\auto_{i,j}$ denotes the automorphism given by the permutation
  interchanging $i$, $j \in I$.  Note that the braid move changes the
  index set for the connected vertices.  To deal with this
  complication we have enlarged the quiver by additional disconnected
  vertices.

  For the square quiver, $\beta_{123}$ acts as
  \begin{equation}
    \begin{tikzpicture}[scale=0.8, rotate=90]
      \braid[gap=0, thick, black!30] a_1 a_2 a_1;

      \node[vertex] (11) at (1,0) {$1_1$};
      \node[vertex] (12) at (2,-1.25) {$1_2$};
      \node[vertex] (13) at (3,-2.75) {$1_3$};
      
      \node[vertex] (21) at (2,0) {$2_1$};
      \node[vertex] (22) at (1,-1.75) {$2_2$}; 
      \node[vertex] (23) at (2,-3.5) {$2_3$};
      
      \node[vertex] (31) at (3,-0.75) {$3_1$};
      \node[vertex] (32) at (2,-2.25) {$3_2$};
      \node[vertex] (33) at (1,-3.5) {$3_3$}; 
      
      \draw[q->] (22) -- (11);
      \draw[q->] (12) -- (22);
      \draw[q->] (21) -- (12);
      \draw[q->] (11) -- (21);
      
      \draw[q->] (33) -- (22);
      \draw[q->] (23) -- (33);
      \draw[q->] (32) -- (23);
      \draw[q->] (22) -- (32);
      
      \draw[q->] (32) -- (12);
      \draw[q->] (13) -- (32);
      \draw[q->] (31) -- (13);
      \draw[q->] (12) -- (31);
    \end{tikzpicture}
    \quad
    \xrightarrow{\beta_{123}}
    \quad
    \begin{tikzpicture}[scale=0.8, rotate=90]
      \braid[gap=0, thick, black!30] a_2 a_1 a_2;
      
      \node[vertex] (11) at (1,-0.75) {$1_1$};
      \node[vertex] (12) at (2,-2.25) {$1_2$};
      \node[vertex] (13) at (3,-3.5) {$1_3$};
      
      \node[vertex] (21) at (2,0) {$2_1$};
      \node[vertex] (22) at (3,-1.75) {$2_2$}; 
      \node[vertex] (23) at (2,-3.5) {$2_3$};
      
      \node[vertex] (31) at (3,0) {$3_1$};
      \node[vertex] (32) at (2,-1.25) {$3_2$};
      \node[vertex] (33) at (1,-2.75) {$3_3$}; 
      
      \draw[q->] (33) -- (11);
      \draw[q->] (12) -- (33);
      \draw[q->] (32) -- (12);
      \draw[q->] (11) -- (32);
      
      \draw[q->] (32) -- (21);
      \draw[q->] (22) -- (32);
      \draw[q->] (31) -- (22);
      \draw[q->] (21) -- (31);
      
      \draw[q->] (23) -- (12);
      \draw[q->] (13) -- (23);
      \draw[q->] (22) -- (13);
      \draw[q->] (12) -- (22);
    \end{tikzpicture}
  \end{equation}
  This transformation can be written as a composition of four
  mutations and one automorphism:
  \begin{equation}
    \beta_{123}
    =
    \auto_{1_2,3_2} \circ \mu_{2_2} \circ \mu_{1_2} \circ \mu_{3_2} \circ \mu_{2_2}
    \quad
    \text{(square quiver)} \,.
  \end{equation}
  See Figure~\ref{fig:beta_123-square}.

  For the butterfly quiver, we have
  \begin{equation}
    \begin{tikzpicture}[scale=0.8, rotate=90]
      \braid[gap=0, thick, black!30] a_1 a_2 a_1;
      \node[vertex] (e) at (0.5,-1.75) {$\emptyset$};
      \node[vertex] (1) at (1.5,0.25) {$\{1\}$};
      \node[vertex] (2) at (1.5,-1.75) {$\{2\}$};
      \node[vertex] (3) at (1.5,-3.75) {$\{3\}$};
      \node[vertex] (2c3) at (2.5,-2.75) {$\{2,3\}$};
      \node[vertex] (1c2) at (2.5,-0.75) {$\{1,2\}$};
      \node[vertex] (1c2c3) at (3.5,-1.75) {$\{1,2,3\}$};

      \node[vertex] (12) at (1.5,-0.75) {$12$};
      \node[vertex] (13) at (2.5,-1.75) {$13$};
      \node[vertex] (23) at (1.5,-2.75) {$23$};
      
      \draw[q->] (12) -- (e);
      \draw[q->] (e) -- (1);
      \draw[q->] (1) -- (12);

      \draw[q->] (12) -- (1c2);
      \draw[q->] (2) -- (12);

      \draw[q->] (13) -- (2);
      \draw[q->] (1c2) -- (13);

      \draw[q->] (13) -- (1c2c3);
      \draw[q->] (1c2c3) -- (2c3);
      \draw[q->] (2c3) -- (13);

      \draw[q->] (23) -- (e);
      \draw[q->] (e) -- (2);
      \draw[q->] (2) -- (23);

      \draw[q->] (23) -- (2c3);
      \draw[q->] (2c3) -- (3);
      \draw[q->] (3) -- (23);
    \end{tikzpicture}
    \quad
    \xrightarrow{\beta_{123}}
    \quad
    \begin{tikzpicture}[scale=0.8, rotate=90]
      \braid[gap=0, thick, black!30] a_2 a_1 a_2;
      \node[vertex] (e) at (0.5,-1.75) {$\emptyset$};
      \node[vertex] (1) at (1.5,-0.75) {$\{1\}$};
      \node[vertex] (3) at (1.5,-2.75) {$\{3\}$};
      \node[vertex] (1c3) at (2.5,-1.75) {$\{1,3\}$};
      \node[vertex] (2c3) at (2.5,-3.75) {$\{2,3\}$};
      \node[vertex] (1c2) at (2.5,0.25) {$\{1,2\}$};
      \node[vertex] (1c2c3) at (3.5,-1.75) {$\{1,2,3\}$};

      \node[vertex] (23) at (2.5,-0.75) {$23$};
      \node[vertex] (13) at (1.5,-1.75) {$13$};
      \node[vertex] (12) at (2.5,-2.75) {$12$};

      \draw[q->] (23) -- (1);
      \draw[q->] (1) -- (1c2);
      \draw[q->] (1c2) -- (23);

      \draw[q->] (23) -- (1c2c3);
      \draw[q->] (1c2c3) -- (1c3);
      \draw[q->] (1c3) -- (23);

      \draw[q->] (13) -- (e);
      \draw[q->] (e) -- (1);
      \draw[q->] (1) -- (13);

      \draw[q->] (13) -- (1c3);
      \draw[q->] (3) -- (13);

      \draw[q->] (12) -- (3);
      \draw[q->] (1c3) -- (12);

      \draw[q->] (12) -- (1c2c3);
      \draw[q->] (1c2c3) -- (2c3);
      \draw[q->] (2c3) -- (12);
    \end{tikzpicture}
  \end{equation}
  This transformation can be written as
  \begin{equation}
    \beta_{123}
    =
    \auto_{\{2\}, \{1,3\}} \circ \auto_{\{2\}, 13} \circ \auto_{12,23}
    \circ \mu_{13} \circ \mu_{23} \circ \mu_{12} \circ \mu_{\{2\}}
    \quad
    \text{(butterfly quiver)} \,.
  \end{equation}
\end{proof}

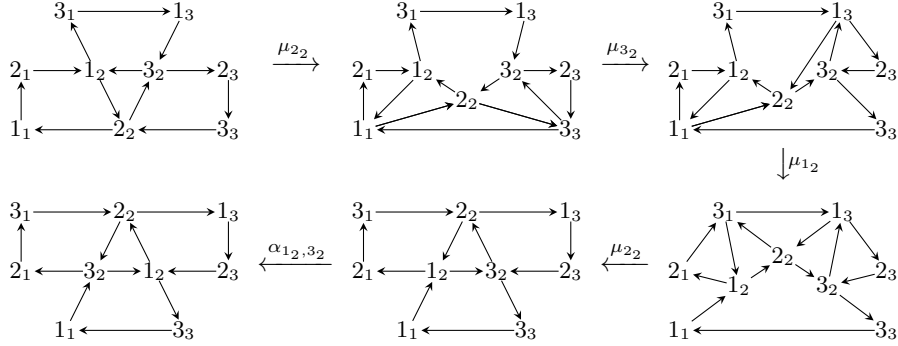
\begin{figure}
  \centering
  \begin{equation*}
    \begin{array}{c@{\ \ }c@{\ \ }c@{\ \ }c@{\ \ }c}
      \begin{tikzpicture}[scale=0.8, rotate=90]
        \node[vertex] (11) at (1,0) {$1_1$};
        \node[vertex] (12) at (2,-1.25) {$1_2$};
        \node[vertex] (13) at (3,-2.75) {$1_3$};
        
        \node[vertex] (21) at (2,0) {$2_1$};
        \node[vertex] (22) at (1,-1.75) {$2_2$}; 
        \node[vertex] (23) at (2,-3.5) {$2_3$};
        
        \node[vertex] (31) at (3,-0.75) {$3_1$};
        \node[vertex] (32) at (2,-2.25) {$3_2$};
        \node[vertex] (33) at (1,-3.5) {$3_3$}; 
        
        \draw[q->] (22) -- (11);
        \draw[q->] (12) -- (22);
        \draw[q->] (21) -- (12);
        \draw[q->] (11) -- (21);
        
        \draw[q->] (33) -- (22);
        \draw[q->] (23) -- (33);
        \draw[q->] (32) -- (23);
        \draw[q->] (22) -- (32);
        
        \draw[q->] (32) -- (12);
        \draw[q->] (13) -- (32);
        \draw[q->] (31) -- (13);
        \draw[q->] (12) -- (31);
      \end{tikzpicture}
      &
      \xrightarrow{\mu_{2_2}}
      &
      \begin{tikzpicture}[scale=0.8, rotate=90]
        \node[vertex] (11) at (1,0) {$1_1$};
        \node[vertex] (12) at (2,-1) {$1_2$};
        \node[vertex] (13) at (3,-2.75) {$1_3$};
        
        \node[vertex] (21) at (2,0) {$2_1$};
        \node[vertex] (22) at (1.5,-1.75) {$2_2$}; 
        \node[vertex] (23) at (2,-3.5) {$2_3$};
        
        \node[vertex] (31) at (3,-0.75) {$3_1$};
        \node[vertex] (32) at (2,-2.5) {$3_2$};
        \node[vertex] (33) at (1,-3.5) {$3_3$}; 
        
        \draw[q->] (33) -- (11);
        
        \draw[q->] (11) -- (22);
        \draw[q->] (12) -- (11);
        \draw[q->] (32) -- (22);
        \draw[q->] (33) -- (32);
        \draw[q->] (22) -- (33);
        
        \draw[q->] (11) -- (22);
        \draw[q->] (22) -- (12);
        \draw[q->] (21) -- (12);
        \draw[q->] (11) -- (21);
        
        \draw[q->] (22) -- (33);
        \draw[q->] (23) -- (33);
        \draw[q->] (32) -- (23);
        
        \draw[q->] (13) -- (32);
        \draw[q->] (31) -- (13);
        \draw[q->] (12) -- (31);
      \end{tikzpicture}
      &
      \xrightarrow{\mu_{3_2}}
      &
      \begin{tikzpicture}[scale=0.8, rotate=90]
        \node[vertex] (11) at (1,0) {$1_1$};
        \node[vertex] (12) at (2,-1) {$1_2$};
        \node[vertex] (13) at (3,-2.75) {$1_3$};
        
        \node[vertex] (21) at (2,0) {$2_1$};
        \node[vertex] (22) at (1.5,-1.75) {$2_2$}; 
        \node[vertex] (23) at (2,-3.5) {$2_3$};
        
        \node[vertex] (31) at (3,-0.75) {$3_1$};
        \node[vertex] (32) at (2,-2.5) {$3_2$};
        \node[vertex] (33) at (1,-3.5) {$3_3$}; 
        
        \draw[q->] (33) -- (11);
        
        \draw[q->] (11) -- (22);
        \draw[q->] (12) -- (11);
        \draw[q->] (22) -- (32);
        \draw[q->] (32) -- (33);
        
        \draw[q->] (11) -- (22);
        \draw[q->] (22) -- (12);
        \draw[q->] (21) -- (12);
        \draw[q->] (11) -- (21);
        
        \draw[q->] (23) -- (32);
        
        \draw[q->] (32) -- (13);
        \draw[q->] (31) -- (13);
        \draw[q->] (12) -- (31);
        \draw[q->] (13) -- (23);
        \draw[q->] (13) -- (22);
      \end{tikzpicture}
      \\[6ex]
      &&&& \phantom{\hspace{-0.15em}\scriptstyle \mu_{1_2}}
           \big\downarrow{\raisebox{0.5ex}{$\hspace{-0.15em}\scriptstyle \mu_{1_2}$}}
      \\[2ex]
    \begin{tikzpicture}[scale=0.8, rotate=90]
      \node[vertex] (11) at (1,-0.75) {$1_1$};
      \node[vertex] (12) at (2,-2.25) {$1_2$};
      \node[vertex] (13) at (3,-3.5) {$1_3$};
      
      \node[vertex] (21) at (2,0) {$2_1$};
      \node[vertex] (22) at (3,-1.75) {$2_2$}; 
      \node[vertex] (23) at (2,-3.5) {$2_3$};
      
      \node[vertex] (31) at (3,0) {$3_1$};
      \node[vertex] (32) at (2,-1.25) {$3_2$};
      \node[vertex] (33) at (1,-2.75) {$3_3$}; 
      
      \draw[q->] (33) -- (11);
      \draw[q->] (12) -- (33);
      \draw[q->] (32) -- (12);
      \draw[q->] (11) -- (32);
      
      \draw[q->] (32) -- (21);
      \draw[q->] (22) -- (32);
      \draw[q->] (31) -- (22);
      \draw[q->] (21) -- (31);
      
      \draw[q->] (23) -- (12);
      \draw[q->] (13) -- (23);
      \draw[q->] (22) -- (13);
      \draw[q->] (12) -- (22);
    \end{tikzpicture}
      &
        \xleftarrow{\auto_{1_2,3_2}}
      &
    \begin{tikzpicture}[scale=0.8, rotate=90]
      \node[vertex] (11) at (1,-0.75) {$1_1$};
      \node[vertex] (32) at (2,-2.25) {$3_2$};
      \node[vertex] (13) at (3,-3.5) {$1_3$};
      
      \node[vertex] (21) at (2,0) {$2_1$};
      \node[vertex] (22) at (3,-1.75) {$2_2$}; 
      \node[vertex] (23) at (2,-3.5) {$2_3$};
      
      \node[vertex] (31) at (3,0) {$3_1$};
      \node[vertex] (12) at (2,-1.25) {$1_2$};
      \node[vertex] (33) at (1,-2.75) {$3_3$}; 
      
      \draw[q->] (33) -- (11);
      \draw[q->] (32) -- (33);
      \draw[q->] (12) -- (32);
      \draw[q->] (11) -- (12);
      
      \draw[q->] (12) -- (21);
      \draw[q->] (22) -- (12);
      \draw[q->] (31) -- (22);
      \draw[q->] (21) -- (31);
      
      \draw[q->] (23) -- (32);
      \draw[q->] (13) -- (23);
      \draw[q->] (22) -- (13);
      \draw[q->] (32) -- (22);
    \end{tikzpicture}
      &
        \xleftarrow{\mu_{2_2}}
      &
    \begin{tikzpicture}[scale=0.8, rotate=90]
      \node[vertex] (11) at (1,0) {$1_1$};
      \node[vertex] (12) at (1.75,-1) {$1_2$};
      \node[vertex] (13) at (3,-2.75) {$1_3$};
      
      \node[vertex] (21) at (2,0) {$2_1$};
      \node[vertex] (22) at (2.25,-1.75) {$2_2$}; 
      \node[vertex] (23) at (2,-3.5) {$2_3$};
      
      \node[vertex] (31) at (3,-0.75) {$3_1$};
      \node[vertex] (32) at (1.75,-2.5) {$3_2$};
      \node[vertex] (33) at (1,-3.5) {$3_3$}; 
      
      \draw[q->] (33) -- (11);

      \draw[q->] (11) -- (12);
      \draw[q->] (22) -- (32);
      \draw[q->] (32) -- (33);

      \draw[q->] (12) -- (22);
      \draw[q->] (12) -- (21);
      
      \draw[q->] (23) -- (32);
      
      \draw[q->] (32) -- (13);
      \draw[q->] (31) -- (13);
      \draw[q->] (31) -- (12);
      \draw[q->] (13) -- (23);
      \draw[q->] (13) -- (22);

      \draw[q->] (21) -- (31);
      \draw[q->] (22) -- (31);
    \end{tikzpicture}
    \end{array}
  \end{equation*}
  \caption{The braid move $\beta_{123}$ for the square quiver is a
    composition of four mutations and one automorphism.}
  \label{fig:beta_123-square}
\end{figure}

\begin{remark}
  In the context of supersymmetric gauge theories and the Yang--Baxter
  equation, the triangle quiver appeared in \cite{Yagi:2015lha}, the
  square quiver in \cite{Bazhanov:2011mz, Yagi:2015lha} and the
  butterfly quiver in \cite{Yamazaki:2015voa}.  The triangle quiver
  has also appeared in relation to cluster algebras before, e.g.\
  in~\cite{FWZ}.  A connection between the butterfly quiver and
  cluster algebras was pointed out in \cite{Yamazaki:2016wnu}.
\end{remark}

\subsection{Trivial cluster transformations}

Take the sequence of braid moves on the left-hand side of
Figure~\ref{fig:TE} and concatenate it with the inverse of the
sequence of braid moves on the right-hand side.  This creates a
sequence of transformations on reduced words for the longest element
of $S_4$ which starts with $123121$ and ends with $121321$.  Let
$\seed_{a_1 a_2 a_3 a_4 a_5 a_6}$ be any of the triangle, square and
butterfly quivers assigned to the reduced word
$a_1 a_2 a_3 a_4 a_5 a_6$ for this element.  By
Proposition~\ref{prop:R}, the sequence of transformations in question
induces the loop of cluster transformations
\begin{equation}
  \newcommand{\verteq}{\rotatebox[origin=c]{90}{$=$}}
  \begin{array}{c@{\hspace{0.278em}}c@{\hspace{0.278em}}c@{\hspace{0.278em}}c@{\hspace{0.278em}}c@{\hspace{0.278em}}c@{\hspace{0.278em}}c@{\hspace{0.278em}}c@{\hspace{0.278em}}c@{\hspace{0.278em}}c@{\hspace{0.278em}}c}
    \seed_{1 2 3 1 2 1}
    &  \xrightarrow{\beta_{234}} 
    & \seed_{1 2 3 2 1 2}
    &  \xrightarrow{\beta_{134}} 
    & \seed_{1 3 2 3 1 2}
    & =
    & \seed_{3 1 2 1 3 2}
    &  \xrightarrow{\beta_{124}} 
    & \seed_{3 2 1 2 3 2}
    &  \xrightarrow{\beta_{123}} 
    & \seed_{3 2 1 3 2 3}
    \\
    \verteq  &&&&&&&&&& \verteq
    \\
    \seed_{1 2 1 3 2 1}
    &  \xleftarrow{\beta_{123}^{-1}} 
    & \seed_{2 1 2 3 2 1}
    &  \xleftarrow{\beta_{124}^{-1}} 
    & \seed_{2 1 3 2 3 1}
    & =
    & \seed_{2 3 1 2 1 3}
    &  \xleftarrow{\beta_{134}^{-1}} 
    & \seed_{2 3 2 1 2 3}
    &  \xleftarrow{\beta_{234}^{-1}} 
    & \seed_{3 2 3 1 2 3}
  \end{array}
\end{equation}
This loop of cluster transformation turns out to be trivial:

\begin{proposition}
  \label{prop:trivial-loop-ct}
  For the triangle, square and butterfly quivers assigned to the
  reduced word $123121$ for the longest element of $S_4$, the cluster
  transformation
  \begin{equation}
    \label{eq:RRRRRRRR}
    \beta_{123}^{-1} \circ \beta_{124}^{-1} \circ \beta_{134}^{-1} \circ \beta_{234}^{-1}
    \circ \beta_{123} \circ \beta_{124} \circ \beta_{134} \circ \beta_{234}
  \end{equation}
  is trivial.
\end{proposition}

The proposition can be proved by Theorem~\ref{eq:trivial-trop} and
straightforward calculations.  The proof is given in
Appendix~\ref{sec:proof}.

\section{Solutions of the tetrahedron equation from trivial cluster
  transformations}
\label{sec:TE}

Now we discuss the solutions of the tetrahedron equation arising from
the three trivial cluster transformations obtained in
section~\ref{sec:trivial-ct}.

\subsection{The trivial cluster transformations and the tetrahedron
  equation}

The triviality of the cluster transformation~\eqref{eq:RRRRRRRR} can
be expressed as the equality
\begin{equation}
  \label{eq:TE-beta}
  \beta_{123} \circ \beta_{124} \circ \beta_{134} \circ \beta_{234}
  =
  \beta_{234} \circ \beta_{134} \circ \beta_{124} \circ \beta_{123}
\end{equation}
between two cluster transformations from the seed tori for
$\seed_{123121} = \seed_{121321}$ to those for
$\seed_{321323} = \seed_{323123}$.  This equation takes the form of
the tetrahedron equation.

By Proposition~\ref{prop:trivial-cq} the corresponding quantum cluster
transformations for the same quivers also satisfy the tetrahedron
equation
\begin{equation}
  \beta^q_{234} \circ \beta^q_{134} \circ \beta^q_{124} \circ \beta^q_{123}
  =
  \beta^q_{123} \circ \beta^q_{124} \circ \beta^q_{134} \circ \beta^q_{234} \,.
\end{equation}
Furthermore, Proposition \ref{prop:trivial-K} and the tetrahedron
equation~\eqref{eq:TE-beta} imply that the unitary operator
\begin{equation}
  R_{abc} = \K{\beta_{abc}}
\end{equation}
solves the tetrahedron equation \eqref{eq:TE}.  This is an equality
between operators from the Hilbert space $\CH_{\seed_{123121}}$ to the
Hilbert space $\CH_{\seed_{321323}}$.

\begin{remark}
  For the triangle quiver, the tropical sign sequence
  $(\epsilon[t])_{t=1}^8$ is given by $(+,+,+,+,-,-,-,-)$ and one can
  use simpler quantum dilogarithm identities (Eqs.~(1.8) and (1.10) of
  \cite{MR2861174}) to construct solutions of the tetrahedron
  equation.
\end{remark}

\subsection{R-matrices and three-dimensional integrable lattice
  models}

The R-matrix $R_{123}$ may be regarded as the local Boltzmann weight
of a three-dimensional statistical lattice model with continuous spin
variables.  Let us explain this point taking the square quiver case as
an example.

Consider the cluster transformation
$\beta_{123}\colon \seed_{121} \to \seed_{212}$ between the square
quivers assigned to the reduced words $121$ and $212$ for the longest
element of $S_3$.  The vertices of the square quivers are placed on
the segments of wires, hence
\begin{equation}
  I = \{1_1, 1_2, 1_3, 2_1, 2_2, 2_3, 3_1, 3_2, 3_3\}
\end{equation}
is the set labeling the vertices of $\seed_{121}$ and $\seed_{121}$.
The quantum variables $(X^q, B^q)$ assigned to $\seed_{121}$ and
$(X'^q, B'^q)$ assigned to $\seed_{212}$ have representations on the
Hilbert spaces $\CH_{\seed_{121}} = L^2(\R^I)$ and the Hilbert space
$\CH_{\seed_{212}} = L^2(\R^I)$, respectively.  Let
$a = (a_i)_{i \in I}$ and $a' = (a'_i)_{i \in I}$ be the standard
coordinates for functions in $\CH_{\seed_{121}}$ and
$\CH_{\seed_{212}}$.  Then, the R-matrix
$R_{123}\colon \CH_{\seed_{212}} \to \CH_{\seed_{121}}$ can be
expressed as an integral operator:
\begin{equation}
  (R_{123} f)(a)
  =
  \int_{\R^I} \rmd a' \, S_{123}(a,a') f(a') \,.
\end{equation}
Here $\rmd a' := \prod_{i \in I} \rmd a'_i$.  Since
$\K{\beta_{123}}(a'_i) = a_i$ unless $\beta_{123}$ contains the
mutation $\mu_i$ or an automorphism acting nontrivially on $i$, the
kernel $S_{123}(a,a')$ takes the form
\begin{equation}
  S_{123}(a,a')
  =
  S \!
  \begin{bmatrix}
    a_{1_1} & a_{1_2} & a_{1_3} & a'_{1_2} \\
    a_{2_1} & a_{2_2} & a_{2_3} & a'_{2_2} \\
    a_{3_1} & a_{3_2} & a_{3_3} & a'_{3_2}
  \end{bmatrix}
  \prod_{i \in I \setminus \{1_2, 2_2, 3_2\}}
  \delta(a_i - a'_i) \,,
\end{equation}
where $S$ is a function of the indicated variables and $\delta$
is the delta function.

In the lattice model, $R_{123}$ resides at an intersection of three
planes constituting part of a lattice, as shown in
Figure~\ref{fig:lattice}.  Each vertex of $\seed_{121}$ and
$\seed_{212}$ corresponds to one quadrant of one of the three planes,
separated from the other quadrants by the intersections with the other
two planes.  To construct a statistical lattice model, for each vertex
$i \in I$ we place an $\R$-valued spin variable $a_i$ or $a'_i$ on the
corresponding region, depending on whether the vertex is from
$\seed_{121}$ or $\seed_{212}$, and identify $a_i = a'_i$ for
$i \neq 1_2$, $2_2$, $3_2$.  Given a configuration of the 12 spin
variables thus prepared, we define the local ``energy''%
\footnote{The quantity $E$ is not real in the present case and cannot
  be interpreted as a physical energy.}
$E$ for that configuration by $S = \exp(-E/k_B T)$, where $k_B$
is the Boltzmann constant and $T$ is the temperature.

The tetrahedron equation \eqref{eq:TE} translates to the equation
\begin{align}
  \int_{\R^4} \rmd a'_{1_3} \, \rmd a'_{2_3} \, \rmd a'_{3_3} \, \rmd a'_{4_3} \,
  &S \!
  \begin{bmatrix}
    a_{2_2} & a_{2_3} & a_{2_4} & a'_{2_3} \\
    a_{3_2} & a_{3_3} & a_{3_4} & a'_{3_3} \\
    a_{4_2} & a_{4_3} & a_{4_4} & a'_{4_3}
  \end{bmatrix}
  S \!
  \begin{bmatrix}
    a_{1_2} & a_{1_3} & a_{1_4} & a'_{1_3} \\
    a_{3_1} & a_{3_2} & a'_{3_3} & a''_{3_2} \\
    a_{4_1} & a_{4_2} & a'_{4_3} & a''_{4_2}
  \end{bmatrix}
  \\ \nonumber
  \times
  &S \!
  \begin{bmatrix}
    a_{1_1} & a_{1_2} & a'_{1_3} & a''_{1_2} \\
    a_{2_1} & a_{2_2} & a'_{2_3} & a''_{2_2} \\
    a''_{4_2} & a'_{4_3} & a_{4_4} & a''_{4_3}
  \end{bmatrix}
  S \!
  \begin{bmatrix}
    a''_{1_2} & a'_{1_3} & a_{1_4} & a''_{1_3} \\
    a''_{2_2} & a'_{2_3} & a_{2_4} & a''_{2_3} \\
    a''_{3_2} & a'_{3_3} & a_{3_4} & a''_{3_3}
  \end{bmatrix}
  \\ \nonumber
  =
  \int_{\R^4} \rmd a'_{1_2} \, \rmd a'_{2_2} \, \rmd a'_{3_2} \, \rmd a'_{4_2} \,
  &S \!
  \begin{bmatrix}
    a_{1_1} & a_{1_2} & a_{1_3} & a'_{1_2} \\
    a_{2_1} & a_{2_2} & a_{2_3} & a'_{2_2} \\
    a_{3_1} & a_{3_2} & a_{3_3} & a'_{3_2}
  \end{bmatrix}
  S \!
  \begin{bmatrix}
    a'_{1_2} & a_{1_3} & a_{1_4} & a''_{1_3} \\
    a'_{2_2} & a_{2_3} & a_{2_4} & a''_{2_3} \\
    a_{4_1} & a_{4_2} & a_{4_3} & a'_{4_2}
  \end{bmatrix}
  \\ \nonumber
  \times
  &S \!
  \begin{bmatrix}
    a_{1_1} & a'_{1_2} & a''_{1_3} & a''_{1_2} \\
    a'_{3_2} & a_{3_3} & a_{3_4} & a''_{3_3} \\
    a'_{4_2} & a_{4_3} & a_{4_4} & a''_{4_3}
  \end{bmatrix}
  S \!
  \begin{bmatrix}
    a_{2_1} & a'_{2_2} & a''_{2_3} & a''_{2_2} \\
    a_{3_1} & a'_{3_2} & a''_{3_3} & a''_{3_2} \\
    a_{4_1} & a'_{4_2} & a''_{4_3} & a''_{4_2}
  \end{bmatrix}
  \,.
\end{align}
This is the form of the tetrahedron equation appropriate for the local
Boltzmann weight of a three-dimensional lattice model with spin
variables placed on the faces of the lattice \cite{Hietarinta:1994pt,
  MR1278735}.

\begin{remark}
  The Boltzmann weights for the lattice models corresponding to the
  triangle and butterfly quivers can be determined in the same
  fashion.  The triangle quiver model has spin variables placed inside
  the regions bounded by the planes making up the cubic lattice, and
  the butterfly quiver model has spin variables in those regions as
  well as the edges of the lattice.  The tetrahedron equation for the
  former can be found in \cite{Hietarinta:1994pt, MR1278735}.
\end{remark}

The fact that the local Boltzmann weight solves the tetrahedron
equation is closely related to the integrability of the lattice model.

In Figure~\ref{fig:lattice}, let us make the direction common to the
first and second planes periodic and replace the third plane with a
stack of $n$ parallel planes.  The resulting configuration of planes
can be thought of as defining a local Boltzmann weight $W$ for a
two-dimensional lattice model, which is obtained from the
three-dimensional lattice model by ``dimensional reduction'' on the
periodic direction.  Schematically, we can write $W$ as
$W = \Tr(S^n)$, where each factor of $S$ represents one of the $n$
parallel planes.  The tetrahedron equation implies that $W$ satisfies
the Yang--Baxter equation.

For a solution of the Yang--Baxter equation to define a
two-dimensional integrable lattice model, it must also depend on a
continuous parameter, called a spectral parameter.  A spectral
parameter can be introduced to $W$ by ``twisting'' of the periodic
boundary condition.  Let
\begin{equation}
  q_{12'} := a'_{1_2} - a_{2_3} \,,
  \qquad
  q_{12} := a_{1_1} - a_{2_2} \,.
\end{equation}
We can regard $q_{12'}$ and $q_{12}$ as ``charges'' assigned to the
edges $12'$ and $12$ of the lattice.  As we will see shortly, these
charges are conserved: the local Boltzmann weight vanishes unless
$q_{12} = q_{12'}$.  As a result, $W(z) = \Tr(z^{q_{12}} S^n)$ defines
a solution of the Yang--Baxter equation with spectral parameter $z$.
The integrability of the corresponding two-dimensional lattice model
implies the integrability of the original three-dimensional lattice
model.

\begin{remark}
  A similar dimensional reduction for the solution of \cite{MR1278735,
    Bazhanov:2005as} was considered in \cite{Bazhanov:2005as}, where
  it was shown that the resulting solution of the Yang--Baxter
  equation is a trigonometric R-matrix associated with the direct sum
  of symmetric tensor representations of $\mathfrak{sl}_n$.  In
  \cite{Yagi:2022tot}, this reduction was interpreted in string theory
  as a duality transformation to a brane configuration studied in
  \cite{Costello:2018txb, Ishtiaque:2021jan}.
\end{remark}

\begin{figure}
  \centering
  \begin{tikzpicture}[
    /pgf/braid/.cd,
    style strands={1}{red!50!black, opacity=0.4},
    style strands={2}{green!50!black, opacity=0.4},
    style strands={3}{blue!50!black, opacity=0.4},
    number of strands=3
  ]
    \begin{scope}[xscale=0.8, yscale=0.4, rotate=90]
      \braid[gap=0, thick] a_1 a_2 a_1;

      \node[font=\small, vertex] (1_1) at (1,0) {$1_1$};
      \node[font=\small, vertex] (1_2) at (2,-1.25) {$1_2$};
      \node[font=\small, vertex] (1_3) at (3,-3.5) {$1_3$};
      
      \node[font=\small, vertex] (2_1) at (2,0) {$2_1$};
      \node[font=\small, vertex] (2_2) at (1,-1.75) {$2_2$}; 
      \node[font=\small, vertex] (2_3) at (2,-3.5) {$2_3$};
      
      \node[font=\small, vertex] (3_1) at (3,0) {$3_1$};
      \node[font=\small, vertex] (3_2) at (2,-2.25) {$3_2$};
      \node[font=\small, vertex] (3_3) at (1,-3.5) {$3_3$};
      
      \node[font=\small, vertex] (12) at (1.5,-0.75) {$12$};
      \node[font=\small, vertex] (13) at (2.5,-1.75) {$13$};
      \node[font=\small, vertex] (23) at (1.5,-2.75) {$23$};

      \begin{scope}[shift={(-5,0)}]
        \braid[gap=0, thick] a_2 a_1 a_2;

        \node[font=\small, vertex] (1_1) at (1,0) {$1_1'$};
        \node[font=\small, vertex] (1_2) at (2,-1.25) {$3_2'$};
        \node[font=\small, vertex] (1_3) at (3,-3.5) {$1_3'$};
        
        \node[font=\small, vertex] (2_1) at (2,0) {$2_1'$};
        \node[font=\small, vertex] (2_2) at (3,-1.75) {$2_2'$}; 
        \node[font=\small, vertex] (2_3) at (2,-3.5) {$2_3'$};
        
        \node[font=\small, vertex] (3_1) at (3,0) {$3_1'$};
        \node[font=\small, vertex] (3_2) at (2,-2.25) {$1_2'$};
        \node[font=\small, vertex] (3_3) at (1,-3.5) {$3_3'$};
        
        \node[font=\small, vertex] (12) at (2.5,-0.75) {$23'$};
        \node[font=\small, vertex] (13) at (1.5,-1.75) {$13'$};
        \node[font=\small, vertex] (23) at (2.5,-2.75) {$12'$};
      \end{scope}

      \node[black] at (2,1) {$121$};
      \node[black] at (-0.5,1) {${\raisebox{0.5ex}{$\hspace{-0.15em}\scriptstyle \beta_{123}$}}\Bigg\downarrow\phantom{\hspace{-0.15em}\scriptstyle \beta_{123}}$};
      \node[black] at (-3,1) {$212$};

      \node[black] at (2,-4.5) {$\CH_{\seed_{121}}$};
      \node[black] at (-0.5,-4.5) {$\phantom{\hspace{-0.15em}\scriptstyle R_{123}}\Bigg\uparrow{\raisebox{0.5ex}{$\hspace{-0.15em}\scriptstyle R_{123}$}}$};
      \node[black] at (-3,-4.5) {$\CH_{\seed_{212}}$};
    \end{scope}
  \end{tikzpicture}
  \quad
  \tdplotsetmaincoords{70}{120}
  \begin{tikzpicture}[tdplot_main_coords, scale=0.8]
    \draw[->, >=latex, thick]
    (-2,0,0) node[shift={(7pt,5pt)}] {$23'$}
    -- (2,0,0) node[shift={(-6pt,-4pt)}] {$23$};

    \draw[->, >=latex, thick]
    (0,-2,0) node[left, shift={(4pt,4pt)}] {$13'$}
    -- (0,2,0) node[right, shift={(-2pt,-2pt)}] {$13$};

    \draw[->, >=latex, thick]
    (0,0,-2) node[shift={(0pt,-6pt)}] {$12'$}
    -- (0,0,2) node[shift={(0pt,6pt)}] {$12$};
    
    \draw[fill=red, opacity=0.2]
    (0,-2,-2) node[shift={(6pt,4pt)}, opacity=1] {$1_2'$}
    -- (0,2,-2) node[shift={(-6pt,8pt)}, opacity=1] {$1_3^{(\prime)}$}
    -- (0,2,2) node[shift={(-6pt,-4pt)}, opacity=1] {$1_2$}
    -- (0,-2,2) node[shift={(8pt,-9pt)}, opacity=1] {$1_1^{(\prime)}$}
    -- cycle;
    \draw[black!50, thick] (0,-2,-2) -- (0,2,-2) -- (0,2,2) -- (0,-2,2) -- cycle;

    \draw[fill=green, opacity=0.2]
    (-2,0,-2) node[shift={(-5pt,4pt)}, opacity=1] {$2_2'$}
    -- (-2,0,2) node[shift={(-6pt,-11pt)}, opacity=1] {$2_1^{(\prime)}$}
    -- (2,0,2) node[shift={(6pt,-4pt)}, opacity=1] {$2_2$}
    -- (2,0,-2) node[shift={(8pt,12pt)}, opacity=1] {$2_3^{(\prime)}$}
    -- cycle;
    \draw[black!50, thick] (-2,0,-2) -- (-2,0,2) -- (2,0,2) -- (2,0,-2) -- cycle;

    \draw[fill=blue, opacity=0.2]
    (-2,-2,0) node[shift={(4pt,-7pt)}, opacity=1] {$3_2'$}
    -- (2,-2,0) node[shift={(20pt,4pt)}, opacity=1] {$3_3^{(\prime)}$}
    -- (2,2,0) node[shift={(-2pt,8pt)}, opacity=1] {$3_2$}
    -- (-2,2,0) node[shift={(-20pt,-4pt)}, opacity=1] {$3_1^{(\prime)}$}
    -- cycle;
    \draw[black!50, thick] (-2,-2,0) -- (2,-2,0) -- (2,2,0) -- (-2,2,0) -- cycle;
  \end{tikzpicture}

  \caption{The braid move
    $\beta_{123}\colon \seed_{121} \to \seed_{212}$ and the
    corresponding three-dimensional lattice.  The vertices $i_a$ and
    $i_a'$ are identified for $i = 1$, $2$, $3$ and $a = 1$, $3$.}
  \label{fig:lattice}
\end{figure}
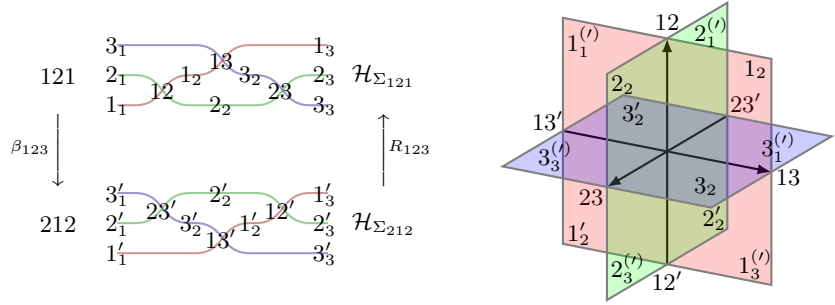

\subsection{R-matrix for the square quiver}

Let us calculate an explicit form of the function $S$ for the square
quiver.  The R-matrix
\begin{equation}
  R_{123}
  =
  \K{\mu_{2_2}}^{\sharp(\epsilon[1])}
  \K{\mu_{2_2}}^{\flat(\epsilon[1])}
  \K{\mu_{3_2}}^{\sharp(\epsilon[2])}
  \K{\mu_{3_2}}^{\flat(\epsilon[2])}
  \K{\mu_{1_2}}^{\sharp(\epsilon[3])}
  \K{\mu_{1_2}}^{\flat(\epsilon[3])}
  \K{\mu_{2_2}}^{\sharp(\epsilon[4])}
  \K{\mu_{2_2}}^{\flat(\epsilon[4])}
  \K{\auto_{1_2,3_2}}
\end{equation}
is independent of the choice of signs
$(\epsilon[1], \epsilon[2], \epsilon[3], \epsilon[4]) \in \{\pm\}^4$,
and different choices lead to different expressions for $R_{123}$.
Let us take them to be the tropical signs, which are all $+$.  Then,
\begin{equation}
  \begin{split}
    R_{123}
    &=
    \Phi^\hbar(\gamma[1] \cdot \xop)
    \Phi^\hbar(\gamma[2] \cdot \xop)
    \Phi^\hbar(\gamma[3] \cdot \xop)
    \Phi^\hbar(\gamma[4] \cdot \xop)
    \\
    & \quad
    \cdot
    \Phi^\hbar(\gamma[1] \cdot \xtop)^{-1}
    \Phi^\hbar(\gamma[2] \cdot \xtop)^{-1}
    \Phi^\hbar(\gamma[3] \cdot \xtop)^{-1}
    \Phi^\hbar(\gamma[4] \cdot \xtop)^{-1}
    \\
    &\quad\cdot
    \K{\mu_{2_2}}^{\flat(+)}
    \K{\mu_{3_2}}^{\flat(+)}
    \K{\mu_{1_2}}^{\flat(+)}
    \K{\mu_{2_2}}^{\flat(+)}
    \K{\auto_{1_2,3_2}}
  \end{split}
\end{equation}
The transformations of the tropical $\CX$-variables is listed in
Table~\ref{tab:y-square}.  According to the table we have
\begin{equation}
  \gamma[1] \cdot \xop = \xop_{2_2} \,,
  \quad
  \gamma[2] \cdot \xop = \xop_{2_2} + \xop_{3_2} \,,
  \quad
  \gamma[3] \cdot \xop =  \xop_{1_2} \,,
  \quad
  \gamma[4] \cdot \xop =  \xop_{3_2} \,.
\end{equation}
Expressing the relation between $\xop[1]$ and $\xop[6]$ in the matrix
form and taking the inverse transpose matrix, we find that
$\K{\mu_{2_2}}^{\flat(+)} \K{\mu_{3_2}}^{\flat(+)}
\K{\mu_{1_2}}^{\flat(+)} \K{\mu_{2_2}}^{\flat(+)} \K{\auto_{1_2,3_2}}$
maps
\begin{equation}
  a_{1_2}' \mapsto a_{1_1} - a_{2_2} + a_{2_3} \,,
  \quad
  a_{2_2}' \mapsto a_{2_3} + a_{3_1} - a_{3_2} \,,
  \quad
  a_{3_2}' \mapsto a_{1_1} - a_{1_2} + a_{3_1} \,.
\end{equation}
Thus we obtain
\begin{align}
  &S \!
  \begin{bmatrix}
    a_{1_1} & a_{1_2} & a_{1_3} & a'_{1_2} \\
    a_{2_1} & a_{2_2} & a_{2_3} & a'_{2_2} \\
    a_{3_1} & a_{3_2} & a_{3_3} & a'_{3_2}
  \end{bmatrix}
  \\ \nonumber
  &
    =
  \frac{
    \Phi^\hbar(\xop_{2_2})
    \Phi^\hbar(\xop_{2_2} + \xop_{3_2})
    \Phi^\hbar(\xop_{1_2})
    \Phi^\hbar(\xop_{3_2})
  }{
    \Phi^\hbar(\xtop_{2_2})
    \Phi^\hbar(\xtop_{2_2} + \xtop_{3_2})
    \Phi^\hbar(\xtop_{1_2})
    \Phi^\hbar(\xtop_{3_2})
  }
  \\ \nonumber
  &\times
    \delta(a_{1_1} - a_{2_2} - a_{1_2}' + a_{2_3})
    \delta(a_{2_3} - a_{3_2} - a_{2_2}' + a_{3_1})
    \delta(a_{3_1} - a_{1_2} - a_{3_2}' + a_{1_1}) \,.
\end{align}

\begin{table}
  \begin{equation*}
    \begin{array}{c|c@{\hspace{0.5em}}c@{\hspace{0.5em}}c@{\hspace{0.5em}}c@{\hspace{0.5em}}c@{\hspace{0.5em}}c@{\hspace{0.5em}}c@{\hspace{0.5em}}c@{\hspace{0.5em}}c}
    t
    & \Xtrop_{1_1}[t] & \Xtrop_{1_2}[t] & \Xtrop_{1_3}[t]
    & \Xtrop_{2_1}[t] & \Xtrop_{2_2}[t] & \Xtrop_{2_3}[t]
    & \Xtrop_{3_1}[t] & \Xtrop_{3_2}[t] & \Xtrop_{3_3}[t]
    \\ \hline\hline
    1
    & \Xtrop_{1_1} & \Xtrop_{1_2} & \Xtrop_{1_3}
    & \Xtrop_{2_1} & \Xtrop_{2_2} & \Xtrop_{2_3}
    & \Xtrop_{3_1} & \Xtrop_{3_2} & \Xtrop_{3_3}
    \\
    2
    & \Xtrop_{1_1}  \Xtrop_{2_2} & \Xtrop_{1_2} & \Xtrop_{1_3}
    & \Xtrop_{2_1} & \Xtrop^{-1}_{2_2} & \Xtrop_{2_3}
    & \Xtrop_{3_1} & \Xtrop_{2_2}  \Xtrop_{3_2} & \Xtrop_{3_3}
    \\
    3
    & \Xtrop_{1_1}  \Xtrop_{2_2} & \Xtrop_{1_2} & \Xtrop_{1_3}
    & \Xtrop_{2_1} & \Xtrop_{3_2} & \Xtrop_{2_2}  \Xtrop_{2_3}  \Xtrop_{3_2}
    & \Xtrop_{3_1} & \Xtrop^{-1}_{2_2} \Xtrop^{-1}_{3_2} & \Xtrop_{3_3}
    \\
    4
    & \Xtrop_{1_1}  \Xtrop_{1_2}  \Xtrop_{2_2} & \Xtrop^{-1}_{1_2} & \Xtrop_{1_3}
    & \Xtrop_{2_1} & \Xtrop_{3_2} & \Xtrop_{2_2}  \Xtrop_{2_3}  \Xtrop_{3_2}
    & \Xtrop_{1_2}  \Xtrop_{3_1} & \Xtrop^{-1}_{2_2} \Xtrop^{-1}_{3_2} & \Xtrop_{3_3}
    \\
    5
    & \Xtrop_{1_1}  \Xtrop_{1_2}  \Xtrop_{2_2} & \Xtrop^{-1}_{1_2} & \Xtrop_{1_3}
    & \Xtrop_{2_1} & \Xtrop^{-1}_{3_2} & \Xtrop_{2_2}  \Xtrop_{2_3}  \Xtrop_{3_2}
    & \Xtrop_{1_2}  \Xtrop_{3_1}  \Xtrop_{3_2} & \Xtrop^{-1}_{2_2} & \Xtrop_{3_3}
    \\
    6
    & \Xtrop_{1_1}  \Xtrop_{1_2}  \Xtrop_{2_2} & \Xtrop^{-1}_{2_2} & \Xtrop_{1_3}
    & \Xtrop_{2_1} & \Xtrop^{-1}_{3_2} & \Xtrop_{2_2}  \Xtrop_{2_3}  \Xtrop_{3_2}
    & \Xtrop_{1_2}  \Xtrop_{3_1}  \Xtrop_{3_2} & \Xtrop^{-1}_{1_2} & \Xtrop_{3_3}
  \end{array}
  \end{equation*}
    
  \caption{The transformation of the tropical $\CX$-variables under
    the cluster transformation
    $\seed[1] \xrightarrow{\mu_{2_2}} \seed[2] \xrightarrow{\mu_{3_2}}
    \seed[3] \xrightarrow{\mu_{1_2}} \seed[4] \xrightarrow{\mu_{2_2}}
    \seed[5] \xrightarrow{\auto_{1_2,3_2}} \seed[6]$ from the square
    quiver $\seed[1] := \seed_{121}$ to $\seed[6] := \seed_{212}$.}
  \label{tab:y-square}
\end{table}

\section{Cluster transformations and three-dimensional gauge theories}
\label{sec:QFT}

Terashima and Yamazaki \cite{Terashima:2013fg} observed that the
sequence of quantum dilogarithms~\eqref{eq:Phi-identity-tropical} can
be interpreted as the partition function of a three-dimensional
$\CN = 2$ supersymmetric gauge theory formulated on the squashed
three-sphere
\begin{equation}
  S^3_b := \{(z_1,z_2) \in \C^2 \mid b|z_1|^2 + b^{-1} |z_2|^2 = 1\} \,,
\end{equation}
where the squashing parameter $b \in \R_{>0}$ is related to the Planck
constant $\hbar$ as
\begin{equation}
  \hbar = b^2 \,.
\end{equation}

In this section we rewrite the intertwiner
$\K{\ct}\colon \CH_{\seed'} \to \CH_\seed$ for a cluster
transformation $\ct\colon \seed \to \seed'$ as the partition function
of a three-dimensional $\CN = 2$ supersymmetric gauge theory on
$S^3_b$.  This result shows that the R-matrices obtained from trivial
cluster transformations in section~\ref{sec:TE} can be identified with
$S^3_b$ partition functions.

Let
\begin{equation}
  \ct\colon \seed
  =: \seed[1]
  \xrightarrow{\mu_{k[1]}} \seed[2]
  \xrightarrow{\mu_{k[2]}} \seed[3]
  \dotsb
  \xrightarrow{\mu_{k[L]}} \seed[L+1]
  \xrightarrow{\auto}
  \seed[L+2] := \seed'
\end{equation}
be a decomposition of $\ct$ into $L$ mutations and one automorphism.
Following \cite{MR2861174}, for each Hilbert space $\CH_{\seed[t]}$ we
introduce a set of position and momentum operators
$(\ah_i[t], \ph_i[t])_{i \in I}$, a basis of position eigenstates
$\{\ket{a[t]}\}_{a \in \R^I}$ and a basis of momentum eigenstates
$\{\ket{p[t]}\}_{p \in \R^I}$ such that
\begin{alignat}{2}
  \ah_i[t] \ket{a[t]} &= a_i[t] \ket{a[t]} \,,
  \qquad &
  \braket{a[t]}{a'[t]} &= \delta(a[t] - a'[t]) \,,
  \\
  \ph_i[t] \ket{p[t]} &= p_i[t] \ket{p[t]} \,,
  \qquad &
  \braket{p[t]}{p'[t]} &= \delta(p[t] - p'[t])
\end{alignat}
and
\begin{equation}
  \braket{a[t]}{p[t]} = e^{2\pi\iu p[t] \cdot a[t]} \,,
  \quad
  \bra{a[t]} \ph_i[t] \ket{f}
  = -\frac{\iu}{2\pi} \frac{\partial}{\partial a_i[t]} \braket{a[t]}{f} \,,
  \quad
  \ket{f} \in \CH_{\seed[t]} \,.
\end{equation}
The completeness relations read
\begin{equation}
  \int_{\R^I} \rmd a[t] \ket{a[t]} \bra{a[t]}
  = \int_{\R^I} \rmd p[t] \ket{p[t]} \bra{p[t]}
  = 1 \,.
\end{equation}

Inserting the completeness relations for the position eigenstates into
the expression
$\K{\ct} = \K{\mu_{k[1]}} \K{\mu_{k[2]}} \dotsm \K{\mu_{k[L]}}
\K{\auto}$, we can write the matrix elements of $\K{\ct}$ as
\begin{multline}
  \bra{a[1]} \K{\ct} \ket{a[L+2]}
  \\
  =
  \int
  \biggl(
  \prod_{t=1}^L \rmd a[t+1]
  \bra{a[t]} \K{\mu_{k[t]}} \ket{a[t+1]}
  \biggr)
  \bra{a[L+1]} \K{\auto} \ket{a[L+2]} \,.
\end{multline}
The last factor is a product of delta functions:
\begin{equation}
  \bra{a[L+1]} \K{\auto} \ket{a[L+2]}
  = \prod_{i \in I} \delta(a_i[L+1] - a_{\auto(i)}[L+2]) \,.
\end{equation}
Let us calculate the remaining factors.

Using the completeness relations we can write
\begin{align}
    &\bra{a[t]} \K{\mu_{k[t]}} \ket{a[t+1]}
    \\ \nonumber
    &=
    \int \rmd p[t] \rmd \at[t]
    \bra{a[t]}
    \Phi^\hbar(\epsilon[t] \xop_{k[t]}[t])^{\epsilon[t]}
    \ket{p[t]}
    \bra{p[t]}
    \Phi^\hbar(\epsilon[t] \xtop_{k[t]}[t])^{-\epsilon[t]}
    \ket{\at[t]}
    \\ \nonumber
    &\qquad
    \times
    \bra{\at[t]}
    \K{\mu_{k[t]}}^{\flat(\epsilon[t])}
    \ket{a[t+1]}
    \\ \nonumber
    &=
    \int \rmd p[t] \rmd \at[t]
    \frac{
      \Phi^\hbar\bigl(\epsilon[t] (-2\pi^2 b^2 p_{k[t]}[t]
      - \sum_{j \in I} \varepsilon_{k[t]j}[t] a_j[t])\bigr)^{\epsilon[t]}
    }{
      \Phi^\hbar\bigl(\epsilon[t](-2\pi^2 b^2 p_{k[t]}[t]
      + \sum_{j \in I} \varepsilon_{k[t]j}[t] \at_j[t])\bigr)^{\epsilon[t]}
    }
  \\ \nonumber
    &\qquad
      \times
      e^{2\pi\iu p[t] \cdot (a[t] - \at[t])}
      \prod_{i \in I \setminus \{k[t]\}}
      \delta\bigl(\at_i[t] -  a_i[t+1]\bigr)
  \\ \nonumber
    &\qquad
      \times
      \delta\Bigl(-\at_{k[t]}[t]
      + \sum_{j \in I} \bigl[-\epsilon[t]\varepsilon_{k[t]j}\bigr]_+ \at_j[t]
      -  a_{k[t]}[t+1]\Bigr) \,,
\end{align}
where we have chosen a sign $\epsilon[t] \in \{\pm\}$.  (Note that
$\ph_{k[t]}$ and $\sum_{j \in I} \varepsilon_{k[t]j}[t] \ah_j$
commute.)  Integrating over $p_i[t]$ for $i \neq k[t]$ yields
$\delta(a_i[t] - \at_i[t])$.  Then, integrating over $\at[t]$ gives
\begin{multline}
  \bra{a[t]} \K{\mu_{k[t]}} \ket{a[t+1]}
  \\
  =
  \frac{1}{\pi b}
  \int \rmd x[t]
  \biggl(\frac{
    \Phi^\hbar\bigl(2\pi b(x[t] - \epsilon[t] u[t])\bigr)
  }{
    \Phi^\hbar\bigl(2\pi b(x[t] + \epsilon[t] u[t])\bigr)
  }\biggr)^{\epsilon[t]}
  e^{2\pi\iu w[t] x[t]}
  \prod_{i \in I \setminus \{k[t]\}}
  \delta\bigl(a_i[t] - a_i[t+1]\bigr) \,,
\end{multline}
where
\begin{align}
  x[t] &:= -\pi b \epsilon[t] p_{k[t]}[t] \,,
  \\
  u[t] &:= \frac{1}{2\pi b} \sum_{j \in I} \varepsilon_{k[t]j}[t] a_j[t] \,,
  \\
  w[t]
  &:=
  -\frac{\epsilon[t]}{\pi b}
  \biggl(a_{k[t]}[t] + a_{k[t]}[t+1]
    - \sum_{j \in I} \bigl[-\epsilon[t]\varepsilon_{k[t]j}[t]\bigr]_+ a_j[t]
  \biggr) \,.
\end{align}

The integration over $x[t]$ can be performed with the
formula~\cite{Faddeev:2000if}%
\footnote{This formula is valid for $\Im(-u + \iu Q/2) > 0$ and
  $\Im(-2u) < \Im(w) < 0$.  To satisfy these conditions we can give
  $a_j[t]$ an imaginary part
  $\iu\varepsilon_{k[t]j}[t] c[t] + \iu \delta_{kj} d[t]$ with some
  constants $c[t]$, $d[t]$.  In the corresponding three-dimensional
  gauge theory this operation amounts to shifting the R-charges of the
  three kinds of chiral multiplets to positive values that sum to
  $2$.}
\begin{equation}
  \int_{\R} \frac{\Phi^\hbar(2\pi b(x-u))}{\Phi^\hbar(2\pi b(x+u))}
  e^{2\pi\iu wx} \rmd x
  =
  s_b\Bigl(2u - \frac{\iu}{2} Q\Bigr)
  s_b\Bigl(w + \frac{\iu}{2} Q\Bigr)
  s_b\Bigl(-2u - w + \frac{\iu}{2} Q\Bigr)
  \,.
\end{equation}
Here
\begin{equation}
  Q = b + b^{-1}
\end{equation}
and
\begin{equation}
  s_b(z)
  :=
  e^{-\pi\iu z^2/2 + \pi\iu(2 - Q^2)/24}
  \Phi^{b^2}(2\pi b z)^{-1}
  = s_b(-z)^{-1} \,.
\end{equation}
We obtain
\begin{multline}
  \bra{a[t]} \K{\mu_{k[t]}} \ket{a[t+1]}
  \\
  =
  \frac{1}{\pi b}
  Z_{\mu_{k[t]}}\bigl(\sigma_{k[t]}[t], \sigma_{k[t]}[t+1]; (\sigma_i[t])_{i \in I \setminus \{k[t]\}}\bigr)
  \prod_{i \in I \setminus \{k[t]\}} \delta(a_i[t] - a_i[t+1])
  \,,
\end{multline}
with
\begin{equation}
  \sigma_i[t] := \frac{a_i[t]}{\pi b}
\end{equation}
and
\begin{align}
  &
    Z_{\mu_{k[t]}}\bigl(\sigma_{k[t]}[t], \sigma_{k[t]}[t+1]; (\sigma_i[t])_{i \in I \setminus \{k[t]\}}\bigr)
  \\ \nonumber
  &\qquad
    :=
  s_b\biggl(
  \sum_{j \in I} \varepsilon_{k[t]j}[t] \sigma_j[t]
  - \frac{\iu}{2} Q\biggr)
  \\ \nonumber
  &\qquad\qquad
  \times
  s_b\biggl(
  \sigma_{k[t]}[t] + \sigma_{k[t]}[t+1]
  - \sum_{j \in I} \bigl[\varepsilon_{k[t]j}[t]\bigr]_+ \sigma_j[t]
  + \frac{\iu}{2} Q\biggr)
  \\ \nonumber
  &\qquad\qquad
  \times s_b\biggl(
  -\sigma_{k[t]}[t] - \sigma_{k[t]}[t+1]
  + \sum_{j \in I} \bigl[-\varepsilon_{k[t]j}[t]\bigr]_+ \sigma_j[t]
  + \frac{\iu}{2} Q\biggr)
  \,.
\end{align}
This is independent of the choice of the sign $\epsilon[t]$, as
claimed before.

From this calculation we find that the matrix elements of $\K{\ct}$
are given by
\begin{multline}
  \bra{a[1]} \K{\ct} \ket{a[L+2]}
  \\
  =
  \int
  \biggl(
  \prod_{t=1}^L \rmd \sigma_{k[t]}[t+1]
  Z_{\mu_{k[t]}}\bigl(\sigma_{k[t]}[t], \sigma_{k[t]}[t+1]; (\sigma_i[t])_{i \in I \setminus \{k[t]\}}\bigr)
  \biggr)
  \\
  \times
  \prod_{i \in I} \delta(a_i[L+1] - a_{\auto(i)}[L+2])
  \,,
\end{multline}
where the variables are understood to satisfy the relation
\begin{equation}
  \label{eq:sigma_it=sigma_it+1}
  \sigma_i[t] = \sigma_i[t+1] \,,
  \quad
  i \neq k[t] \,,
  \quad
  \text{$t = 1$, $2$, $\dotsc$, $L$} \,.
\end{equation}
Let
\begin{equation}
  K := \bigcup_{t=1}^L \{k[t]\}
\end{equation}
and write the product of delta functions in the above expression as
\begin{multline}
  \prod_{i \in I} \delta(a_i[L+1] - a_{\auto(i)}[L+2])
  \\
  =
  \prod_{k \in K}
  \delta(a_k[L+1] - a_{\auto(k)}[L+2])
  \prod_{l \in I \setminus K}
  \delta(a_l[L+1] - a_{\auto(l)}[L+2]) \,.
\end{multline}
For each $k \in K$, let $t_k$ be the largest $t$ such that $k = k[t]$.
The integration over $\sigma_k[t_k+1]$ sets
\begin{equation}
  \label{eq:sigma_kt+1=sigma_kL+2}
  \sigma_k[t_k+1] = \sigma_{\auto(k)}[L+2] \,,
  \quad
  k \in K \,.
\end{equation}
By the relations \eqref{eq:sigma_it=sigma_it+1} and
\eqref{eq:sigma_kt+1=sigma_kL+2}, each variable $a_i[t]$ is now equal
to either $a_i[1]$, $a_{\auto(i)}[L+2]$ or one of the remaining
integration variables.  Finally, we arrive at the expression
\begin{multline}
  \label{eq:Kct-ME}
  \bra{a[1]} \K{\ct} \ket{a[L+2]}
  \\
  =
  \int
  \biggl(
  \prod_{s \in \{1, 2, \dotsc, L\} \setminus
    \bigcup_{k \in K} \{t_k\}}
  \rmd \sigma_{k[s]}[s+1]
  \prod_{t=1}^L
  Z_{\mu_{k[t]}}\bigl(\sigma_{k[t]}[t], \sigma_{k[t]}[t+1]; (\sigma_i[t])_{i \in I \setminus \{k[t]\}}\bigr)
  \biggr)
  \\
  \times
  (\pi b)^{-|K|}
  \prod_{l \in I \setminus K}
  \delta(a_l[1] - a_{\auto(l)}[L+2])
  \,.
\end{multline}

The first line of the right-hand side of the matrix
element~\eqref{eq:Kct-ME} coincides with the integral formula for the
partition function of a three-dimensional $\CN = 2$ supersymmetric
gauge theory on $S^3_b$ \cite{Hama:2011ea}.  This theory has
\begin{itemize}
\item abelian symmetry groups $\U(1)_{\sigma_i[t]}$, $i \in I$,
  $t \in \{1 ,2, \dotsc, L+2\}$, with the identification
  \eqref{eq:sigma_it=sigma_it+1} and \eqref{eq:sigma_kt+1=sigma_kL+2},
  among which $\U(1)_{\sigma_{k[s]}[s+1]}$,
  $s \in \{1, 2, \dotsc, L\} \setminus \bigcup_{k \in K} \{t_k\}$, are
  gauged;

\item a vector multiplet for each gauge group
  $\U(1)_{\sigma_{k[s]}[s+1]}$;
  
\item a background (i.e.\ nondynamical) vector multiplet for each
  global symmetry group $\U(1)_{\sigma_i[t]}$ whose real scalar
  component is $\sigma_i[t]$;

\item a chiral multiplet for each $t \in \{1,2, \dotsc, L\}$ with
  $\U(1)_{\sigma_j[t]}$-charge $\varepsilon_{k[t]j}[t]$ and R-charge
  $2$;
  
\item a chiral multiplet for each $t \in \{1,2, \dotsc, L\}$ with
  $\U(1)_{\sigma_{k[t]}[t]}$-charge $+1$,
  $\U(1)_{\sigma_{k[t]}[t+1]}$-charge $+1$,
  $\U(1)_{\sigma_j[t]}$-charge $-[\varepsilon_{k[t]j}[t]]_+$ for
  $j \neq k[t]$ and R-charge $0$;

\item a chiral multiplet for each $t \in \{1,2, \dotsc, L\}$ with
  $\U(1)_{\sigma_{k[t]}[t]}$-charge $-1$,
  $\U(1)_{\sigma_{k[t]}[t+1]}$-charge $-1$,
  $\U(1)_{\sigma_j[t]}$-charge $[-\varepsilon_{k[t]j}[t]]_+$ for
  $j \neq k[t]$ and R-charge $0$; and

\item zero Chern--Simons levels and zero Fayet--Iliopoulos parameters.
\end{itemize}
The field content is almost that of an $\CN = 4$ supersymmetric gauge
theory and is compatible with the cubic superpotential required for
such a theory, but the $\U(1)_{\sigma_j[t]}$-charge assignment does
not allow the fields to form $\CN = 4$ supermultiplets.

\appendix

\section{Proof of Proposition \ref{prop:trivial-loop-ct}}
\label{sec:proof}

In this appendix we give a proof of Proposition
\ref{prop:trivial-loop-ct}.  For the ease of presentation we will
employ different naming conventions for the vertices of quivers.

The proof is based on calculations.  For each of the triangle, square
and butterfly quivers assigned to $\seed_{123121}$, we will decompose
the cluster transformation~\eqref{eq:RRRRRRRR} into a sequence of
mutations followed by automorphisms.  Then, for each step of the
decomposition, we will list the tropical $\CX$-variables that are
transformed nontrivially and their relations to the initial tropical
$\CX$-variables $(\Xtrop_i)_{i \in I}$.  We will find that if $\Xtrop_i$ is ever
transformed, then the last entry in which a variable $\Xtrop_i[t]$ of any
$t$ appears (emphasized in bold letters) is always a relation
$\Xtrop_i[t] = \Xtrop_i$.  Therefore, the cluster transformation acts on the
tropical $\CX$-variables trivially, and
Proposition~\ref{prop:trivial-loop-ct} follows from
Theorem~\ref{eq:trivial-trop}.

\subsection{Triangle quiver}

There are 10 connected vertices in the triangle quivers assigned to
reduced words for the longest element of $S_4$.  For the triangle
quivers assigned to 123121, we label the vertices as
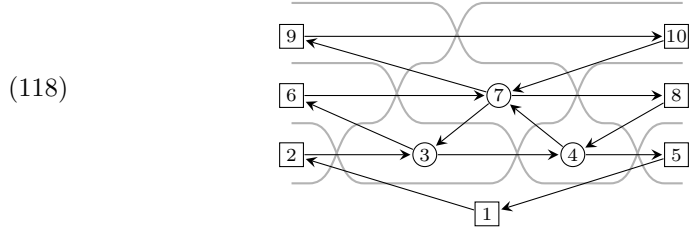
\begin{equation}
  \centering
  \begin{tikzpicture}[scale=0.8, rotate=90]
    \braid[gap=0, thick, black!30] a_1 a_2 a_3 a_1 a_2 a_1;
    
    \node[fnode] (e) at (0.5,-3.3) {1};
    \node[fnode] (1) at (1.5,0) {2};
    \node[fnode] (1c2) at (2.5,0) {6};
    \node[fnode] (1c2c3) at (3.5,0) {9};
    \node[fnode] (4) at (1.5,-6.5) {5};
    \node[fnode] (3c4) at (2.5,-6.5) {8};
    \node[fnode] (2c3c4) at (3.5,-6.5) {10};
    \node[vertex] (1c2c3c4) at (4.5,-3.3) {\phantom{\tikz{\node[gnode]{};}}};
    
    \node[gnode] (2) at (1.5,-2.25) {3};
    \node[gnode] (2c3) at (2.5,-3.5) {7};
    \node[gnode] (3) at (1.5,-4.75) {4}; 
    
    \draw[q->] (e) -- (1);
    \draw[q->] (1) -- (2);
    
    \draw[q->] (2) -- (1c2);
    \draw[q->] (2c3) -- (2);
    \draw[q->] (1c2) -- (2c3);
    
    \draw[q->] (2c3) -- (1c2c3);
    \draw[q->] (2c3c4) -- (2c3);
    \draw[q->] (1c2c3) -- (2c3c4);
    
    \draw[q->] (2) -- (3);
    
    \draw[q->] (3) -- (2c3);
    \draw[q->] (3c4) -- (3);
    \draw[q->] (2c3) -- (3c4);
    
    \draw[q->] (4) -- (e);
    \draw[q->] (3) -- (4);
  \end{tikzpicture}
\end{equation}

The cluster transformation~\eqref{eq:RRRRRRRR} has a decomposition
\begin{multline}
 \beta_{123}^{-1} \circ \beta_{124}^{-1} \circ \beta_{134}^{-1} \circ \beta_{234}^{-1}
  \circ \beta_{123} \circ \beta_{124} \circ \beta_{134} \circ \beta_{234}
  \\
  =
  \auto_{4,7} \circ \auto_{3,7}
  \circ \mu_7 \circ \mu_4 \circ \mu_3 \circ \mu_7
  \circ \mu_4 \circ \mu_3 \circ \mu_7 \circ \mu_4 \,.
\end{multline}
Under this sequence of mutations and automorphisms the tropical
$\CX$-variables transform as follows:
\begin{multicols}{2}
  \begin{itemize}
  \item[$\mu_4$:]
    $\Xtrop_4[2] = 1/\Xtrop_4$ \\
    $\Xtrop_5[2] = \Xtrop_4 \Xtrop_5$ \\
    $\Xtrop_7[2] = \Xtrop_4 \Xtrop_7$
		
  \item[$\mu_7$:]
    $\Xtrop_4[3] = \Xtrop_7$ \\
    $\Xtrop_7[3] = 1/\Xtrop_4 \Xtrop_7$ \\
    $\Xtrop_9[3] = \Xtrop_4 \Xtrop_7 \Xtrop_9$
		
  \item[$\mu_3$:]
    $\Xtrop_3[4] = 1/\Xtrop_3$ \\
    $\Xtrop_5[4] = \Xtrop_3 \Xtrop_4 \Xtrop_5$ \\
    $\Xtrop_6[4] = \Xtrop_3 \Xtrop_6$
	
  \item[$\mu_4$:]
    $\Xtrop_4[5] = 1/\Xtrop_7$ \\
    $\Xtrop_7[5] = 1/\Xtrop_4$ \\
    $\Xtrop_8[5] = \Xtrop_7 \Xtrop_8$

  \item[$\mu_7$:]
    $\Xtrop_3[6] = 1/\Xtrop_3 \Xtrop_4$ \\
    $\Xtrop_7[6] = \Xtrop_4$ \\
    $\Xtrop_9[6] = \Xtrop_7 \Xtrop_9$
		
  \item[$\mu_3$:]
    $\Xtrop_3[7] = \Xtrop_3 \Xtrop_4$ \\
    $\boldsymbol{\Xtrop_5[7] = \Xtrop_5}$ \\
    $\Xtrop_7[7] = 1/\Xtrop_3$
	
  \item[$\mu_4$:]
    $\Xtrop_4[8] = \Xtrop_7$ \\
    $\boldsymbol{\Xtrop_8[8] = \Xtrop_8}$ \\
    $\boldsymbol{\Xtrop_9[8] = \Xtrop_9}$
	
  \item[$\mu_7$:]
    $\Xtrop_3[9] = \Xtrop_4$ \\
    $\boldsymbol{\Xtrop_6[9] = \Xtrop_6}$ \\
    $\Xtrop_7[9] = \Xtrop_3$
		
  \item[$\auto_{3,7}$:]
    $\boldsymbol{\Xtrop_3[10] = \Xtrop_3}$ \\
    $\Xtrop_7[11] = \Xtrop_4$

  \item[$\auto_{4,7}$:]
    $\boldsymbol{\Xtrop_4[11] = \Xtrop_4}$ \\
    $\boldsymbol{\Xtrop_7[11] = \Xtrop_7}$
  \end{itemize}
\end{multicols}

\subsection{Square quiver}	

We label the 16 connected vertices of the square quiver assigned to
the reduced word 123121 for the longest element of $S_4$ as
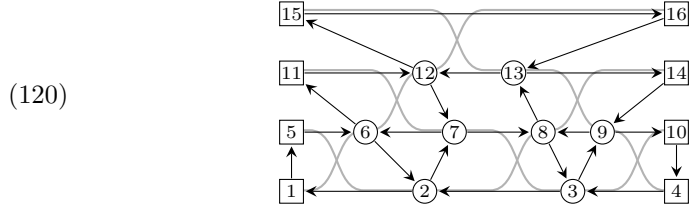
\begin{equation}
  \centering
  \begin{tikzpicture}[scale=0.8, rotate=90]
    \braid[gap=0, thick, black!30] a_1 a_2 a_3 a_1 a_2 a_1;
    
    \node[fnode] (1_1) at (1,0) {1};
    \node[fnode] (2_1) at (2,0) {5};
    \node[fnode] (3_1) at (3,0) {11};
    \node[fnode] (4_1) at (4,0) {15};
    \node[fnode] (4_4) at (1,-6.5) {4};
    \node[fnode] (3_4) at (2,-6.5) {10};
    \node[fnode] (2_4) at (3,-6.5) {14};
    \node[fnode] (1_4) at (4,-6.5) {16};
    
    \node[gnode] (2_2) at (1,-2.25) {2}; 
    \node[gnode] (1_2) at (2,-1.25) {6};
    \node[gnode] (3_2) at (2,-2.75) {7};
    \node[gnode] (1_3) at (3,-2.25) {12};
    \node[gnode] (3_3) at (1,-4.75) {3}; 
    \node[gnode] (2_3) at (2,-4.25) {8};
    \node[gnode] (4_2) at (3,-3.75) {13};
    \node[gnode] (4_3) at (2,-5.25) {9};
    
    \draw[q->] (1_1) -- (2_1);
    \draw[q->] (2_1) -- (1_2);
    \draw[q->] (1_2) -- (2_2);
    \draw[q->] (2_2) -- (1_1);
    
    \draw[q->] (1_2) -- (3_1);
    \draw[q->] (3_1) -- (1_3);
    \draw[q->] (1_3) -- (3_2);
    \draw[q->] (3_2) -- (1_2);
    
    \draw[q->] (1_3) -- (4_1);
    \draw[q->] (4_1) -- (1_4);
    \draw[q->] (1_4) -- (4_2);
    \draw[q->] (4_2) -- (1_3);
    
    \draw[q->] (2_2) -- (3_2);
    \draw[q->] (3_2) -- (2_3);
    \draw[q->] (2_3) -- (3_3);
    \draw[q->] (3_3) -- (2_2);
    
    \draw[q->] (2_3) -- (4_2);
    \draw[q->] (4_2) -- (2_4);
    \draw[q->] (2_4) -- (4_3);
    \draw[q->] (4_3) -- (2_3);
    
    \draw[q->] (3_3) -- (4_3);
    \draw[q->] (4_3) -- (3_4);
    \draw[q->] (3_4) -- (4_4);
    \draw[q->] (4_4) -- (3_3);
  \end{tikzpicture}
\end{equation}

Under the cluster transformation
\begin{align}
  &\beta_{123}^{-1}\circ \beta_{124}^{-1} \circ \beta_{134}^{-1} \circ \beta_{234}^{-1}
  \circ \beta_{123} \circ \beta_{124} \circ \beta_{134} \circ \beta_{234}
  \\ \nonumber
  &\qquad
    =
    \auto_{8,9} \circ \auto_{6,7} \circ \auto_{3,13} \circ \auto_{2,12}
    \circ \mu_{12} \circ \mu_6 \circ \mu_7 \circ \mu_{12}
    \circ \mu_9 \circ \mu_3 \circ \mu_2 \circ \mu_9
  \\ \nonumber
  &\qquad\qquad
    \circ \mu_{13} \circ \mu_8 \circ \mu_6 \circ \mu_{13}
    \circ \mu_7 \circ \mu_2 \circ \mu_{12} \circ \mu_7
    \circ \mu_9 \circ \mu_{13} \circ \mu_3 \circ \mu_9
  \\ \nonumber
  &\qquad\qquad
    \circ \mu_2 \circ \mu_6 \circ \mu_8 \circ \mu_2
    \circ \mu_7 \circ \mu_{12} \circ \mu_{13} \circ \mu_7
    \circ \mu_3 \circ \mu_8 \circ \mu_9 \circ \mu_3
\end{align}
the tropical $\CX$-variables transform as
\begin{multicols}{2}
  \begin{itemize}
  \item[$\mu_{3}$:]
    $\Xtrop_{2}[2] = \Xtrop_{2} \Xtrop_{3}$ \\
    $\Xtrop_{3}[2] = 1/\Xtrop_{3}$ \\
    $\Xtrop_{9}[2] = \Xtrop_{3} \Xtrop_{9}$
		
  \item[$\mu_{9}$:]
    $\Xtrop_{3}[3] = \Xtrop_{9}$ \\
    $\Xtrop_{9}[3] = 1/\Xtrop_{3} \Xtrop_{9}$ \\
    $\Xtrop_{10}[3] = \Xtrop_{3} \Xtrop_{9} \Xtrop_{10}$

  \item[$\mu_{8}$:]
    $\Xtrop_{2}[4] = \Xtrop_{2} \Xtrop_{3} \Xtrop_{8}$ \\
    $\Xtrop_{8}[4] = 1/\Xtrop_{8}$ \\
    $\Xtrop_{13}[4] = \Xtrop_{8} \Xtrop_{13}$

  \item[$\mu_{3}$:]
    $\Xtrop_{3}[5] = 1/\Xtrop_{9}$ \\
    $\Xtrop_{9}[5] = 1/\Xtrop_{3}$ \\
    $\Xtrop_{13}[5] = \Xtrop_{8} \Xtrop_{9} \Xtrop_{13}$

  \item[$\mu_{7}$:]
    $\Xtrop_{6}[6] = \Xtrop_{6} \Xtrop_{7}$ \\
    $\Xtrop_{7}[6] = 1/\Xtrop_{7}$ \\
    $\Xtrop_{13}[6] = \Xtrop_{7} \Xtrop_{8} \Xtrop_{9} \Xtrop_{13}$
		
  \item[$\mu_{13}$:]
    $\Xtrop_{3}[7] = \Xtrop_{7} \Xtrop_{8} \Xtrop_{13}$ \\
    $\Xtrop_{7}[7] = \Xtrop_{8} \Xtrop_{9} \Xtrop_{13}$ \\
    $\Xtrop_{13}[7] = 1/\Xtrop_{7} \Xtrop_{8} \Xtrop_{9} \Xtrop_{13}$
		
  \item[$\mu_{12}$:]
    $\Xtrop_{6}[8] = \Xtrop_{6} \Xtrop_{7} \Xtrop_{12}$ \\
    $\Xtrop_{12}[8] = 1/\Xtrop_{12}$ \\
    $\Xtrop_{15}[8] = \Xtrop_{12} \Xtrop_{15}$
		
  \item[$\mu_{7}$:]
    $\Xtrop_{7}[9] = 1/\Xtrop_{8} \Xtrop_{9} \Xtrop_{13}$ \\
    $\Xtrop_{13}[9] = 1/\Xtrop_{7}$ \\
    $\Xtrop_{15}[9] = \Xtrop_{8} \Xtrop_{9} \Xtrop_{12} \Xtrop_{13} \Xtrop_{15}$

  \item[$\mu_{2}$:]
    $\Xtrop_{1}[10] = \Xtrop_{1} \Xtrop_{2} \Xtrop_{3} \Xtrop_{8}$ \\
    $\Xtrop_{2}[10] = 1/\Xtrop_{2} \Xtrop_{3} \Xtrop_{8}$ \\
    $\Xtrop_{8}[10] = \Xtrop_{2} \Xtrop_{3}$
		
  \item[$\mu_{8}$:]
    $\Xtrop_{2}[11] = 1/\Xtrop_{8}$ \\
    $\Xtrop_{8}[11] = 1/\Xtrop_{2} \Xtrop_{3}$ \\
    $\Xtrop_{9}[11] = \Xtrop_{2}$
		
  \item[$\mu_{6}$:]
    $\Xtrop_{1}[12] = \Xtrop_{1} \Xtrop_{2} \Xtrop_{3} \Xtrop_{6} \Xtrop_{7} \Xtrop_{8} \Xtrop_{12}$ \\
    $\Xtrop_{6}[12] = 1/\Xtrop_{6} \Xtrop_{7} \Xtrop_{12}$ \\
    $\Xtrop_{12}[12] = \Xtrop_{6} \Xtrop_{7}$

  \item[$\mu_{2}$:]
    $\Xtrop_{2}[13] = \Xtrop_{8}$ \\
    $\Xtrop_{6}[13] = 1/\Xtrop_{6} \Xtrop_{7} \Xtrop_{8} \Xtrop_{12}$ \\
    $\Xtrop_{13}[13] = 1/\Xtrop_{7} \Xtrop_{8}$
		
  \item[$\mu_{9}$:]
    $\Xtrop_{3}[14] = \Xtrop_{2} \Xtrop_{7} \Xtrop_{8} \Xtrop_{13}$ \\
    $\Xtrop_{8}[14] = 1/\Xtrop_{3}$ \\
    $\Xtrop_{9}[14] = 1/\Xtrop_{2}$
		
  \item[$\mu_{3}$:]
    $\Xtrop_{3}[15] = 1/\Xtrop_{2} \Xtrop_{7} \Xtrop_{8} \Xtrop_{13}$ \\
    $\Xtrop_{9}[15] = \Xtrop_{7} \Xtrop_{8} \Xtrop_{13}$ \\
    $\Xtrop_{14}[15] = \Xtrop_{2} \Xtrop_{7} \Xtrop_{8} \Xtrop_{13} \Xtrop_{14}$
		
  \item[$\mu_{13}$:]
    $\Xtrop_{2}[16] = 1/\Xtrop_{7}$ \\
    $\Xtrop_{9}[16] = \Xtrop_{13}$ \\
    $\Xtrop_{13}[16] = \Xtrop_{7} \Xtrop_{8}$
		
  \item[$\mu_{9}$:]
    $\Xtrop_{3}[17] = 1/\Xtrop_{2} \Xtrop_{7} \Xtrop_{8}$ \\
    $\Xtrop_{7}[17] = 1/\Xtrop_{8} \Xtrop_{9}$ \\
    $\Xtrop_{9}[17] = 1/\Xtrop_{13}$
		
  \item[$\mu_{7}$:]
    $\Xtrop_{2}[18] = 1/\Xtrop_{7} \Xtrop_{8} \Xtrop_{9}$ \\
    $\Xtrop_{7}[18] = \Xtrop_{8} \Xtrop_{9}$ \\
    $\Xtrop_{15}[18] = \Xtrop_{12} \Xtrop_{13} \Xtrop_{15}$
		
  \item[$\mu_{12}$:]
    $\Xtrop_{7}[19] = \Xtrop_{6} \Xtrop_{7} \Xtrop_{8} \Xtrop_{9}$ \\
    $\Xtrop_{11}[19] = \Xtrop_{6} \Xtrop_{7} \Xtrop_{11}$ \\
    $\Xtrop_{12}[19] = 1/\Xtrop_{6} \Xtrop_{7}$
		
  \item[$\mu_{2}$:]
    $\Xtrop_{2}[20] = \Xtrop_{7} \Xtrop_{8} \Xtrop_{9}$ \\
    $\Xtrop_{7}[20] = \Xtrop_{6}$ \\
    $\Xtrop_{13}[20] = 1/\Xtrop_{9}$
		
  \item[$\mu_{7}$:]
    $\Xtrop_{6}[21] = 1/\Xtrop_{7} \Xtrop_{8} \Xtrop_{12}$ \\
    $\Xtrop_{7}[21] = 1/\Xtrop_{6}$ \\
    $\Xtrop_{12}[21] = 1/\Xtrop_{7}$

  \item[$\mu_{13}$:]
    $\Xtrop_{2}[22] = \Xtrop_{7} \Xtrop_{8}$ \\
    $\Xtrop_{8}[22] = 1/\Xtrop_{3} \Xtrop_{9}$ \\
    $\Xtrop_{13}[22] = \Xtrop_{9}$
		
  \item[$\mu_{6}$:]
    $\Xtrop_{1}[23] = \Xtrop_{1} \Xtrop_{2} \Xtrop_{3} \Xtrop_{6}$ \\
    $\Xtrop_{2}[23] = 1/\Xtrop_{12}$ \\
    $\Xtrop_{6}[23] = \Xtrop_{7} \Xtrop_{8} \Xtrop_{12}$

  \item[$\mu_{8}$:]
    $\Xtrop_{8}[24] = \Xtrop_{3} \Xtrop_{9}$ \\
    $\boldsymbol{\Xtrop_{10}[24] = \Xtrop_{10}}$ \\
    $\Xtrop_{13}[24] = 1/\Xtrop_{3}$
		
  \item[$\mu_{13}$:]
    $\Xtrop_{1}[25] = \Xtrop_{1} \Xtrop_{2} \Xtrop_{6}$ \\
    $\Xtrop_{8}[25] = \Xtrop_{9}$ \\
    $\Xtrop_{13}[25] = \Xtrop_{3}$
		
  \item[$\mu_{9}$:]
    $\Xtrop_{3}[26] = 1/\Xtrop_{2} \Xtrop_{7} \Xtrop_{8} \Xtrop_{13}$ \\
    $\Xtrop_{9}[26] = \Xtrop_{13}$ \\
    $\Xtrop_{15}[26] = \Xtrop_{12} \Xtrop_{15}$
		
  \item[$\mu_{2}$:]
    $\Xtrop_{2}[27] = \Xtrop_{12}$ \\
    $\Xtrop_{6}[27] = \Xtrop_{7} \Xtrop_{8}$ \\
    $\boldsymbol{\Xtrop_{15}[27] = \Xtrop_{15}}$
		
  \item[$\mu_{3}$:]
    $\Xtrop_{3}[28] = \Xtrop_{2} \Xtrop_{7} \Xtrop_{8} \Xtrop_{13}$ \\
    $\Xtrop_{9}[28] = 1/\Xtrop_{2} \Xtrop_{7} \Xtrop_{8}$ \\
    $\boldsymbol{\Xtrop_{14}[28] = \Xtrop_{14}}$

  \item[$\mu_{9}$:]
    $\Xtrop_{3}[29] = \Xtrop_{13}$ \\
    $\Xtrop_{6}[29] = 1/\Xtrop_{2}$ \\
    $\Xtrop_{9}[29] = \Xtrop_{2} \Xtrop_{7} \Xtrop_{8}$
		
  \item[$\mu_{12}$:]
    $\Xtrop_{6}[30] = 1/\Xtrop_{2} \Xtrop_{7}$ \\
    $\Xtrop_{11}[30] = \Xtrop_{6} \Xtrop_{11}$ \\
    $\Xtrop_{12}[30] = \Xtrop_{7}$
		
  \item[$\mu_{7}$:]
    $\Xtrop_{1}[31] = \Xtrop_{1} \Xtrop_{2}$ \\
    $\Xtrop_{7}[31] = \Xtrop_{6}$ \\
    $\boldsymbol{\Xtrop_{11}[31] = \Xtrop_{11}}$
		
  \item[$\mu_{6}$:]
    $\Xtrop_{6}[32] = \Xtrop_{2} \Xtrop_{7}$ \\
    $\Xtrop_{9}[32] = \Xtrop_{8}$ \\
    $\Xtrop_{12}[32] = 1/\Xtrop_{2}$

  \item[$\mu_{12}$:]
    $\boldsymbol{\Xtrop_{1}[33] = \Xtrop_{1}}$ \\
    $\Xtrop_{6}[33] = \Xtrop_{7}$ \\
    $\Xtrop_{12}[33] = \Xtrop_{2}$
		
  \item[$\auto_{2,12}$:]
    $\boldsymbol{\Xtrop_{2}[34] = \Xtrop_{2}}$ \\
    $\boldsymbol{\Xtrop_{12}[34] = \Xtrop_{12}}$

  \item[$\auto_{3,13}$:]
    $\boldsymbol{\Xtrop_{3}[35] = \Xtrop_{3}}$ \\
    $\boldsymbol{\Xtrop_{13}[35] = \Xtrop_{13}}$

  \item[$\auto_{6,7}$:]
    $\boldsymbol{\Xtrop_{6}[36] = \Xtrop_{6}}$ \\
    $\boldsymbol{\Xtrop_{7}[36] = \Xtrop_{7}}$

  \item[$\auto_{8,9}$:]
    $\boldsymbol{\Xtrop_{8}[37] = \Xtrop_{8}}$ \\
    $\boldsymbol{\Xtrop_{9}[37] = \Xtrop_{9}}$
  \end{itemize}	
\end{multicols}

\subsection{Butterfly quiver}

For the butterfly quiver assigned to $123121$, we label the vertices
as
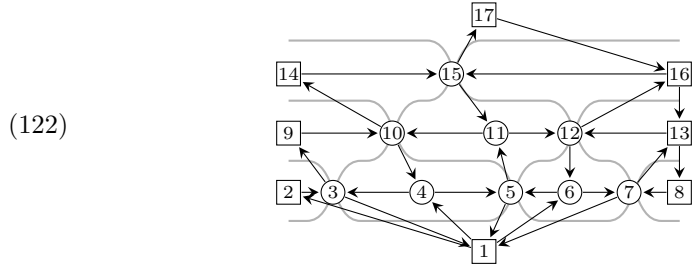
\begin{equation}
    \begin{tikzpicture}[scale=0.8, rotate=90]
    \braid[gap=0, thick, black!30] a_1 a_2 a_3 a_1 a_2 a_1;

    \node[fnode] (e) at (0.5,-3.3) {$1$};
    \node[fnode] (1) at (1.5,0) {$2$};
    \node[fnode] (1c2) at (2.5,0) {$9$};
    \node[fnode] (1c2c3) at (3.5,0) {$14$};
    \node[fnode] (4) at (1.5,-6.6) {$8$};
    \node[fnode] (3c4) at (2.5,-6.6) {$13$};
    \node[fnode] (2c3c4) at (3.5,-6.6) {$16$};
    \node[fnode] (1c2c3c4) at (4.5,-3.3) {$17$};

    \node[gnode] (2) at (1.5,-2.25) {$4$};
    \node[gnode] (2c3) at (2.5,-3.5) {$11$};
    \node[gnode] (3) at (1.5,-4.75) {$6$}; 
    
    \node[gnode] (12) at (1.5,-0.75) {$3$};
    \node[gnode] (13) at (2.5,-1.75) {$10$};
    \node[gnode] (14) at (3.5,-2.75) {$15$};
    \node[gnode] (23) at (1.5,-3.75) {$5$};
    \node[gnode] (24) at (2.5,-4.75) {$12$};
    \node[gnode] (34) at (1.5,-5.75) {$7$};
    
    \draw[q->] (12) -- (e);
    \draw[q->] (e) -- (1);
    \draw[q->] (1) -- (12);
    \draw[q->] (12) -- (1c2);
    \draw[q->] (2) -- (12);

    \draw[q->] (13) -- (2);
    \draw[q->] (1c2) -- (13);
    \draw[q->] (13) -- (1c2c3);
    \draw[q->] (2c3) -- (13);

    \draw[q->] (14) -- (2c3);
    \draw[q->] (1c2c3) -- (14);
    \draw[q->] (14) -- (1c2c3c4);
    \draw[q->] (1c2c3c4) -- (2c3c4);
    \draw[q->] (2c3c4) -- (14);

    \draw[q->] (23) -- (e);
    \draw[q->] (e) -- (2);
    \draw[q->] (2) -- (23);
    \draw[q->] (23) -- (2c3);
    \draw[q->] (3) -- (23);

    \draw[q->] (24) -- (3);
    \draw[q->] (2c3) -- (24);
    \draw[q->] (24) -- (2c3c4);
    \draw[q->] (2c3c4) -- (3c4);
    \draw[q->] (3c4) -- (24);

    \draw[q->] (34) -- (e);
    \draw[q->] (e) -- (3);
    \draw[q->] (3) -- (34);
    \draw[q->] (34) -- (3c4);
    \draw[q->] (3c4) -- (4);
    \draw[q->] (4) -- (34);
  \end{tikzpicture}
\end{equation}

Under the cluster transformation
\begin{align}
  &\beta_{123}^{-1} \circ \beta_{124}^{-1} \circ \beta_{134}^{-1} \circ \beta_{234}^{-1}
  \circ \beta_{123} \circ \beta_{124} \circ \beta_{134} \circ \beta_{234}
  \\ \nonumber
  &\qquad
    =
    \auto_{7,12} \circ \auto_{5,15} \auto_{4,6} \circ \auto_{3,10}
    \circ \mu_{6} \circ \mu_{10} \circ \mu_{15} \circ \mu_{3} 
    \circ \mu_{11} \circ \mu_{15} \circ \mu_{7} \circ \mu_{5} 
  \\ \nonumber
  &\qquad\qquad
    \circ \mu_{4} \circ \mu_{6} \circ \mu_{12} \circ \mu_{11}
    \circ \mu_{3} \circ \mu_{10} \circ \mu_{6} \circ \mu_{15}
    \circ \mu_{5} \circ \mu_{7} \circ \mu_{6} \circ \mu_{12}
  \\ \nonumber
  &\qquad\qquad
    \circ \mu_{11} \circ \mu_{6} \circ \mu_{3} \circ \mu_{4}   
    \circ \mu_{15} \circ \mu_{5} \circ \mu_{10} \circ \mu_{11}
    \circ \mu_{12} \circ \mu_{7} \circ \mu_{5} \circ \mu_{6}
\end{align}
the tropical $\CX$-variables transform as
\begin{multicols}{2}
\begin{itemize}
\item[$\mu_{6}$:]
  $\Xtrop_{5}[2] = \Xtrop_{5} \Xtrop_{6}$ \\
  $\Xtrop_{6}[2] = 1/\Xtrop_{6}$ \\
  $\Xtrop_{7}[2] = \Xtrop_{6} \Xtrop_{7}$

\item[$\mu_{5}$:]
  $\Xtrop_{5}[3] = 1/\Xtrop_{5} \Xtrop_{6}$ \\
  $\Xtrop_{6}[3] = \Xtrop_{5}$ \\
  $\Xtrop_{11}[3] = \Xtrop_{5} \Xtrop_{6} \Xtrop_{11}$

\item[$\mu_{7}$:]
  $\Xtrop_{6}[4] = \Xtrop_{5} \Xtrop_{6} \Xtrop_{7}$ \\
  $\Xtrop_{7}[4] = 1/\Xtrop_{6} \Xtrop_{7}$ \\
  $\Xtrop_{13}[4] = \Xtrop_{6} \Xtrop_{7} \Xtrop_{13}$

\item[$\mu_{12}$:]
  $\Xtrop_{6}[5] = \Xtrop_{5} \Xtrop_{6} \Xtrop_{7} \Xtrop_{12}$ \\
  $\Xtrop_{12}[5] = 1/\Xtrop_{12}$ \\
  $\Xtrop_{16}[5] = \Xtrop_{12} \Xtrop_{16}$
 
\item[$\mu_{11}$:]
  $\Xtrop_{5}[6] = \Xtrop_{11}$ \\
  $\Xtrop_{10}[6] = \Xtrop_{5} \Xtrop_{6} \Xtrop_{10} \Xtrop_{11}$ \\
  $\Xtrop_{11}[6] = 1/\Xtrop_{5} \Xtrop_{6} \Xtrop_{11}$

\item[$\mu_{10}$:]
  $\Xtrop_{10}[7] = 1/\Xtrop_{5} \Xtrop_{6} \Xtrop_{10} \Xtrop_{11}$ \\
  $\Xtrop_{11}[7] = \Xtrop_{10}$ \\
  $\Xtrop_{14}[7] = \Xtrop_{5} \Xtrop_{6} \Xtrop_{10} \Xtrop_{11} \Xtrop_{14}$

\item[$\mu_{5}$:]
  $\Xtrop_{5}[8] = 1/\Xtrop_{11}$ \\
  $\Xtrop_{11}[8] = \Xtrop_{10} \Xtrop_{11}$ \\
  $\Xtrop_{16}[8] = \Xtrop_{11} \Xtrop_{12} \Xtrop_{16}$

\item[$\mu_{15}$:]
  $\Xtrop_{11}[9] = \Xtrop_{10} \Xtrop_{11} \Xtrop_{15}$ \\
  $\Xtrop_{15}[9] = 1/\Xtrop_{15}$ \\
  $\Xtrop_{17}[9] = \Xtrop_{15} \Xtrop_{17}$

\item[$\mu_{4}$:]
  $\Xtrop_{3}[10] = \Xtrop_{3} \Xtrop_{4}$ \\
  $\Xtrop_{4}[10] = 1/\Xtrop_{4}$ \\
  $\Xtrop_{6}[10] = \Xtrop_{4} \Xtrop_{5} \Xtrop_{6} \Xtrop_{7} \Xtrop_{12}$

\item[$\mu_{3}$:]
  $\Xtrop_{3}[11] = 1/\Xtrop_{3} \Xtrop_{4}$ \\
  $\Xtrop_{4}[11] = \Xtrop_{3}$ \\
  $\Xtrop_{9}[11] = \Xtrop_{3} \Xtrop_{4} \Xtrop_{9}$

\item[$\mu_{6}$:]
  $\Xtrop_{4}[12] = \Xtrop_{3} \Xtrop_{4} \Xtrop_{5} \Xtrop_{6} \Xtrop_{7} \Xtrop_{12}$ \\
  $\Xtrop_{6}[12] = 1/\Xtrop_{4} \Xtrop_{5} \Xtrop_{6} \Xtrop_{7} \Xtrop_{12}$ \\
  $\Xtrop_{12}[12] = \Xtrop_{4} \Xtrop_{5} \Xtrop_{6} \Xtrop_{7}$

\item[$\mu_{11}$:]
  $\Xtrop_{4}[13] = \Xtrop_{3} \Xtrop_{4} \Xtrop_{5} \Xtrop_{6} \Xtrop_{7} \Xtrop_{10} \Xtrop_{11} \Xtrop_{12} \Xtrop_{15}$ \\
  $\Xtrop_{11}[13] = 1/\Xtrop_{10} \Xtrop_{11} \Xtrop_{15}$ \\
  $\Xtrop_{15}[13] = \Xtrop_{10} \Xtrop_{11}$

\item[$\mu_{12}$:]
  $\Xtrop_{6}[14] = 1/\Xtrop_{12}$ \\
  $\Xtrop_{7}[14] = \Xtrop_{4} \Xtrop_{5}$ \\
  $\Xtrop_{12}[14] = 1/\Xtrop_{4} \Xtrop_{5} \Xtrop_{6} \Xtrop_{7}$

\item[$\mu_{6}$:]
  $\Xtrop_{5}[15] = 1/\Xtrop_{11} \Xtrop_{12}$ \\
  $\Xtrop_{6}[15] = \Xtrop_{12}$ \\
  $\Xtrop_{11}[15] = 1/\Xtrop_{10} \Xtrop_{11} \Xtrop_{12} \Xtrop_{15}$

\item[$\mu_{7}$:]
  $\Xtrop_{7}[16] = 1/\Xtrop_{4} \Xtrop_{5}$ \\
  $\Xtrop_{12}[16] = 1/\Xtrop_{6} \Xtrop_{7}$ \\
  $\Xtrop_{16}[16] = \Xtrop_{4} \Xtrop_{5} \Xtrop_{11} \Xtrop_{12} \Xtrop_{16}$

\item[$\mu_{5}$:]
  $\Xtrop_{5}[17] = \Xtrop_{11} \Xtrop_{12}$ \\
  $\Xtrop_{6}[17] = 1/\Xtrop_{11}$ \\
  $\Xtrop_{7}[17] = 1/\Xtrop_{4} \Xtrop_{5} \Xtrop_{11} \Xtrop_{12}$

\item[$\mu_{15}$:]
  $\Xtrop_{6}[18] = \Xtrop_{10}$ \\
  $\Xtrop_{10}[18] = 1/\Xtrop_{5} \Xtrop_{6}$ \\
  $\Xtrop_{15}[18] = 1/\Xtrop_{10} \Xtrop_{11}$

\item[$\mu_{6}$:]
  $\Xtrop_{6}[19] = 1/\Xtrop_{10}$ \\
  $\Xtrop_{11}[19] = 1/\Xtrop_{11} \Xtrop_{12} \Xtrop_{15}$ \\
  $\Xtrop_{15}[19] = 1/\Xtrop_{11}$

\item[$\mu_{10}$:]
  $\Xtrop_{3}[20] = 1/\Xtrop_{3} \Xtrop_{4} \Xtrop_{5} \Xtrop_{6}$ \\
  $\Xtrop_{10}[20] = \Xtrop_{5} \Xtrop_{6}$ \\
  $\Xtrop_{14}[20] = \Xtrop_{10} \Xtrop_{11} \Xtrop_{14}$

\item[$\mu_{3}$:]
  $\Xtrop_{3}[21] = \Xtrop_{3} \Xtrop_{4} \Xtrop_{5} \Xtrop_{6}$ \\
  $\Xtrop_{6}[21] = 1/\Xtrop_{3} \Xtrop_{4} \Xtrop_{5} \Xtrop_{6} \Xtrop_{10}$ \\
  $\Xtrop_{10}[21] = 1/\Xtrop_{3} \Xtrop_{4}$

\item[$\mu_{11}$:]
  $\Xtrop_{4}[22] = \Xtrop_{3} \Xtrop_{4} \Xtrop_{5} \Xtrop_{6} \Xtrop_{7} \Xtrop_{10}$ \\
  $\Xtrop_{5}[22] = 1/\Xtrop_{15}$ \\
  $\Xtrop_{11}[22] = \Xtrop_{11} \Xtrop_{12} \Xtrop_{15}$

\item[$\mu_{12}$:]
  $\Xtrop_{4}[23] = \Xtrop_{3} \Xtrop_{4} \Xtrop_{5} \Xtrop_{10}$ \\
  $\Xtrop_{12}[23] = \Xtrop_{6} \Xtrop_{7}$ \\
  $\boldsymbol{\Xtrop_{13}[23] = \Xtrop_{13}}$

\item[$\mu_{6}$:]
  $\Xtrop_{3}[24] = 1/\Xtrop_{10}$ \\
  $\Xtrop_{4}[24] = 1/\Xtrop_{6}$ \\
  $\Xtrop_{6}[24] = \Xtrop_{3} \Xtrop_{4} \Xtrop_{5} \Xtrop_{6} \Xtrop_{10}$

\item[$\mu_{4}$:]
  $\Xtrop_{4}[25] = \Xtrop_{6}$ \\
  $\Xtrop_{6}[25] = \Xtrop_{3} \Xtrop_{4} \Xtrop_{5} \Xtrop_{10}$ \\
  $\Xtrop_{12}[25] = \Xtrop_{7}$

\item[$\mu_{5}$:]
  $\Xtrop_{5}[26] = \Xtrop_{15}$ \\
  $\Xtrop_{11}[26] = \Xtrop_{11} \Xtrop_{12}$ \\
  $\boldsymbol{\Xtrop_{17}[26] = \Xtrop_{17}}$

\item[$\mu_{7}$:]
  $\Xtrop_{7}[27] = \Xtrop_{4} \Xtrop_{5} \Xtrop_{11} \Xtrop_{12}$ \\
  $\Xtrop_{11}[27] = 1/\Xtrop_{4} \Xtrop_{5}$ \\
  $\boldsymbol{\Xtrop_{16}[27] = \Xtrop_{16}}$

\item[$\mu_{15}$:]
  $\Xtrop_{11}[28] = 1/\Xtrop_{4} \Xtrop_{5} \Xtrop_{11}$ \\
  $\Xtrop_{14}[28] = \Xtrop_{10} \Xtrop_{14}$ \\
  $\Xtrop_{15}[28] = \Xtrop_{11}$

\item[$\mu_{11}$:]
  $\Xtrop_{7}[29] = \Xtrop_{12}$ \\
  $\Xtrop_{11}[29] = \Xtrop_{4} \Xtrop_{5} \Xtrop_{11}$ \\
  $\Xtrop_{15}[29] = 1/\Xtrop_{4} \Xtrop_{5}$

\item[$\mu_{3}$:]
  $\Xtrop_{3}[30] = \Xtrop_{10}$ \\
  $\Xtrop_{6}[30] = \Xtrop_{3} \Xtrop_{4} \Xtrop_{5}$ \\
  $\boldsymbol{\Xtrop_{14}[30] = \Xtrop_{14}}$

\item[$\mu_{15}$:]
  $\Xtrop_{6}[31] = \Xtrop_{3}$ \\
  $\boldsymbol{\Xtrop_{11}[31] = \Xtrop_{11}}$ \\
  $\Xtrop_{15}[31] = \Xtrop_{4} \Xtrop_{5}$

\item[$\mu_{10}$:]
  $\Xtrop_{6}[32] = 1/\Xtrop_{4}$ \\
  $\boldsymbol{\Xtrop_{9}[32] = \Xtrop_{9}}$ \\
  $\Xtrop_{10}[32] = \Xtrop_{3} \Xtrop_{4}$

\item[$\mu_{6}$:]
  $\Xtrop_{6}[33] = \Xtrop_{4}$ \\
  $\Xtrop_{10}[33] = \Xtrop_{3}$ \\
  $\Xtrop_{15}[33] = \Xtrop_{5}$

\item[$\auto_{3,10}$:]
  $\boldsymbol{\Xtrop_{3}[34] = \Xtrop_{3}}$ \\
  $\boldsymbol{\Xtrop_{10}[34] = \Xtrop_{10}}$

\item[$\auto_{4,6}$:]
  $\boldsymbol{\Xtrop_{4}[35] = \Xtrop_{4}}$ \\
  $\boldsymbol{\Xtrop_{6}[35] = \Xtrop_{6}}$

\item[$\auto_{5,15}$:]
  $\boldsymbol{\Xtrop_{5}[36] = \Xtrop_{5}}$ \\
  $\boldsymbol{\Xtrop_{15}[36] = \Xtrop_{15}}$

\item[$\auto_{7,12}$:]
  $\boldsymbol{\Xtrop_{7}[37] = \Xtrop_{7}}$ \\
  $\boldsymbol{\Xtrop_{12}[37] = \Xtrop_{12}}$
\end{itemize}
\end{multicols}

\section*{Acknowledgments}

We are grateful to Hyun Kyu Kim and Mauricio Romo for illuminating
discussions and to Dylan Allegretti for helpful comments.

\bibliographystyle{alpha}
\newcommand{\etalchar}[1]{$^{#1}$}

\end{document}